\documentclass[11pt]{article}
\usepackage{amssymb, amsmath}
\usepackage[margin=1in]{geometry}
\usepackage{amsmath}
\usepackage{amsthm}
\usepackage{mathtools}
\usepackage{mathrsfs}
\usepackage{amssymb}
\usepackage{verbatim}
\usepackage{algorithmic}
\usepackage[linesnumbered,ruled,vlined]{algorithm2e}
\usepackage{color}
\usepackage[dvipsnames]{xcolor}
\usepackage{mdframed}
\usepackage{marginnote}
\usepackage{amsfonts}
\usepackage{lineno}

\newcommand\mydef{\mathrel{\overset{\makebox[0pt]{\mbox{\normalfont\tiny\sffamily def}}}{=}}}
\providecommand{\keywords}[1]
{
  \small	
  \textbf{\textit{Keywords---}} #1
}

\newcommand{\sos}{\text{\sc{SoS}}}

\newcommand{\F}{\mathcal{F}}

\newcommand{\Cc}{\mathcal{C}}

\newcommand{\GIdeal}[1]{\left\langle #1 \right\rangle}

\newcommand{\CSP}{\textsc{CSP}}
\newcommand{\IMP}{\textsc{IMP}}
\newcommand{\Pol}{\textsf{Pol}}

\newcommand{\PC}{\textsc{PC}}
\newcommand{\Nsatz}{\textsc{Nsatz}}

\newcommand{\Max}{\textsf{Max}}
\newcommand{\Min}{\textsf{Min}}

\newcommand{\lex}{\textsf{lex }}
\newcommand{\lexns}{\textsf{lex}}
\newcommand{\grlex}{\textsf{grlex }}
\newcommand{\grlexns}{\textsf{grlex}}

\newcommand{\Variety}[1]{{\textbf{V}}\left( #1 \right)}

\newcommand{\I}{\emph{$\rm{I}$}}
\newcommand{\Zz}{\mathbb{Z}}
\newcommand{\N}{\mathbb{N}}

\newcommand{\multideg}{\textnormal{multideg}}
\newcommand{\LM}{\textnormal{LM}}
\newcommand{\LT}{\textnormal{LT}}

\newcommand{\LCM}{\textnormal{lcm}}
\newcommand{\GB}{\text{Gr\"{o}bner} }

\newcommand{\Field}{\mathbb{Q}}
\newcommand{\Do}{\mathcal{D}}

\newcommand{\T}{\mathcal{T}}

\newcommand\p{\mathcal{P}}
\newcommand\q{\mathcal{Q}}

\newcommand\uval{{(|D|-1)}}

\newcommand\db{2(d+k-1)}

\usepackage{lineno}
\newcommand*\patchAmsMathEnvironmentForLineno[1]{%
  \expandafter\let\csname old#1\expandafter\endcsname\csname #1\endcsname
  \expandafter\let\csname oldend#1\expandafter\endcsname\csname end#1\endcsname
  \renewenvironment{#1}%
     {\linenomath\csname old#1\endcsname}%
     {\csname oldend#1\endcsname\endlinenomath}}%
\newcommand*\patchBothAmsMathEnvironmentsForLineno[1]{%
  \patchAmsMathEnvironmentForLineno{#1}%
  \patchAmsMathEnvironmentForLineno{#1*}}%
\AtBeginDocument{%
\patchBothAmsMathEnvironmentsForLineno{equation}%
\patchBothAmsMathEnvironmentsForLineno{align}%
\patchBothAmsMathEnvironmentsForLineno{flalign}%
\patchBothAmsMathEnvironmentsForLineno{alignat}%
\patchBothAmsMathEnvironmentsForLineno{gather}%
\patchBothAmsMathEnvironmentsForLineno{multline}%
}
\usepackage[shortlabels]{enumitem}
\usepackage{url}
\usepackage[pdftex, colorlinks,
  linkcolor=blue,
  citecolor=blue,
  filecolor=blue,urlcolor=blue]
{hyperref}
\usepackage[capitalise]{cleveref}
\newtheorem{theorem}{Theorem}[section]
\newtheorem{lemma}[theorem]{Lemma}
\newtheorem{definition}[theorem]{Definition}
\newtheorem{proposition}[theorem]{Proposition}

\newtheorem{corollary}[theorem]{Corollary}

\newtheorem{example}[theorem]{Example}
\newtheorem{problem}[theorem]{Problem}

\newtheorem{remark}[theorem]{Remark}
\newcommand{\ignore}[1]{}

\newcommand{\SOSe}{\rm{SOS}_\varepsilon\text{-complete}}
\usepackage[normalem]{ulem}

\title{On the Degree Automatability of Sum-of-Squares Proofs} 
\date{}
\author{Alex Bortolotti \thanks{University of Applied Sciences and Arts of Southern Switzerland, IDSIA, Lugano, Switzerland. E-mail: \href{mailto:alex.bortolotti@supsi.ch}{\texttt{alex.bortolotti@supsi.ch}}.} \  \and Monaldo Mastrolilli \thanks{University of Applied Sciences and Arts of Southern Switzerland, IDSIA, Lugano, Switzerland. E-mail: \href{mailto:monaldo.mastrolilli@supsi.ch}{\texttt{monaldo.mastrolilli@supsi.ch}}.} \ \and Luis Felipe Vargas\thanks{University of Applied Sciences and Arts of Southern Switzerland, IDSIA, Lugano, Switzerland. E-mail: \href{mailto:luis.vargas@supsi.ch}{\texttt{luis.vargas@supsi.ch}}.}}

\begin{document}

\maketitle

\begin{abstract}
The Sum-of-Squares (\sos) hierarchy, also known as Lasserre hierarchy, has emerged as a promising tool in optimization.
However, it remains unclear whether fixed-degree $\sos$ proofs can be automated [O'Donnell (2017)].
Indeed, there are examples of polynomial systems with bounded coefficients that admit low-degree $\sos$ proofs, but these proofs necessarily involve numbers with an exponential number of bits, implying that low-degree $\sos$ proofs cannot always be found efficiently.
 
A sufficient condition derived from the Nullstellensatz proof system [Raghavendra and Weitz (2017)] identifies cases where bit complexity issues can be circumvented.
One of the main problems left open by Raghavendra and Weitz is proving any result for refutations, as their condition applies only to polynomial systems with a large set of solutions.

In this work, we broaden the class of polynomial systems for which degree-$d$ $\sos$ proofs can be automated. To achieve this, we develop a new criterion and we demonstrate how our criterion applies to polynomial systems beyond the scope of Raghavendra and Weitz's result. In particular, we establish a separation for instances arising from Constraint Satisfaction Problems ($\CSP$s).
Moreover, our result extends to refutations, establishing that polynomial-time refutation is possible for broad classes of polynomial time solvable constraint problems, highlighting a first advancement in this area.

\noindent\keywords{Sum of squares, Polynomial calculus, Polynomial ideal membership, Polymorphisms, \GB basis theory, Constraint satisfaction problems, Proof complexity.}
\end{abstract}
\newpage
{
  \hypersetup{linkcolor=black}
  \tableofcontents
}


\newpage
\section{Introduction}\label{sect:introduction}
Semidefinite programming (SDP) relaxations have been a powerful technique for approximation algorithm design ever since the celebrated result of Goemans and Williamson \cite{GoemansWilliamson1995}. With the aim to construct stronger and stronger SDP relaxations, the Sum-of-Squares ($\sos$) hierarchy has emerged as a systematic and versatile method for approximating many combinatorial optimization problems, see e.g. \cite{Lasserre2001, Parrilo03, FlemingKothariPitassi19, Laurent2009}.
However, fundamental questions remain unanswered. For instance, it is still unknown under what conditions \sos\ can be \emph{automated}, meaning whether one can find a degree-$d$ \sos\ proof in time $n^{O(d)}$, provided it exists.
%
O'Donnell \cite{odonnell2017} observed that the prevailing belief regarding the automatability of \sos\ using ellipsoid algorithms is not entirely accurate. Issues may arise when the only degree-$d$ proofs contain exceedingly large coefficients, thereby hindering the ellipsoid method from operating within polynomial time. 
In this paper, we establish novel conditions that ensure \sos\ automatability.

\paragraph{Polynomial optimization.}
Polynomial optimization asks for minimizing a polynomial over a given set of polynomial constraints. That is, given polynomials $r, p_1, \ldots, p_m \in \mathbb{R}[x_1, \ldots, x_n]$, the task is to find (or approximate) the infimum of the following probth:
\begin{align}\label{eqn:POP_formulation}
    \inf_{x\in S} r(x), \quad {\text{where}} \quad S=\{x\in \mathbb{R}^n \ | \ p_1(x)=\cdots=p_m(x)=0\}.
\end{align}
%
Typically, \( S \) is defined by a set of equality constraints, in this case \( \mathcal{P} = \{p_1, \ldots, p_m\}\), as well as a set of inequality constraints, \( \mathcal{Q} \). For all applications considered here, however, it suffices to restrict to the case where \( \mathcal{Q} = \emptyset \) and \( S \) is finite, enabling the modeling of various relevant combinatorial problems. Nonetheless, we emphasize that our results readily extend to the semialgebraic setting, where \( \mathcal{Q} \neq \emptyset \). For further details, we refer to \cref{sect:semialgebraic_sos_criterion}.

A common approach for solving (or approximating) a polynomial optimization problem is by means of sums of squares of polynomials, as we now explain.
\begin{definition}[$\sos$ Proof System]\label{def:SOS_proof}
    Let $\mathcal{P} = \{p_1=0,\ldots,p_m=0\}$ be a set of polynomial equations, and consider a polynomial $r\in \mathbb{R}[x_1, \dots, x_n]$. An $\sos$ proof of $``r\geq 0"$ (over $S$) from $\mathcal{P}$ is an identity of the form $r = \sum_{i=1}^{t_0} s_i^2 + \sum_{i=1}^m h_i p_i,$
    where $s_i, h_i\in \mathbb{R}[x_1, \ldots, x_n]$. Moreover, we say that the above $\sos$ proof 
    has \emph{degree} at most $d$ if $\deg(s_i^2) \leq d$, for all $i \in [t_0]$, and $\deg(h_i p _i) \leq d$ for all $i \in [s]$. An \emph{$\sos$ refutation of $\mathcal{P}$} is an $\sos$ proof of $``-1 \geq 0"$ from $\mathcal{P}$.
\end{definition}

The $\sos$ hierarchy is based on the following observation: if there exists an $\sos$ proof of $`` r - \theta \geq 0"$ from $\mathcal{P}$, then we have that $\min_{x \in S} r(x) \geq \theta$. Moreover, the supremum of the values $\theta$ such that there is an $\sos$ proof of $``r - \theta \geq 0"$ from $\mathcal{P}$ of degree $d$, is called $d$-th $\sos$ relaxation, also known as the $d$-th Lasserre relaxation of problem (\ref{eqn:POP_formulation}) \cite{Lasserre2001,Parrilo03}. It turns out that the $d$-th  $\sos$ relaxation can be formulated as an SDP of size $n^{O(d)}$.

\(\sos\) relaxations have gained increasing popularity and success; yet, they remain a relatively recent development. Fundamental questions about their properties and capabilities still lack definitive answers. 
O'Donnell~\cite{odonnell2017} posed the open problem of identifying meaningful conditions that ensure that “small” \(\sos\) proofs can be found. We will consider systems $\mathcal{P} = \{p_1 = 0,\ldots,p_m = 0\}$ of polynomials and an ``input" polynomial $r$ of degree at most $d$, with the (mild) assumption that the bit complexity needed to represent $\mathcal{P}$ and $r$ is polynomial in $n$. Moreover, we assume that $\mathcal{P}$ is explicitly Archimedean, i.e. there is $N<2^{poly(n^d)}$ such that there exists a ``small" $\sos$ proof of $``N - x_i^2 \geq 0"$ from $\mathcal{P}$ for any variable $x_i$.
We restate O'Donnell's question as follows: \emph{Consider an explicitly Archimedean polynomial system \(\mathcal{P}\); under what conditions on $\mathcal{P}$ does the following property hold?}
\begin{itemize}[label=$(\rm{P})$]
    \item Assume there exists an $\sos$ proof of $``r \geq 0"$ from $\mathcal{P}$ of degree $2d$. Then, for every $\varepsilon>0$, there also exists an $\sos$ proof of $``r +\varepsilon \geq 0"$ from $\mathcal{P}$ with degree $O(d)$ and coefficients bounded by $2^{poly(n^d,\lg \frac{1}{\varepsilon})}$.
\end{itemize}

The assumption of explicitly Archimedeanity guarantees that if there exists an approximate $\sos$ proof of $``r - \theta \geq 0"$, then there exists an (exact) $\sos$ proof of $``r - \theta + \varepsilon \geq 0"$, up to any arbitrary precision $\varepsilon$.
Moreover, explicit Archimedeanity implies that the SDP has no duality gap~\cite{JoszH16}. Therefore, it is often assumed in literature since numerical methods for solving SDPs are guaranteed to converge only when the duality gap is zero.

%
\paragraph{$\Nsatz$ criterion.}
Since O'Donnell~\cite{odonnell2017} raised his question in 2017, very few papers have been published that address this issue. An initial elegant solution to this question is provided by Raghavendra and Weitz~\cite{raghavendra_weitz2017}, which is based on the Nullstellensatz proof system~\cite{BeameIKPP94}, as we will now outline. For additional results from the literature related to this problem, see~\cref{sect:previous_work}.

We denote the vector space of polynomials for variables $x_1, \ldots, x_n$ up to degree $d$ as $\mathbb{R}[x_1, \dots, x_n]_d$. Moreover, we denote by $\I(S) = \{p \in \mathbb{R}[x_1,\ldots, x_n] \ | \ p(x) = 0 \  \forall x \in S\}$ the vanishing ideal generated by $S$, and by $\I_{d}(S) = \I(S) \cap \mathbb{R}[x_1, \ldots, x_n]_{d}$ the $d$-truncated ideal.  
\begin{definition}[$\Nsatz$ Proof System]\label{def:d-completeness}
    Consider a system of polynomial equations $\mathcal{P}=\{p_1=0,\ldots,p_m=0\}$. A \emph{Nullstellensatz ($\Nsatz$) proof} of $``p=0"$ from $\mathcal{P}$ is a sequence of polynomials $(h_1,\ldots,h_m)$ such that the polynomial identity $p=\sum_{i=1}^m h_ip_i$ holds. We say that the proof has degree $d$ if $\max_i\{\deg h_ip_i\}=d$.  We say that $\mathcal{P}$ is $\Nsatz$ \emph{$d$-complete} over $S$ if for every $p\in \I_{d}(S)$, the identity $``p=0"$ can be derived using a degree-$O(d)$ $\Nsatz$ proof from $\mathcal{P}$.
\end{definition}
Next, we recall the criterion proposed by Raghavendra and Weitz for the algebraic setting~\footnote{We remark that their criterion is formulated for the semialgebraic setting, i.e. when there are also inequalities.}. Moreover, for the sake of clarity of the exposition, we present their result in the case $S$ is finite. We define the algebraic variety \(S\) as the set of common zeros of \(\mathcal{P}=\{p_1=0,\ldots,p_m=0\}\). We first observe that this criterion necessitates a technical condition on the solution set \( S \), referred to as \emph{\(\delta\)-spectrality}, which we will outline below.
Let \(\mathbf{v}_d\) represent the column vector whose entries correspond to the elements of the standard monomial basis of \(\mathbb{R}[x_1, \dots, x_n]_d\). For \(\alpha \in \mathbb{R}^n\), \(\mathbf{v}_d(\alpha)\) denotes the vector of real numbers obtained by evaluating the entries of \(\mathbf{v}_d\) at \(\alpha\). 

\begin{definition}
    Let $S$ be a finite algebraic variety. We say that $S$ is $\delta$-spectrally rich up to degree $d$ if every nonzero eigenvalue of the moment matrix $\frac{1}{|S|} \sum_{\alpha \in S} \mathbf{v}_d(\alpha)\mathbf{v}_d^{\sf{T}}(\alpha)$ is at least $\delta$.
\end{definition}
This property holds for $\frac{1}{\delta} = 2^{poly(n^d)}$ in many natural instances, for example when $S \subseteq \{0,1\}^n$, or more in general, when $S \subseteq D^n$ for any finite domain $D \subseteq \mathbb{Q}$ (see \cref{sect:delta_spectrality} and \cite{raghavendra_weitz2017}).

\begin{theorem}[$\Nsatz$ criterion~\cite{raghavendra_weitz2017}]\label{th:Nsatz_crit}
    Let $\mathcal{P}$ be a system of polynomial equalities over $n$ variables with solution set $S$. Assume that 
    \begin{itemize}
        \item [(1)] $S$ is $\delta$-spectrally rich up to degree $d$.
        \item [(2)] $\mathcal{P}$ is $\Nsatz$ $d$-complete over $S$.\label{d-comp}
    \end{itemize} 
    Let $r$ be a polynomial and assume there exists an $\sos$ proof of $``r \geq 0"$ from $\mathcal{P}$ of degree $d$. Then, there also exists an $\sos$ proof of $``r \geq 0"$ from $\mathcal{P}$ with degree $O(d)$ and with absolute values of the coefficients bounded by $2^{poly(n^d ,\lg \frac{1}{\delta})}$.
\end{theorem}

This criterion is applicable to various optimization problems, including \textsc{Max-Clique}, \textsc{Matching}, and \textsc{Max-CSP}~\cite{raghavendra_weitz2017}. 
However, the \(\Nsatz\) criterion is subject to significant limitations. First, the criterion is sufficient but not necessary.
Second, it is important to observe how the $\Nsatz$ criterion (see condition \textit{(2)} in \cref{th:Nsatz_crit}) is influenced by the complexity of a well-known problem known as the \emph{Ideal Membership Problem} (\(\IMP\)). This problem involves determining whether an input polynomial \(r\) belongs to the ideal generated by \(\{p_1, \ldots, p_m\}\). We denote the \(\IMP\) where the input polynomial \(r\) has degree at most \(d=O(1)\) as \(\IMP_d\). The \(\IMP\) was first studied by Hilbert~\cite{Hilbert1893} and is a fundamental algorithmic problem with significant applications in solving polynomial systems and polynomial identity testing (see, for example, \cite{Cox}). In general, the \(\IMP\) is notoriously intractable, and the results of Mayr and Meyer demonstrate that it is EXPSPACE-complete \cite{Mayr1989, MAYR1982305}.
It remains unclear under what conditions the \(\IMP\) is tractable within the \(\Nsatz\) proof system, specifically regarding when condition \textit{(2)} in \cref{th:Nsatz_crit} is satisfied. More importantly, the limitations of the \(\Nsatz\) proof system (see e.g. \cite{FlemingKothariPitassi19}) affect the applicability of \cref{th:Nsatz_crit}. 
In simpler terms, it is intuitive to suggest that if we could replace the \(\Nsatz\) proof system with a more powerful proof system, we would be able to broaden the applicability of the criterion to new problems.

Finally, a key limitation—and one of the main open problems left by Raghavendra and Weitz \cite{raghavendra_weitz2017, Weitz:Phd}—is the inapplicability of the $\Nsatz$ criterion to \(\sos\) refutations.

For example the $\Nsatz$ criterion does not allow one to show that the following decision problem can be solved in polynomial time. 

\begin{problem}[Degree-$d$ Sum-of-Squares Refutation for $\CSP$]\label{prob:sos-csp}
Given a Constraint Satisfaction Problem (\(\CSP\)) with constraints \(\phi_1(x) = 0, \dots, \phi_m(x) = 0\) over a finite domain, decide whether:
\begin{itemize}
    \item \textbf{YES}: There exists a degree-\(d\) sum-of-squares (\(\sos\)) proof of the infeasibility of the system, i.e., a derivation of \(-1 \geq 0\) from the axioms \(\phi_1(x) = 0, \dots, \phi_m(x) = 0\) and the domain constraints.
    \item \textbf{NO}: No such degree-\(d\) \(\sos\) proof exists.
\end{itemize}
\end{problem}
Let us call this problem ``$\sos$-$\CSP$". This is perhaps the most natural formulation of ``the $\sos$ algorithm for $\CSP$s". 
It is quite striking that we still do not know whether there exists or not a polynomial-time decider for $\sos$-$\CSP$ (even for certain restricted classes of problems).

\subsection*{References to the related literature}\label{sect:previous_work}
O'Donnell \cite{odonnell2017} raised the issue of \(\sos\) bit complexity in 2017, as discussed in \cref{sect:introduction}. 
O'Donnell also presented an example of a polynomial system with bounded coefficients that allows for a degree 2 \(\sos\) proof, which necessarily has doubly-exponential coefficients.

The aforementioned result (\cref{th:Nsatz_crit}) by Raghavendra and Weitz \cite{raghavendra_weitz2017} offered an initial elegant, albeit partial, solution.
Raghavendra and Weitz expanded O'Donnell's work and presented an example of a polynomial system containing the Boolean constraints and a polynomial that admits degree 2 $\sos$ proof, but for which any $\sos$ proof of degree $O(\sqrt{n})$ must have coefficients of doubly-exponential magnitude in $n$.

Interestingly, Hakoniemi \cite{Hakoniemi21} demonstrated that both \(\sos\) and Polynomial Calculus ($\PC$)  refutations over Boolean variables encounter the same bit complexity issue. This finding also raises significant concerns regarding the frequently asserted degree automatability of \(\PC\).

Furthermore, strategies in \cite{BharathiM21,BulatovRSTOC22,BulatovARXIV21, Mastrolilli21TALG} to address the problem of \sos\ bit complexity involve replacing the original input polynomial constraints \(\mathcal{P}\) with a new set of polynomials $\mathcal{P}^{(d)}$ that satisfies the \(\Nsatz\) criterion, and generally depends on the \sos\ degree \(d\). This set $\mathcal{P}^{(d)}$ is computed externally (by an algorithm specifically designed for this purpose \footnote{In general such an algorithm cannot be simulated by $\sos$. We defer the interested reader to \cref{sect:csp_literature} for details.}), serving as the input for \(\sos\) in place of $\mathcal{P}$. For example, in the semilattice case, if $\mathcal{P}$ consists of $m$ polynomials, the set $\mathcal{P}^{(d)}$, used in \cite{BulatovRSTOC22,Mastrolilli21TALG}, is generated by a specific algorithm and has a size of \( m^{O(d)} \); that is, $\mathcal{P}^{(d)}$ depends on \( d \) and grows exponentially with the \(\sos\) degree \( d \). This preprocessing step ensures that \(\sos\) retains ``low" bit complexity, but only if $\mathcal{P}$ is substituted with $\mathcal{P}^{(d)}$. Essentially, the approach utilized in \cite{BharathiM21,BulatovRSTOC22,BulatovARXIV21, Mastrolilli21TALG} is to apply the \(\Nsatz\) criterion without enhancing or extending it, with the goal of replacing the initial input polynomial system with a new one that is computed externally and satisfies the \Nsatz\ criterion.
Our results demonstrate that all preprocessing steps employed in \cite{BharathiM21,BulatovRSTOC22,Mastrolilli21TALG} are unnecessary, as \(\sos\) achieves low bit complexity for any fixed \( d \) when \(\mathcal{P}\) is provided directly as input.

Recently, Gribling et al. \cite{Gribling23} showed that, under specific algebraic and geometric conditions, $\sos$ relaxations can be computed in polynomial time. However, as they noted, their algebraic conditions are more restrictive than $d$-completeness. Their geometric conditions apply to systems of only inequality constraints with full dimensional feasibility set.  Additionally, Palomba et al. \cite{palomba} showed that, under some mild conditions, sum-of-squares bounds for copositive programs can be computed in polynomial time. These results also yielded insights into the bit complexity of $\sos$ proofs.

The quest of characterizing conditions for ensuring tractability of the $\sos$ proof system fits into the more general context of real algebraic geometry, and in particular of the so-called \textit{effective Positivstellensatz}, i.e., the study of the complexity for representing polynomials using rational sums of squares. To this end, we mention Baldi et al. \cite{BaldiKM24}, who recently proved an exponential upper bound on the coefficients' bit complexity of sums of squares proofs in the general case of radical zero-dimensional ideals when the equality constraints form a graded basis.
Furthermore, as observed, $\sos$ feasibility can be reformulated as an SDP feasibility problem, which remains a well-known open question.
In this context, we highlight the work of Pataki and Touzov \cite{PatakiT24}, who showed that many SDP's have feasible sets whose elements have large encoding size. Moreover, they initiated the study of characterizing the conditions under which SDP feasibility sets with large encoding sizes occur.

Several papers in the literature investigate the automatability of the $\sos$ proof system in relation to degree lower bounds. Specifically, various instances have been studied where a set of polynomial equations $\mathcal{P} = \{p_1=0, \ldots, p_m=0\}$ and a polynomial $q$ satisfy the condition that `` $q \geq 0$ on $S$", but any $\sos$ proof of this fact necessarily requires a high degree (see, e.g.,~\cite{GrigorievHP_MMJ02,KLM-hard-prob-formulation,KLM-unbounded-SOS-hierarchy-integrality-gap,KLM-tight-SOS-LB-binary-POP,KLM-SOS-hierarchy-LB-symmetric-formulations,potechin:SOS-LB-from-symmetry}). Within the context of the Lasserre hierarchy, these examples correspond to polynomial optimization problems for which many rounds of the hierarchy are needed to reach optimality.

In contrast, our focus is on a different aspect of automatability: we aim to understand under what conditions a fixed-degree level of the Lasserre hierarchy can be computed in time polynomial in the input size (up to a prescribed precision). This shifts the question from degree necessity to degree tractability within a computational framework.

\subsection{Our contributions}\label{sect:contribution}

Our main contribution is to study and expand the class of polynomial systems for which finding degree-$d$ \(\sos\) proofs can be automated. To this end, we first introduce a new criterion based on the \textit{Polynomial Calculus} ($\PC$) proof system which guarantees that property (\textsc{P}) holds, referred to as the \textit{$\PC$ criterion}. This criterion holds both for $\sos$ refutations and for $\sos$ proofs over feasible systems over finite domains. Remarkably, as we will demonstrate, this criterion applies to a broad class of instances arising from Constraint Satisfaction Problems ($\CSP$s) where the $\Nsatz$ criterion does not. Specifically, we will establish tractability results for broad class of $\sos$-$\CSP$, and, moreover, prove complete degree-$d$ automatability beyond refutations for certain polynomial systems arising from $\CSP(\Gamma)$. The proof of the $\PC$ criterion combines several results, including a simulation of the $\PC$ proof system by the $\sos$  proof system together with the development of a different criterion based on the $\sos$ proof system, called the \(\sos_{\varepsilon}\) criterion.

\subsubsection*{$\PC$ criterion} \label{sect:PC-crit}
We begin by introducing \textit{polynomial systems over finite domains}, i.e., systems where each variable is restricted to assume values from a fixed set of $k$ rational numbers $ \rho_{1}, \rho_{2}, \ldots, \rho_{k} $. Given a polynomial system of equations $\mathcal{P}$ and a finite domain $D = \{\rho_1, \ldots, \rho_{k}\}$, we say that $\mathcal{P}$ is a \emph{polynomial system over finite domain $D$} if it includes the following \textit{domain polynomials}:
\begin{align}
D_{k}(x_i)&=(x_i-\rho_{1})(x_i-\rho_{2})\cdots (x_i-\rho_{k}) = 0\qquad i\in[n].
\end{align}
These polynomials ensure that each variable $x_1, \ldots, x_n$ is constrained to take values from $D$. Note that this formulation generalizes polynomial systems with Boolean variables, i.e., where the constraint \( x_i^2 - x_i = 0 \) enforces $x_i$ to be either 0 or 1.

Polynomial Calculus over $\mathbb{R}$ ($\textrm{PC}\slash \mathbb{R}$), or in short Polynomial Calculus ($\PC$), is a proof system used in computational complexity and proof complexity to study the efficiency of algebraic reasoning. It operates over polynomials and is particularly useful in analyzing the complexity of solving systems of polynomial equations. The goal is to derive the polynomial \(1\) (i.e., show inconsistency) or to demonstrate that a certain polynomial is implied by the given system. 
Originally introduced as a \textit{refutation system} in \cite{CleggEI96}, $\PC$ can be viewed as a degree-truncated version of Buchberger’s algorithm \cite{JeffersonJGD13, Cox}. Essentially, $\PC$ is a dynamic version of the $\Nsatz$ proof system, employing schematic inference rules to reason about polynomial equations. We emphasize that, for the remainder of the paper, we will consider \(\PC\) in the broader sense of polynomial derivation (so not restricted to refutation) and with polynomials over the reals.

$\PC$ consists of the following derivation rules for polynomial equations $(f = 0), (g = 0) \in  \mathcal{P}$, domain polynomial equations $(D_{k}(x_j)=0)$, variable $x_j$, and numbers $a, b \in \mathbb{R}$
\begin{equation}\label{eq:PCrules_introduction}
    \frac{}{f=0} \qquad \frac{}{D_{k}(x_j)=0} \qquad \frac{f=0 \qquad g=0}{a f + b g =0} \qquad \frac{f=0}{x_jf=0} 
\end{equation}

A \emph{$\PC$ derivation} of $``r=0"$ from $\mathcal{P}$ is a sequence $(r_1=0,\ldots,r_L=0)$ of polynomial equations iteratively derived by using \eqref{eq:PCrules_introduction} with $r=r_L$. The \emph{size} of a derivation is the sum of the sizes of the binary encoding of the polynomials in the derivation and the \emph{degree} is the maximum degree of the polynomials in the derivation. A \emph{PC refutation} is a derivation of $``1=0"$.

Next, we present one of our main contributions, namely a framework for showing that Property (\textsc{P}) holds for certain polynomial systems over finite domains.

\begin{theorem}[$\PC$ criterion]\label{th:PC_criterion_introduction}
    Let $\mathcal{P}$ be a polynomial system over a finite domain $D$ of $k$ rational values, let $S$ be its variety. Let $\mathcal{G}_{2d}$ be a $2d$-truncated \GB basis of $I(S)$ according to the $\grlexns$ order such that $\|\mathcal{G}_{2d}\|_{\infty}\leq 2^{poly(n^d)}$. Assume that, for every $g\in \mathcal{G}_{2d}$, there exist a $\PC$ derivation of $g$ from $\mathcal{P}$ of size $poly(n^d)$ and degree $O(d)$. 
    
    Let $r$ be a polynomial and assume there exists an $\sos$ proof of $``r \geq 0"$ from $\mathcal{P}$ of degree $2d$. Then, for every $\varepsilon>0$, there exists an $\sos$ proof of $``r + \varepsilon \geq 0"$ of degree $O(d)$
    such that the absolute values of the coefficients of every polynomial appearing in the proof are bounded by $2^{poly(n^d, \lg \frac{1}{\varepsilon})}$.
\end{theorem}

Moreover, suppose $\mathcal{P}$ is $\Nsatz$ $d$-complete over $S$. It follows that, for every $g \in \mathcal{G}_{2d}$, the identity $``g=0"$ admits a degree-$O(d)$ $\Nsatz$ proof from $\mathcal{P}$. Further, $\PC$ is known to be strictly stronger than $\Nsatz$ \cite{CleggEI96}, thus implying that our criterion is also more powerful than the $\Nsatz$ criterion.
The separation between the $\Nsatz$ and the $\PC$ criteria is strict and it will be further discussed in \cref{sec:separation-intro}.

The proof of \cref{th:PC_criterion_introduction} combines multiple techniques, which we will present in greater detail in \cref{sect:outline_PC_proof}. A key component of this proof is the development of a different criterion, the $\sos_\varepsilon$ criterion. As we will see, this criterion provides a more general framework that extends beyond finite domains (see \cref{ex:separation2}). Nevertheless, the strong connection between $\PC$ and Buchberger's algorithm makes the $\PC$ criterion an effective tool for many instances where the $\IMP_d$ can be efficiently solved \cite{Cox}.

\subsubsection*{Two main applications from $\CSP$s.}
In the following, we present our main two applications of \cref{th:PC_criterion_introduction}.
In both cases we focus on restricted classes of Constraint Satisfaction Problems (\(\CSP\)s), denoted as \(\CSP(\Gamma)\), where constraints are limited to relations from a specified set~\(\Gamma\). These language restrictions have proven effective for analyzing computational complexity classifications and other algorithmic properties of $\CSP$s, leading to recent breakthroughs in \cite{Bulatov17,Zhuk17,Zhuk20} (see, e.g., \cite{barto_et_al:DFU:2017:6959,2017dfu7,Chen09} and \cref{sect:CSP_IMP_background} for further details and necessary background).

\paragraph{First main application: refutation for bounded width $\CSP$s}
%
All known tractable Constraint Satisfaction Problems $\CSP(\Gamma)$ for a fixed constraint language $\Gamma$  are solvable using two fundamental algorithmic principles. The first relies on the \emph{few subpowers property} (see e.g. \cite{barto_et_al:DFU:2017:6959}). The second, \emph{local consistency checking}, is the most widely known and natural approach for solving $\CSP$s \cite{BartoK14,barto_et_al:DFU:2017:6959,Bulatov17}.

We consider the class of constraint languages $\Gamma$ for which $\CSP(\Gamma)$ has \emph{bounded width}, meaning that it can be solved by a local consistency checking algorithm (see e.g.~\cite{BartoK14,barto_et_al:DFU:2017:6959}). Identifying and characterizing such languages is crucial for understanding the tractability of constraint satisfaction problems \cite{BartoK14,barto_et_al:DFU:2017:6959}. Note that for languages that rely on the {few subpowers property} in general $\sos$ requires high degree for refutation \cite{barto_et_al:DFU:2017:6959,BartoK14, GRIGORIEV2001613}.

As a corollary of \cref{th:PC_criterion_introduction}, we establish the polynomial-time feasibility of the $\sos$ refutations for the whole class $\sos$-$\CSP(\Gamma)$ problems (see \cref{prob:sos-csp}) for which $\CSP(\Gamma)$ has \emph{bounded width}. 
\begin{corollary}\label{th:BW}
    For constraint languages $\Gamma$ over finite domains for which $\CSP(\Gamma)$ has \emph{bounded width}, the $\sos$-$\CSP(\Gamma)$ \cref{prob:sos-csp} can be solved in polynomial time for any fixed degree~$d$.
\end{corollary}
\begin{proof}[Proof sketch]
    For refutations, \cref{th:PC_criterion_introduction} requires that there exists a $\PC$ derivation of $``1 = 0"$ from $\mathcal{P}$ of size $poly(n^d)$ and degree $O(d)$.
    The claim follows by observing that the local consistency algorithm can be simulated by a truncated Buchberger's algorithm (that we call $\PC$). Thus, any information obtained by enforcing local consistency{, and therefore, by definition, by deciding any bounded width language,} can be obtained by performing a truncated Buchberger's algorithm \cite{JeffersonJGD13,Bulatov23}.
\end{proof}

Note that both the decision and search versions of \cref{prob:sos-csp} with bounded width are solvable in polynomial time for any fixed degree~$d$, as a consequence of \cref{th:PC_criterion_introduction}.

Further, we mention that in \cite{ThapperZ18} it was obtained a similar result in the context of the Sherali-Adams proof system. However, this result applies only to a fixed limited form of $\mathcal{P}$, namely the Boolean canonical linear program. As remarked in \cref{sect:previous_work}, our focus is on deriving SoS proofs directly from $\mathcal{P}$ with variables over general finite domains without any preprocessing. To this end, in \cref{th:BW}, we demonstrate that for \emph{any} system of equations $\mathcal{P}$ over a finite domain that defines a bounded-width relation, finding SoS refutations directly derived from $\mathcal{P}$ can be automated.

\paragraph{Second main application: strong separations arising from $\CSP$s}\label{sec:separation-intro}
%
As second application of \cref{th:PC_criterion_introduction}, we examine constraint languages (and polynomial equations) that are closed under the semilattice and dual discriminator polymorphisms\footnote{In the context of $\CSP$s, a \emph{polymorphism} is a special kind of function that helps us understand the structure of the constraints. Specifically, it is a function that combines multiple solutions of a \CSP\ in a way that still satisfies the constraints. Polymorphisms are useful because they reveal patterns in the constraints, and studying them can help determine how easy or hard a \CSP\ is to solve. We refer to \cref{def:polymorph} for a formal definition.} (see, e.g., \cite{barto_et_al:DFU:2017:6959} and \cref{sect:PCsemilattice,sect:PCdual} for the necessary background). 
Propositional formulas from \textsc{HORN-SAT} or \textsc{2-SAT} can be easily translated into system of polynomial equations that are semillatice or dual discriminator closed, respectively.
Moreover, these two classes extend \textsc{HORN-SAT} and \textsc{2-SAT} formulas, respectively, to general finite domain cases and have held a significant role in the theory of Constraint Satisfaction Problems of the form $\CSP(\Gamma)$; see, e.g.,~\cite{barto_et_al:DFU:2017:6959,2017dfu7} and references therein.

\begin{theorem}\label{th:semilattice+dualdiscr}
    For a system $\mathcal{P}$ of polynomial equations over $n$ variables that is closed under the semilattice (or dual discriminator) polymorphism, then the $\PC$ criterion (\cref{th:PC_criterion_introduction}) applies.
\end{theorem}
%
Note that these classes of problems are known to be bounded width (see e.g. \cite{barto_et_al:DFU:2017:6959}). Therefore, by \cref{th:BW} the refutation \cref{prob:sos-csp} can be solved in polynomial time.
However, \cref{th:semilattice+dualdiscr} establishes a significantly stronger result.
Indeed, \cref{th:semilattice+dualdiscr} indicates that \emph{any} degree \( d \) \sos\ proof of $p\geq 0$, for any polynomial $p$, can be computed in \( n^{O(d)} \) time with arbitrary precision (not only degree-$d$ $\sos$ proofs for $-1\geq 0$, as required by refutation). 

We emphasize that Buss and Pitassi \cite{BussP98} 
show that the \(\Nsatz\) proof system necessitates a degree \(\Theta(\log n)\) proof for the induction principle \(\text{IND}_n\), a polynomial inference rule that can be formalized as a derivation in either \textsc{HORN-SAT} or \textsc{2-SAT} formulae. As a result, if \(\mathcal{P}\) is closed under the semilattice (or dual discriminator) polymorphism, it cannot be \(\Nsatz\) \(d\)-complete for any \(d = o(\log n)\).
Thus, \cref{th:PC_criterion_introduction}, \cref{th:semilattice+dualdiscr} and \cite{BussP98} establish a clear separation between the $\PC$ criterion and the \(\Nsatz\) criterion. 

    Polynomial Calculus (\PC) is a rule-based, dynamic extension of Nullstellensatz (see e.g.~\cite{FlemingKothariPitassi19}). Due to its dynamic nature, it can sometimes achieve a refutation of significantly lower degree through cancellations than would be possible with the static Nullstellensatz system. A notable example is the induction principle \(\text{IND}_n\) mentioned above, which has degree 2 refutations in $\PC$. By contrast, its Nullstellensatz degree has been shown to be \(\Theta(\log n)\)~\cite{BussP98}.

    We demonstrate that $\PC$, in addition to solving refutation for the very special case of \(\text{IND}_n\) with low degree, also addresses the much more general Ideal Membership Problem $\IMP_d$ in $n^{O(d)}$ time for two families of problems that significantly generalize \textsc{HORN-SAT} and \textsc{2-SAT} in multiple respects and apply to all finite rational domains. This also demonstrates that \(\PC\) is complete and free from bit complexity issues for these problems (Hakoniemi recently raised concerns regarding the bit complexity in \(\PC\) \cite{Hakoniemi21}, see also \cref{sect:open problem}).  
    

    {This result is not only intrinsically interesting but also closely aligned with the main goal of this article. Indeed, the $\PC$ criterion demonstrates that solvability via Polynomial Calculus and the bit complexity of Sum-of-Squares are deeply interconnected.
    }
    Finally, we emphasize that it is not implied by the recent result of Bulatov and Rafiey \cite{BulatovRSTOC22}; more details are given in \cref{sect:applications}.

The proof of this broad generalization is technically complex and lengthy, necessitating a dedicated space with the necessary preliminaries. Therefore, we defer the full discussion—including a detailed review, proof, the underlying intuition and the literature review—to \cref{sect:applications}.



This paper aims to deepen our understanding of the bit complexity issue and to explore the conditions under which it arises. For instance, we demonstrate that all preprocessing steps aimed at replacing \(\mathcal{P}\) with a new set \(\mathcal{P'}\) to satisfy the \(\Nsatz\) criterion, as used in \cite{BulatovRSTOC22,Mastrolilli21TALG, BharathiM22, BharathiM21} to circumvent the bit complexity issues of \(\sos\) for semilattice and dual discriminator polymorphisms, are unnecessary. Specifically, \(\sos\), when applied directly to \(\mathcal{P}\) as input, achieves low bit complexity for any fixed~\(d\) (refer to \cref{sect:csp_literature} for a more detailed discussion). This result appears to support and extend the hypothesis that CNF formulas do not exhibit a bit complexity issue, an open question posed by Hakoniemi~\cite{Hakoniemi21} (see also \cref{sect:open problem}).

\subsubsection*{$\sos$ (approximately) simulates $\PC$}\label{sect:PC_criterion_introduction}

As a main contribution, we address the existing knowledge gap regarding the relationship between $\sos$ and $\PC$ in the general finite domain setting.
Our main result shows that $\sos$ can simulate $\PC$ derivations in this setting with an arbitrarily small error. Essentially, if $\PC$ can derive the equation $``p=0"$, then, for any arbitrary $\varepsilon>0$, $\sos$ can prove the statements $``p+\varepsilon\geq 0"$ and $``-p+\varepsilon \geq 0"$ with only a polynomial increase in size.
While this result serves as a main technique for proving the $\PC$ criterion, as outlined in \cref{sect:outline_PC_proof}, it is also valuable on its own, as we present below.
Our result builds upon and generalizes the simulation result of Berkholz~\cite{berkholz18} for Boolean variables to the broader context of general finite domains.

Berkholz~\cite{berkholz18} related different approaches for proving the unsatisfiability of a system of real polynomial equations. Over Boolean variables, he showed that SoS simulates $\PC$ refutations: any \(\PC\) refutation of degree \(d\) can be converted into an \(\sos\) refutation of degree \(2d\), with only a polynomial increase in size.

In the non-Boolean setting, there are systems of equations that are easier to refute for $\PC$ than for SoS \cite{GrigorievV01}. Grigoriev and Vorobjov \cite{GrigorievV01} show that the simulation of $\PC$ by SoS does not hold in the non-Boolean case, namely when the Boolean axioms $x_j^2-x_j=0$ are omitted. For example, the so-called telescopic system of equations, $\{ yx_1=1, x_1^2=x_2, \ldots, x_{n-1}^2=x_n, x_n=0\}$, has a $\PC$ refutation of degree $n$, but it requires exponential refutation degree in SoS \cite{GrigorievV01}. It is worth noting that a similar (although much weaker than the one present in \cref{th:sos_sim_PC_introduction}) generalization of Berkholz's result was considered in \cite{Sokolov20}, when the variables take the values $\pm 1$, and in \cite{PartFTT21}, for a variation of $\PC$ endowed with a ``radical rule" and a ``sum-of-squares rule".

Whether $\sos$ could simulate $\PC$ in the general finite domain setting, where variables can take values from any finite set, has remained an open question, despite known limitations in specific non-Boolean cases. 

In this work, we answer this open question and complement the results in~\cite{GrigorievV01,berkholz18} by establishing the following theorem. In summary, we first derive the following lemma.

\begin{lemma}\label{th:sos_sim_PC_introduction}
    Let $\mathcal{P}$ be a system of polynomial equations over a finite domain $D$ with $|D| = k$. Assume that $``r=0"$ has a $\PC$ derivation of degree $d$ and size $S$ from $\mathcal{P}$. Then $`` -r^2\geq 0"$ has an $\sos$ proof of degree $2(d+k-1)$ with coefficients of size $poly(k,S)$.
\end{lemma}

The overall structure of the proof partially mirrors that in \cite{berkholz18}, but with notable differences. In particular, new ideas and techniques are introduced in the simulation of the multiplication rule of \(\PC\).

Furthermore, note that the result holds for the particular case of refutations, i.e. when $r=1$. Indeed, if there exists a $\PC$ refutation of $\mathcal{P}$, i.e., a derivation of $1=0$, then \cref{th:sos_sim_PC_introduction} implies that there exists an $\sos$ refutation $``-1 \geq 0"$ with only polynomial increasing.

Then, employing \cref{th:sos_sim_PC_introduction}, we prove the following result.

\begin{theorem}\label{th:sospc_introduction}
    $\sos$ \emph{approximates} $\PC$ with degree linear in the domain size $k$ over general finite domains. That is, if there exists a \(\PC\) derivation of \( ``r=0" \) with degree \( d \) and size \( S \), then for every \( \varepsilon > 0 \), we have \(\sos\) proofs of \( ``r+\varepsilon \geq 0" \) and \( ``-r+\varepsilon \geq 0" \) with degree \( O(d+k) \) and coefficients bounded by \( 2^{\text{poly}(k, S, \lg \frac{1}{\varepsilon})} \).
\end{theorem}

\subsubsection*{Outline of the techniques for proving the $\PC$ criterion}\label{sect:outline_PC_proof}

Below we outline the main techniques that will be used in the proof of the PC criterion. 

\begin{enumerate}
    \item \textbf{$\sos_\varepsilon$ criterion} (\cref{sect:sos_criterion_subsection}) We begin by introducing a general criterion, called the $\sos_\varepsilon$ criterion, which ensures that Property (\textsc P) is satisfied and, consequently, that $\sos$ can be automated. This criterion is a natural generalization of the $\Nsatz$ criterion~\cite{raghavendra_weitz2017} (see \cref{th:Nsatz_crit}), and serves as a foundational tool for presenting our main contributions. The key distinction lies in the notion of approximate completeness: the $\sos_\varepsilon$ criterion requires $\sos_\varepsilon$ completeness, a relaxed condition than the one required by \cref{th:Nsatz_crit}, as discussed in \cref{sect:sos-criterion}.
Informally, a system $\mathcal{P}$ is $\sos_\varepsilon$ complete if for every $q \in \I_d(S)$ and every $\varepsilon > 0$, there exist an $\sos$ proof of the inequality $``q + \varepsilon \geq 0"$ using bounded coefficients (see Section~\ref{sect:SOS_completeness} for the precise definition).
    \item \textbf{$\sos$ approximability of polynomial systems}(\cref{sect:SOS_completeness}) As previously mentioned, the main difference between the $\Nsatz$ and $\sos_\varepsilon$ criteria lies in their respective notions of completeness. Therefore, a main challenge in applying the $\sos_\varepsilon$ criterion is proving that a given system $\mathcal{P}$ is $\sos_\varepsilon$ complete. To this end, in \cref{sect:SOS_completeness} we present the notion of $\sos$-approximation ($\lesssim$) between polynomial systems defining the same zero set, which turns out to be a powerful tool for showing $\sos_\varepsilon$ completeness, and applying the $\sos_\varepsilon$ criterion.
The key advantage of $\sos$-approximation is that it allows the inheritance of $\sos_\varepsilon$ completeness between systems: under mild conditions, if $\mathcal{P}$ is $\sos_\varepsilon$ complete and $\mathcal{P} \lesssim \mathcal{Q}$, then $\mathcal{Q}$ is also $\sos_\varepsilon$ complete.
    \item {\textbf{SoS approximately simulates $\PC$}} (\cref{sect:proof_sos_sim_PC}) The final tool we use to establish the PC criterion is the simulation of the $\PC$ derivation system by $\sos$. As previously emphasized, this simulation, presented formally in Lemma~\ref{th:sos_sim_PC_introduction} and Theorem~\ref{th:sospc_introduction}, is one of the main contributions of this work and can be appreciated independently. However, it also plays a crucial role in proving another major result: the PC criterion (Theorem~\ref{th:PC_criterion_introduction}).
The proof strategy proceeds as follows. Under the theorem’s hypotheses, the truncated Gröbner basis $\mathcal{G}$ is easily shown to be $\Nsatz$-complete and hence $\sos_\varepsilon$ complete. Using our simulation results, we then establish the following chain of $\sos$-approximations:
   $$\mathcal{G}\lesssim \mathcal{G}^2 \lesssim \mathcal{P}.$$
Finally, using the properties of $\sos$-approximability, we conclude that $\mathcal{P}$ is also $\sos_\varepsilon$ complete. 
\end{enumerate}

In the following sections, we explore these three concepts in more detail. Then, in \cref{sect:pc_crit}, we combine these techniques to prove the $\PC$ criterion (\cref{th:PC_criterion_introduction}).

\subsection{Structure of the article}
In Section \ref{sect:sos-criterion}, we give a full exposition of the $\sos_\varepsilon$ criterion. Subsequently, we develop several tools to facilitate its application. In Section \ref{sect:PC_criterion}, we focus on the case of polynomial systems over finite domains, where we establish a connection to the Polynomial Calculus ($\PC$) proof system and derive a weaker version of the $\sos_\varepsilon$, called the $\PC$ criterion. The latter will be used for the separation in the radical ideal case. \cref{sect:applications,sect:semilattice-proof,sect:dual-proof} are devoted to construct a class of examples arising from Constraint Satisfaction Problems, demonstrating a separation between our new criterion and the $\Nsatz$ criterion. Section \ref{sect:applications} begins with an overview of relevant background and related results, followed by a description of our proof strategy and main results. In Sections \ref{sect:semilattice-proof} and \ref{sect:dual-proof}, we present the detailed proofs. 
Finally, in Section \ref{sect:open problem}, we finish with concluding remarks and discuss potential research directions.

\section{Preliminaries}\label{sect:prelim}
Consider a set of variables $\{x_1,\ldots,x_n\}$ and denote the vector space of polynomials in $x$ up to a fixed degree $d = O(1)$ as $\mathbb{R}[x_1, \dots, x_n]_d$. We denote as $\mathbf{v}_d$ being the column vector whose entries are the elements of the usual monomial basis of $\mathbb{R}[x_1, \dots, x_n]_d$ and, if $\alpha \in \mathbb{R}^n$, $\mathbf{v}_d(\alpha)$ is the vector of reals whose entries are the entries of $\mathbf{v}_d$ evaluated at $\alpha$. It follows that for any polynomial $u(x) \in \mathbb{R}[x_1, \dots, x_n]_d$, it holds that $u(x) = u^{T} \mathbf{v}_d$ for some $u \in \mathbb{R}^n$. 
We will consider systems $\mathcal{P} = \{p_1 = 0,\ldots,p_m = 0\}$ of polynomial equations and an "input" polynomial $r$ of degree at most~$d$, with the (mild) assumption that the bit complexity needed to represent $\mathcal{P}$ and $r$ is polynomial in $n$. The string $l$ representing polynomials $r, p_1, \dots, p_m$ will be the input for our certification problems and this last assumption allows us to reduce any complexity reasoning to $n$ instead of the length of $l$. We will sometimes refer to $\mathcal{P}$ as a set of \emph{constraints} or \emph{axioms}.

Next, we need to define the measures of norm and bit complexity for different objects. For the first measure, consider a polynomial $p = \sum_\alpha c_\alpha x^\alpha$, we define $\| p \|_{\infty} = \max_\alpha |c_{\alpha}|$. Similarly, for a set of polynomial we have $\| \mathcal{P} \|_{\infty} = \max_{p \in \mathcal{P}} \| p \|_{\infty}$. Throughout this paper, we will assume that the set $S$ of common zeros of $\mathcal P$ is \textit{finite} and that  $\| S \| := \max_{\alpha \in S} \| \alpha \| < 2^{poly(n^d)}$. These assumptions are very general and are met in many different contexts. For the second measure, consider a polynomial $p$ (or a polynomial system $\mathcal{P}$). We define the \emph{bit complexity} of $p$ (or $\mathcal{P}$) as the minimum length of a bit-string representing $p$ (or $\mathcal{P}$) when the rational numbers are represented with their reduced fractions written in binary (see e.g.~\cite{Hakoniemi21}). As noted above, the bit complexity of $\mathcal{P}$ is assumed polynomial in $n$.

\subsection*{Explicit Archimedeanity}

We recall the notion of explicit Archimedeanity of polynomial systems. This property plays a crucial role in the context of computations of $\sos$ proofs \cite{JoszH16,Laurent2009}. The Archimedean property in algebraic optimization requires all variables to be bounded within some compact set. The \emph{explicitly Archimedean} condition strengthens this by demanding that boundedness of variables is efficiently certifiable via $\sos$ proofs. We give the following formal definition.

\begin{definition}[Explicitly Archimedean System]
    A system of polynomials $\mathcal{P}$ is said to be \emph{explicitly Archimedean} if one of the following equivalent conditions holds.
    \begin{itemize}
        \item \label{def:exp_arch} For every degree $k$ polynomial $p\in \mathbb{R}[x_1, \dots, x_n]$, there exists $0<N_p \leq 2^{poly(n^{\max\{k,d\}}, \lg \|p\|_{\infty})}$ such that there exists a $\sos$ proof of $``N_p-p\geq 0"$ from $\mathcal{P}$ of degree $O(k)$ and with coefficients bounded by $2^{poly(n^{\max\{k,d\}}, \lg \|p\|_{\infty})}$.
        \item \label{prop:bounded_variables_archimedeanity} For every $i \in [n]$ there exists $0 < N_{x_i} \leq 2^{poly(n^d)}$ such that there exist $\sos$ proofs $``N_{x_i} - x_i \geq 0"$ and $``N_{x_i} + x_i \geq 0"$ from $\mathcal{P}$ of degree $O(d)$ and with coefficients bounded by $2^{poly(n^d)}$.
        \item For every $i \in [n]$ there exists $0 < N_{x_i^2} \leq 2^{poly(n^d)}$ such that there exist $\sos$ proofs $``N_{x_i^2} - x_i^2 \geq 0"$ from $\mathcal{P}$ of degree $O(d)$ and with coefficients bounded by $2^{poly(n^d)}$.
    \end{itemize}
\end{definition}

The assumption of explicit Archimedeanity is met in numerous natural cases. This is the case for Boolean systems, i.e. systems that contain the Boolean constraints $x_i^2 - x_i = 0$ for every variable, where it suffices to set $N_{x_i} = 1$ for $i \in [n]$. Furthermore, every system with variables constrained over a finite domain $D$ is explicitly Archimedean, as shown in \cref{lemma-new-t}.

\subsection*{Ideals and varieties}
Let us recall here the notions of (polynomial) ideals, (algebraic) varieties and some of their properties (see e.g. \cite{Cox}).
Consider a set of polynomials $\mathcal{P} = \{p_1,\ldots,p_m\} \subseteq \mathbb{R}[x_1, \dots, x_n]$. 
\begin{definition}
    The \textit{algebraic variety generated by $\mathcal{P}$} is defined as
    \begin{align*}
        S:= \Variety{\mathcal{P}} = \{x \in \mathbb{R}^n \, | \, p_i(x) = 0, \, \forall i \in [m] \}.
    \end{align*}
\end{definition}
We also define the following two sets.
\begin{definition}
        \begin{align*}
        \langle \mathcal{P} \rangle &:= \{q \in \mathbb{R}[x_1, \dots, x_n] \ | \ q = \sum_i h_i p_i, \ h_i \in \mathbb{R}[x_1, \dots, x_n]\}, \\
        \I(S) &:= \{ q \in \mathbb{R}[x_1, \dots, x_n] \ | \ q(\alpha) = 0, \ \forall \alpha \in S\},
    \end{align*}
    as the \emph{ideal generated by $\mathcal{P}$}, the former, and the \emph{(vanishing) ideal generated by set $S$}, the latter. We refer to a \emph{$d$-truncated ideal} when we consider $I_d := I \cap \mathbb{R}[x_1, \dots, x_n]_d$, where $I \subseteq \mathbb{R}[x_1, \dots, x_n]$ is a polynomial ideal.
\end{definition}

\begin{definition}\label{def:radical}
    We say that an ideal $\I$ is \emph{radical} if $p^m \in \I$ for some $m \in \mathbb{N}$ implies $p \in \I$.
\end{definition}

\subsection*{\GB bases}

We give here a very brief introduction on the notion of $\GB$ basis. For a complete exposition, we refer the reader to \cite{Cox}.

We first establish an order on the polynomial ring $\mathbb{R}[x_1, \ldots, x_n]$. Given a monomial $x^{\alpha} = x_1^{\alpha_1} \dots x_n^{\alpha_n}$, this can be unambiguously associated to the $n$-tuple $\alpha = (\alpha_1, \ldots, \alpha_n) \in \mathbb{N}^n$.

\begin{definition}\label{def:lex and grlex} Let $\alpha =(\alpha_1,\ldots,\alpha_n),\beta=(\beta_1,\ldots,\beta_n)\in \mathbb{N}^n$ and $|\alpha| = \sum_{i=1}^n\alpha_i$, $|\beta| = ~\sum_{i=1}^n\beta_i$.
  \begin{enumerate}[(i)]
      \item Lexicographic order (\lexns): We say $\alpha>_\lex \beta$ if, in the vector difference $\alpha -\beta \in \Zz^n$, the left most nonzero entry is positive. 
      \item Graded lexicographic order (\grlexns): We say $\alpha>_\grlex \beta$ if $|\alpha| >|\beta|$, or $|\alpha| =|\beta|$ and $\alpha>_\lex \beta$.
  \end{enumerate}
\end{definition}

Throughout this paper we will always assume that $\mathbb{R}[x_1,\ldots,x_n]$ is ordered according to the \textit{graded lexicographic order} \grlexns.

One potential approach to solving the $\IMP$ is through polynomial division. The idea is that, given a polynomial $r$ and a set of polynomials $\mathcal{P}$, if the remainder of the division of $r$ by $\mathcal{P}$ is zero, then $r$ belongs to the ideal $\langle \mathcal{P} \rangle$. However, it is well-known that polynomial division is generally not well-defined. Specifically, the remainder resulting from the division of $r$ by $\mathcal{P}$ can vary depending on the order in which the polynomials in $\mathcal{P}$ are used for division.

To fix this issue, a special set of generators was introduced in \cite{BuchbergerThesis}.
\begin{definition}[\GB Basis]
    Let $\mathcal{G} = \{g_1, \ldots, g_s\}$ be a set of polynomials. Consider an ideal $\I \subseteq \mathbb{R}[x_1, \ldots, x_n]$ such that $\I = \langle g_1, \ldots, g_s \rangle$, and consider $r \in \mathbb{R}[x_1, \ldots, x_n]$. We say that $\mathcal{G}$ is a \emph{\GB basis of $\I$} if the following property holds
    \begin{align*}
        r \in \I \iff r|_{\mathcal{G}} = 0,
    \end{align*}
where $r|_{\mathcal{G}}$ is the remainder of the polynomial division of $r$ by $\mathcal{G}$.
\end{definition}
Moreover, we will be mainly interested in solving problems of the form $\IMP_d$, i.e. when $r$ has degree at most $d$. Because $\mathbb{R}[x_1,\ldots,x_n]$ is ordered according to the \grlex order, the only polynomials in $\mathcal{G}$ that can divide $r$ are those with degree $d$ or lower. Consequently, we give the following definition.

\begin{definition}\label{def:dTruncated GB}
  Let $\mathcal{G}$ be a \GB basis of an ideal in $\I \in \mathbb{R}[x_1,\dots,x_n]$, the \emph{d-truncated \GB basis} $\mathcal{G}_d$ of $\I$ is defined as
  \begin{equation}\label{eq:Gd}
      \mathcal{G}_d := \mathcal{G} \cap \mathbb{R}[x_1,\dots,x_n]_d.
  \end{equation}
\end{definition}

By the definition of \GB basis, we immediately conclude that $\IMP_d$ can be solved when the $d$-truncated \GB basis is available. Specifically, we have
\begin{align*}
    r \in \I_d \iff r|_{\mathcal{G}_d} = 0.
\end{align*}
Furthermore, if $\mathcal{G}_d$ can be calculated in polynomial time, then the $\IMP_d$ can be solved in polynomial time (see also \cref{sect:PC_bit}).

\subsection*{Polynomial Calculus over general finite domains}

In this paper we consider $\PC$ over a general finite domain $D$, which is an immediate generalization of the classical $\PC$ over the Boolean domain \cite{CleggEI96}, i.e. when the Boolean constraints $x_j^2 - x_j = 0$ belong to the set of constraints. Let $D = \{\rho_1, \rho_2, \ldots, \rho_k\}$ be a finite domain. For every variable $x_j$ where $j \in [n]$, to enforce $x_j$ to assume values in $D$, we include the following univariate \emph{domain polynomials} in $\mathcal{P}$
\begin{align*}
    D_{k}(x_j)&=(x_j-\rho_{1})(x_j-\rho_{2})\cdots (x_j-\rho_{k}) \qquad j \in [n].
\end{align*}
$\PC$ over a general finite domain is a proof system that consists of the following derivation rules for polynomial equations $(f = 0), (g = 0) \in  \mathcal{P}$, domain polynomial equations $(D_{k}(x_j)=0)$, variable $x_j$, and numbers $a, b \in \mathbb{R}$
\begin{equation}\label{eq:PCrules}
    \frac{}{f=0} \qquad \frac{}{D_{k}(x_j)=0} \qquad \frac{f=0 \qquad g=0}{a f + b g =0} \qquad \frac{f=0}{x_jf=0} 
\end{equation}

\begin{definition}\label{def:pc_derivation}
    A \emph{$\PC$ derivation} (or \emph{$\PC$ proof}) of $``r=0"$ from $\mathcal{P}$ is a sequence $(r_1=0,\ldots,r_L=0)$ of polynomial equations iteratively derived by using \eqref{eq:PCrules} with $r=r_L$. The \emph{size} of a derivation is the sum of the sizes of the binary encoding of the polynomials in the derivation and the \emph{degree} is the maximum degree of the polynomials in the derivation. A \emph{PC refutation} is a derivation of $``1=0"$.
\end{definition}

\section[SoS epsilon criterion]{$\sos_\varepsilon$ criterion}\label{sect:sos-criterion}
As outlined in \cref{sect:introduction}, we will examine \emph{sufficient} conditions on a polynomial system \(\mathcal{P}\) for ensuring that the following property holds.
\begin{itemize}[label=$(\rm{P})$]
    \item \label{prob:1} Assume there exists an $\sos$ proof of $``r \geq 0"$ from $\mathcal{P}$ of degree $2d$. Then, for every $\varepsilon>0$, there also exists an $\sos$ proof of $``r +\varepsilon \geq 0"$ from $\mathcal{P}$ with degree $O(d)$ and coefficients bounded by $2^{poly(n^d,\lg \frac{1}{\varepsilon})}$.
\end{itemize}

In this section we formulate the \emph{$\sos_{\varepsilon}$ criterion}, a set of sufficient conditions that guarantee that property (\textsc{P}) holds (see \cref{sect:sos_criterion_subsection}, in particular \cref{th:SoS_Criterion}). The $\sos_{\varepsilon}$ criterion has two requirements: $\delta$-spectrality and $\SOSe$ness. We then proceed to give general settings and techniques to verify that the requirements are satisfied (see \cref{sect:delta_spectrality} and \cref{sect:SOS_completeness}). Finally, we discuss the separation between the \Nsatz\ criterion and the \sos\ criterion (see \cref{sect:separation_Nsatz_SOS,sect:pc_crit,sect:applications}).

\subsection[SoS epsilon criterion]{$\sos_{\varepsilon}$ criterion}\label{sect:sos_criterion_subsection}
Recall we are assuming that $S$ is finite and that $\| S \| < 2^{poly(n)}$. The \emph{moment} matrix is defined as follows.
\begin{align}
M = M_{S,d} := \mathbb{E}_{\alpha \in S} [\mathbf{v}_d (\alpha) \mathbf{v}_d^{\sf{T}}(\alpha)] = \frac{1}{|S|}\sum_{\alpha\in S}\mathbf{v}_d (\alpha) \mathbf{v}_d^{\sf{T}}(\alpha),
\end{align}
where the expectation is over the uniform distribution over $S$. Note that $M$ is positive semidefinite, i.e. it is a real symmetric matrix with nonnegative eigenvalues. Let $\lambda_1$, $\lambda_2$, $\dots, \lambda_{\binom{n+d}{d}} $ be the eigenvalues of $M_{s,d}$ with corresponding eigenvectors $u_1, \dots, u_{\binom{n+d}{d}} $ forming an orthonormal basis for $\mathbb{R}^{\binom{n+d}{d}}$. Let $U$ be the matrix where the columns are the eigenvectors $u_1,\ldots,u_{\binom{n+d}{d}}$. 
Now, we define 
\begin{align}\label{Pi}
    \Pi^+ := \sum_{i \text{ s.t. } \lambda_i >0} u_i u_i^{\sf{T}}, \qquad \Pi^0 := \sum_{i \text{ s.t. } \lambda_i =0} u_i u_i^{\sf{T}}. 
\end{align}
    Then, we have 
    \begin{equation*}
        I = U^{\sf{T}} U \qquad \text{and} \qquad I = \Pi^+ + \Pi^0.
    \end{equation*}
We have the following lemma.
\begin{lemma}\label{th:eigen-to-ideal}
Let $u$ be an eigenvector for the zero eigenvalue $\lambda=0$ of $M$. Then, we have $u^{\sf{T}}\mathbf{v}_d\in \I_d(S)$.
\end{lemma} 
\begin{proof}
    By the assumptions, we have
    \begin{equation*}
        0=u^{\sf{T}} M u = u^{\sf{T}}\left(\frac{1}{|S|}\sum_{\alpha\in S}\mathbf{v}_d (\alpha) \mathbf{v}_d^{\sf{T}}(\alpha)\right)u=\frac{1}{|S|}\sum_{\alpha\in S}(u^{\sf{T}}\mathbf{v}_d(\alpha))^2.
    \end{equation*}
Since all terms on the right-hand side are nonnegative, all are equal to zero. That is, the polynomial $u^{\sf{T}} \mathbf{v}_d$ vanishes at all points in $S$.   
\end{proof}
  \ignore{  where the first equality follows from the definition of orthogonality. To see the second, consider a vector $v$. Since $u_1, \ldots, u_n$ form an orthonormal basis, then there exist $\alpha_1,\ldots, \alpha_n \in \mathbb{R}^n$ such that
    \begin{equation*}
        v = \alpha_1 u_1 + \ldots + \alpha_n u_n.
    \end{equation*}
    Thus, 
    \begin{align*}
        (\Pi^+ + \Pi^0)v &= \sum_i u_i u_i^{\sf{T}} v \\
        &= u_i u_i^{\sf{T}} (\alpha_1 u_1 + \ldots + \alpha_n u_n) = \sum_i \alpha_i u_i = v.
    \end{align*}

We will focus just on discrete sets. We observe that the eigenvectors for the eigenvalue $\lambda=0$ define polynomials that vanish on the algebraic variety. (\textcolor{red}{Do we really need the second part?)}
\begin{lemma}\label{th:eigen-to-ideal}
    Let $u_i(x) := \langle u_i, \mathbf{v}_d \rangle$, then
    \begin{itemize}
        \item For all $i$ s.t. $\lambda_i = 0$, we have $u_i(x) \in I(S)$;
        \item For all $i$ s.t. $\lambda_i > 0$, we have $u_i(x) \notin I(S)$.
    \end{itemize}
\end{lemma}

\begin{proof}
    Firstly observe that:
    \begin{equation*}
        u_i^{\sf{T}} M u_i = \mathbb{E}_{\alpha \in S} [u_i^{\sf{T}} \mathbf{v}_d(\alpha) \mathbf{v}_d^{\sf{T}}(\alpha) u_i] = \mathbb{E}_{\alpha \in S} [u_i(\alpha)^2] = \frac{1}{|S|} \sum_{\alpha \in S} u_i(\alpha)^2,
    \end{equation*}
    where the equalities follow from the definitions of $M$ and of $u_i(x)$.\\
    Consider the two cases:
    \begin{itemize}
        \item Assume $i$ s.t. $\lambda_i = 0$. Since $u_i$ is a zero eigenvector, then
            \begin{equation*}
                u_i^{\sf{T}} M u_i = 0.
            \end{equation*}
            Since for every $\alpha \in S$ we have $u_i(\alpha)^2 \geq 0$, it follows that $u_i(\alpha) = 0$;
        \item Assume $i$ s.t. $\lambda_i > 0$. Since $u_i$ is \emph{not} a zero eigenvector, then
            \begin{equation*}
                u_i^{\sf{T}} M u_i = \lambda_i > 0.
            \end{equation*}
            Since for every $\alpha \in S$ we have $u_i(\alpha)^2 \geq 0$, it follows that there exists $\beta \in S$ s.t. $u_i(\beta) > 0$.
    \end{itemize}
\end{proof}

\begin{corol}
    If $\langle \mathcal{P} \rangle$ is radical, then also
    \begin{itemize}
        \item For all $i$ s.t. $\lambda_i = 0$, we have $u_i(x) \in \langle \mathcal{P} \rangle$.
        \item For all $i$ s.t. $\lambda_i > 0$, we have $u_i(x) \notin \langle \mathcal{P} \rangle$.
    \end{itemize}
\end{corol}
}

\begin{definition}[$\delta$-spectrality and $\sos_\varepsilon$-completeness]\label{def:delta+soscomplete}
    Let $\mathcal{P}=\{p_1=0, \dots p_m=0\}$ be polynomial system with variety $S = \Variety{\mathcal{P}}$, and let $\delta \in \mathbb{R}_{>0}$. 
    \begin{enumerate}
        \item We say that \emph{$S$ is $\delta$-spectrally rich up to degree $d$} if every nonzero eigenvalue of $M$ is at least~$\delta$.
        \item We say that \emph{$\mathcal{P}$ is $\sos_\varepsilon$-$d$-complete over $S$} (or simply $\sos_\varepsilon$-complete) if, for every polynomial in the $2d$-truncated vanishing ideal $q \in \I_{2d}(S)$ and every $\varepsilon > 0$, there exists an $\sos$ proof of $``q + \varepsilon \geq 0"$ of degree $O(d)$ from $\mathcal{P}$ with absolute value of the coefficients bounded by $2^{poly(n^d, \, \lg \| q \|_{\infty}, \lg \frac{1}{\varepsilon})}$.
    \end{enumerate}
\end{definition}

Let $r$ be a polynomial. Suppose that $``r \geq 0"$ admits an $\sos$ proof $\Pi$ from $\mathcal{P}$. Although $\Pi$ may, in principle, have coefficients with magnitude of the order $2^{2^n}$ (see \cite{odonnell2017,raghavendra_weitz2017}), we show in the next result that $r$ can be decomposed as a sum-of-squares component plus an ``ideal part" component, both having bounded coefficients. In general, the ideal part does not necessarily take the form $\sum h_i p_i$ as in \cref{def:SOS_proof}. Nevertheless, this result allows us to reduce the discussion to focus on the ideal part of the decomposition. The following lemma, essentially from Raghavendra and Weitz~\cite{raghavendra_weitz2017}, is presented here separately as we use it to extend their result.

\begin{lemma}\label{th:SOS_decomposition}
    Consider the a polynomial system $\mathcal{P} = \{p_1 = 0, \ldots, p_m = 0\}$ with finite variety $S = \Variety{\mathcal{P}}$ such that $\| S \| \leq 2^{poly(n^d)}$. Assume that $S$ is $\delta$-spectrally rich up to degree $d$. \\
    Let $r$ be a polynomial nonnegative on $S$ with coefficients bounded by $2^{poly(n^d)}$. If there exists an $\sos$ proof of $``r \geq 0"$ from $\mathcal{P}$ with degree at most $2d$
    \begin{equation}\label{original-SOS}
        r = \sum_{i=1}^{t_0} q_i^2 + \sum_{i=1}^m h_i p_i,
    \end{equation}
    for some $s_i, h_i \in \mathbb{R}[x_1, \ldots, x_n]$. Then, we have
    \begin{equation}\label{eqn:SoS_decomposition}
        r = \sum_{i=1}^{t_1} s_i^2 + p,
    \end{equation}
    for some polynomials $s_i$ of degree at most $d$, some $p\in \I_{2d}(S)$ and with all the coefficients on the right-hand side of \eqref{eqn:SoS_decomposition} bounded by $2^{poly(n^d, \lg \frac{1}{\delta})}$.
\end{lemma}

\begin{proof}
    First, we have that
    \begin{equation}
        \sum_{i=1}^{t_0} q_i^2 = \langle C, \mathbf{v}_d \mathbf{v}_d^{\sf{T}} \rangle,
    \end{equation}
    for some positive semidefinite matrix $C$.
    
    Recall the matrices $\Pi^0$ and $\Pi^+$ from Equations (\ref{Pi}). Observe that $\Pi^0 \mathbf{v}_d \mathbf{v}_d^{\sf{T}} = \sum_{\lambda_i = 0} (u_i^{\sf{T}} \mathbf{v}_d ) u_i \mathbf{v}_d^{\sf{T}}$. With this in mind, we can decompose the matrix $\mathbf{v}_d \mathbf{v}_d^{\sf{T}}$ into its projections as follows.
    \begin{align*}
        \langle C, \mathbf{v}_d \mathbf{v}_d^{\sf{T}} \rangle &= \langle C, (\Pi^0 + \Pi^+) \mathbf{v}_d \mathbf{v}_d^{\sf{T}} (\Pi^0 + \Pi^+) \rangle \\
        &= \langle C, \Pi^+ \mathbf{v}_d \mathbf{v}_d^{\sf{T}} \Pi^+ \rangle + \sum_{\lambda_i = 0} {u_i}^{\sf{T}} \mathbf{v}_d \langle C, \Pi^+ \mathbf{v}_d {u_i}^{\sf{T}} + {u_i} \mathbf{v}_d^{\sf{T}} \Pi^+ + {u_i} \mathbf{v}_d^{\sf{T}} \Pi^0 \rangle\\
        & = \langle \Pi^+ C \Pi^+, \mathbf{v}_d \mathbf{v}_d^{\sf{T}} \rangle + P,
    \end{align*}
    where we have set $P := \sum_{\lambda_i = 0} {u_i}^{\sf{T}} \mathbf{v}_d \langle C, \Pi^+ \mathbf{v}_d {u_i}^{\sf{T}} + {u_i} \mathbf{v}_d^{\sf{T}} \Pi^+ + {u_i} \mathbf{v}_d^{\sf{T}} \Pi^0 \rangle \in \I_{2d}(S)$. The polynomial $\langle \Pi^+ C \Pi^+, \mathbf{v}_d \mathbf{v}_d^{\sf{T}} \rangle$ is a sum of squares as the matrix $C':= \Pi^+ C \Pi^+$ is positive semidefinite. We write $\sum_{i=1}^{t_1}s_i^2 = \langle  C', \mathbf{v}_d \mathbf{v}_d^{\sf{T}} \rangle $. We now observe that the entries of $C'$ are bounded, thus showing that the coefficients of $s_i$ are bounded. For this, take the expected value of both sides of \eqref{original-SOS} and note that
    \begin{align*}
        poly(\| r \|_{\infty}, \| S \|)=\mathbb{E}_{\alpha \in S} [r(\alpha)] &= \langle C', M \rangle = \langle U^{\sf{T}} C' U, \Lambda \rangle = \sum_i u_i^{\sf{T}} C' u_i \lambda_i\\
        &\geq tr(U^{\sf{T}} C' U) \delta,
    \end{align*}
    where $M = U^{\sf{T}} \Lambda U$ is the spectral decomposition of $M$, $\Lambda = diag(\lambda_1, \lambda_2, \dots, \lambda_{\binom{n+d}{d}})$ is the diagonal matrix of eigenvalues of $M$, and we used that the $(i,j)$-th element of $U^{\sf{T}} C' U$ is $u_i^{\sf{T}} C' u_j$. In addition, we have $tr(U^{\sf{T}} C' U) = \langle U^{\sf{T}} C' U, \mathrm{I} \rangle = \langle C', U^{\sf{T}} U \rangle =  \langle C', \mathrm{I} \rangle = tr(C')$, where $\mathrm{I}$ is the identity matrix. Thus
    \begin{align*}
        poly(\| r \|_{\infty}, \| S \|) = \mathbb{E}_{\alpha \in S} [r(\alpha)] \geq tr(U^{\sf{T}} C' U) \delta = tr(C') \delta.
    \end{align*}
    It follows that we can give a polynomial upper bound to the size of $C'$. Indeed, for every entry of $C'$, we have
    \begin{equation*}
        |C'_{ij}| \leq tr(C') \leq \frac{\mathbb{E}_{\alpha \in S} [r(\alpha)]}{\delta} \leq \frac{2^{poly(n^d)}}{\delta} = 2^{poly(n^d, \lg \frac{1}{\delta})}.
    \end{equation*}
    Finally, observe that we have
    $$r=\sum_{i=1}^{t_1}s_i^2 + P + \sum_{i=1}^mh_ip_i.$$
    We define $p:= P + \sum_{i=1}^mh_ip_i$, and we observe that $p\in \I_{2d}(S)$.  We conclude that the coefficients of $p$ are bounded by $2^{poly(n^d, \lg \frac{1}{\delta})}$, since $p=r-\sum_{i=1}^{t_1}s_i^2$.
\end{proof}

Note that if $\mathcal{P}$ is $\Nsatz$ $d$-complete, then the identity $p=\sum h_i p_i$ for the ``ideal part" can be computed efficiently by the $\Nsatz$ proof system, i.e. with degree at most $O(d)$ and coefficients bounded by $2^{poly(n^d, \lg \frac{1}{\delta})}$. This is the idea behind the $\Nsatz$ criterion. However, this is sufficient but not a necessary condition (see \cite{raghavendra_weitz2017} for a further discussion on the limitations of the $\Nsatz$ criterion).

We next present a new criterion called \textit{$\sos_{\varepsilon}$ criterion}. In essence, the $\sos_{\varepsilon}$ criterion requires that any degree $2d$ polynomial from the ideal part can be $\sos$ proven efficiently to be nonnegative (up to an additive error $\varepsilon$).  This replaces the requirement that there exists $\Nsatz$ proofs of the the ideal part in the $\Nsatz$ criterion. Since $\sos$ is stronger than $\Nsatz$ as a proof system, it follows that the $\sos_{\varepsilon}$ criterion extends and generalizes the $\Nsatz$ criterion. In the following, we will provide natural examples of separation between the two criteria (see \cref{sect:separation_Nsatz_SOS}).

\begin{theorem}[$\sos_{\varepsilon}$ criterion]\label{th:SoS_Criterion}
    Consider a polynomial system $\mathcal{P}=\{p_1 = 0, \dots, p_m = 0\}$ with finite variety $S = \Variety{\mathcal{P}}$ such that $\| S \| \leq 2^{poly(n^d)}$.
    Assume that (see \cref{def:delta+soscomplete})
    \begin{itemize}
        \item [1)] $S$ is $\delta$-spectrally rich up to degree $d$, and
        \item [2)] $\mathcal{P}$ is $\sos_\varepsilon$-complete over $S$.
    \end{itemize}
    Let $r$ be a polynomial. If $``r \geq 0"$ has a degree $2d$ $\sos$ proof
    \begin{equation*}
        r = \sum_{i=1}^{t_0} \sigma_i^2 + \sum_{i=1}^m h_i p_i,
    \end{equation*}
    then, for every $\varepsilon>0$, there exists an $\sos$ proof of $``r+\varepsilon \geq 0"$ of degree $O(d)$
    \begin{equation}\label{eqn:SOS_criterion}
        r +\varepsilon = \sum_{i=1}^{t} \tilde{\sigma}_i^2 + \sum_{i=1}^m \tilde{h}_i p_i,
    \end{equation}
    such that the coefficients of every polynomial appearing in the proof are bounded by $2^{poly(n^d, \lg \frac{1}{\delta}, \lg \frac{1}{\varepsilon})}$.
\end{theorem}

\begin{proof}
    By Lemma \ref{th:SOS_decomposition}, as $S$ is $\delta$-spectrally rich, we can rewrite the proof as 
     \begin{equation*}
            r = \sum_{i=1}^{t_1} s_i^2 + p,
        \end{equation*}
      where the coefficients of $s_i$, and $p$ are bounded by $2^{poly(n^d, \lg\frac{1}{\delta})}$, and $p\in \I_{2d}(S)$. Let $\varepsilon > 0$. Since $\mathcal{P}$ is $\sos_\varepsilon$-complete over $S$, there exists a $\sos$ proof of degree $O(d)$
      \begin{equation*}
          p+ \varepsilon = \sum_{i=1}^{t_2}s_i'^{2} + \sum_{i=1}^{m}h_i' p_i
      \end{equation*}
      with coefficients bounded by $2^{poly(n^d , \lg \frac{1}{\varepsilon})}$. By combining these two proofs we obtain the desired result.
\end{proof}

\begin{corollary}
    Suppose $\mathcal{P}$ satisfies the assumptions of the $\sos_\varepsilon$ criterion with $\frac{1}{\delta}=2^{poly(n^d)}$. Assume, moreover, that $\mathcal{P}$ is explicitly Archimedean. If $``r\geq 0"$ has an $\sos$ proof of degree $2d$ from $\mathcal{P}$, then $``r + \varepsilon \geq 0"$ has an $\sos$ proof of degree $O(d)$ from $\mathcal{P}$ that can be computed in time $poly(n^d, \lg \frac{1}{\varepsilon})$, up to any additive error $\varepsilon$.
\end{corollary}

\subsection[Delta-spectrality]{$\delta$-spectrality}\label{sect:delta_spectrality}

The $\delta$-spectrality (\cref{def:delta+soscomplete}) hypothesis in \cref{th:SoS_Criterion} is, to some extent, a mild hypothesis. It is satisfied by many interesting instances where the variety $S$ is discrete. For example, it is satisfied by combinatorial problems having varieties contained in the Boolean hypercube $\{0,1\}^n$. To see this, we first state a lemma for integer-valued matrices.

\begin{lemma}[\cite{raghavendra_weitz2017}]\label{th:spectrality_integer_matrix}
    Let $M \in \mathbb{S}^{N \times N}$ be an integer matrix with $|M_{ij}| \leq B$ for all $i,j \in [N]$. The smallest non-zero eigenvalue of $M$ is at least $(BN)^{-N}$. 
\end{lemma}

By observing that $|S| \cdot M_{S,d}$ is a $O(n^d) \times O(n^d)$ integer-valued matrix, we immediately get $\delta$-spectrality over integer-valued varieties.

\begin{corollary}[\cite{raghavendra_weitz2017}]
    Let $\mathcal{P}$ be a polynomial system such that $S \subseteq \mathbb{Z}^n$ such that $\| S \| < 2^{poly(n^d)}$. Then $S$ is $\delta$-spectrally rich with $\frac{1}{\delta} = 2^{poly(n^d)}$.
\end{corollary}

Another wide class of polynomial problems for which $\delta$-spectrality is easily satisfied are the polynomial systems $\mathcal{P}$ with \textit{variables constrained over a finite domain $D$}. Assuming $D =\{\rho_1, \dots, \rho_k\} \subseteq \mathbb{Q}$ with constant $k=O(1)$, these systems are described as containing the domain polynomials $(x_i - \rho_1)(x_i - \rho_2) \cdot \dots \cdot (x_i - \rho_k)$ for each variables $x_i$.

\begin{corollary}\label{cor:rich-finite}
    Let $\mathcal{P}$ be a polynomial system with variables constrained over a finite domain $D \subseteq \mathbb{Q}$, i.e. $S \subseteq D^n$. Then $S$ is $\delta$-spectrally rich (up to degree $d$) for some $\frac{1}{\delta} = 2^{poly(n^d)}$ 
\end{corollary}

\begin{proof}
    Note that
    \begin{equation*}
       \Pi_{i=1}^k \rho_i^d \cdot |S| \cdot M_{S,d}
    \end{equation*}
    is an integer matrix with values bounded by $2^{poly(n^d)}$. The result follows from \cref{th:spectrality_integer_matrix}.
\end{proof}

\subsection[SoS epsilon completeness]{$\sos_\varepsilon$-completeness}\label{sect:SOS_completeness}
In this section we develop tools for showing that a polynomial system is $\SOSe$.
We will consider multiple polynomial systems $\mathcal{Q}$ preserving geometric and bit complexity properties of $\mathcal{P}$, namely
\begin{enumerate}[label=A\arabic*.]
    \item \label{assumption_1} \textit{Same zero set:} $S = \Variety{\mathcal{P}} = \Variety{\mathcal{Q}}$.
    \item \label{assumption_2} \textit{Same degree order:} $\deg(q) = O(d)$, $ \forall q \in \mathcal{Q}$, where $d$ is the maximum degree of the polynomials in~$\mathcal{P}$.
    \item \label{assumption_3} \textit{Polynomial bit complexity:} the bit complexity for representing $\mathcal{Q}$ is polynomial in $n$. Note that this implies that the cardinality of $\mathcal{Q}$ is polynomially bounded, i.e. $|\mathcal{Q}| = poly(n)$, and that all coefficients of the polynomials in $\mathcal{Q}$ are bounded by $2^{poly(n^d)}$.
\end{enumerate}
\paragraph{$\sos$-approximability.}
We define the relation of \emph{$\sos$-approximability} between polynomial systems with the same zero set. This relation arises by considering $\sos$ proofs of approximate objective polynomials $p + \varepsilon$. Therefore, it cannot be simulated by the $\Nsatz$ proof system. We will see that $\sos$-approximability is a powerful tool for showing that a polynomial system $\mathcal{P}$ is $\SOSe$.
\begin{definition}[$\sos$ approximation]\label{def:sos_approx}
     Let $\mathcal{P}=\{p_1 = 0,p_2=0, \dots, p_m=0\}$ and $\mathcal{P}'=\{p_1'=0, p_2'=0, \dots, p_l'=0\}$ be two polynomial systems such that $\Variety{\mathcal{P}} = \Variety{\mathcal{P}'}$. We say that \emph{$\mathcal{P}'$ $\sos$-approximates $\mathcal{P}$}, and it is denoted by $\mathcal{P} \lesssim_{SoS} \mathcal{P}'$, if for every $p \in \mathcal{P}$ and every $\varepsilon >0$ there exist $\sos$ proofs
    \begin{align}
      \begin{split}
    ``p + \varepsilon \geq 0&"  \text{ from }\mathcal{P}' \text{ and}\\
    ``-p + \varepsilon \geq 0&" \text{ from } \mathcal{P}'     
      \end{split}
    \end{align} 
    with degree $O(d)$ and with coefficients bounded by $2^{poly(n^d, \lg \frac{1}{\varepsilon})}$.
\end{definition}

Next, we introduce a property that will prove valuable throughout the rest of this section. Roughly speaking, we will show that, under the assumption of explicit Archimedeanity, if $\sos$ can (approximately) prove $``p=0"$ in a precise sense, then it can (approximately) prove the product $``gp \geq 0"$ for any polynomial $g$.

\begin{lemma}\label{lemma-pg-e}
    Let $\mathcal{P}$ be an explicitly Archimedean polynomial system. Let $p \in \mathbb{R}[x_1,\ldots,x_n]$ be a polynomial of degree (at most) $2d$ with coefficient norm bounded by $2^{poly(n^d)}$.
    Assume that, for every $\varepsilon>0$, we have $\sos$ proofs of degree $2d$ from $\mathcal{P}$ of
    \begin{align}
      \begin{split}
          ``p + \varepsilon &\geq 0", \text{ and of} \\
          ``-p + \varepsilon &\geq 0",      
      \end{split}
    \end{align} 
    with coefficients bounded by $2^{poly(n^d, \lg \frac{1}{\varepsilon})}$. Then, for every $\varepsilon>0$ and every $g \in \mathbb{R}[x_1 \dots, x_n]$ with $\deg(g) = O(d)$ and $\| g \|_{\infty} < 2^{poly(n^d)}$, there exists an $\sos$ proof from $\mathcal{P}$ of  
    $$`` pg + \varepsilon \geq 0",$$
    of degree $O(d)$ with coefficients bounded by $2^{poly(n^d, \lg \frac{1}{\varepsilon})}$.
\end{lemma}

\begin{proof}
Since $\mathcal{P}$ is explicitly Archimedean, there exists a number $0 < N_{g^2+1} < 2^{poly(n^d)}$ such that there exists a proof $\Pi$ of
\begin{align}\label{Ng}
    ``N_{g^2+1} - (g^2+1) \geq 0"
\end{align}
of degree $O(d)$ and coefficients bounded by $2^{poly(n^d)}$. Let $\varepsilon>0$. We set $\varepsilon':= 2\varepsilon/N_{g^2+1}$.  Observe that the following identity holds:
    \begin{equation}\label{identity-fg}
        p g + \frac{\varepsilon'}{2}(g^2 + 1) = (p+\varepsilon')\left( \frac{g+1}{2} \right)^2 + (- p+\varepsilon' )\left( \frac{g-1}{2} \right)^2.
    \end{equation}
After combining the last identity with the proof $\Pi$ in (\ref{Ng}) (multiplied by $\frac{\varepsilon'}{2}$), we obtain  
    \begin{equation}\label{QM-e}
        pg + \frac{\varepsilon'}{2}N_{g^2+1} = \frac{\varepsilon'}{2}\Pi + (p+\varepsilon' )\left( \frac{g+1}{2} \right)^2 + (-p+\varepsilon')\left( \frac{g-1}{2} \right)^2.
    \end{equation}
 By hypothesis, there exist $\sos$ proofs  $``p + \varepsilon' \geq 0"$ and  
      $``-p + \varepsilon' \geq 0"$ of degree $2d$ and coeffcients bounded by $2^{poly(n^d,\lg\frac{1}{\varepsilon'})}$. By the definition of $\varepsilon'$ this is also bounded by $2^{poly(n^d,\lg\frac{1}{\varepsilon})}$. Then, these two proofs combined with Equation~(\ref{QM-e}) give an $\sos$ proof of $``pg+\varepsilon\geq 0"$ of degree $O(d)$, with coefficients bounded by $2^{poly(n^d, \lg\frac{1}{\varepsilon})}$, as desired.
\end{proof}

With the notions of $\sos$-approximability and \cref{lemma-pg-e} at hand, we begin by demonstrating an interesting property of the relation $\lesssim_{\sos}$.

\begin{lemma}[Transitivity]\label{th:transitivity_of_approximations}
    Let $\mathcal{P}_1, \mathcal{P}_2$ and $\mathcal{P}_3$ be three systems of polynomials with zero set $S$. Assume that $\mathcal{P}_3$ is explicitly Archimedean. If $\mathcal{P}_1 \lesssim_{SoS} \mathcal{P}_2$ and $\mathcal{P}_2 \lesssim_{SoS} \mathcal{P}_3$, then $\mathcal{P}_1 \lesssim_{SoS} \mathcal{P}_3$.
\end{lemma}
\begin{proof}
     Let $\mathcal{P}_1=\{p_1, \dots, p_{m_1}\}$, $\mathcal{P}_2=\{q_1, \dots, q_{m_2}\}$ and let $\mathcal{P}_3=\{r_1, \dots, r_{m_3}\}$ and $\varepsilon>0$. For showing that $\mathcal{P}_1 \lesssim_{SoS} \mathcal{P}_3$, we have to show that there exists an $\sos$ proof of $``p_i+\varepsilon\geq 0"$, and $``-p_i+\varepsilon\geq0"$ (for $i\in [m_1])$ from $\mathcal{P}_3$ of degree $O(d)$, with coefficients bounded by $2^{poly(n^d,\lg\frac{1}{\varepsilon})}$. We will show this for $``p_1+\varepsilon\geq 0"$. The other polynomials are shown similarly. Since $\mathcal{P}_1 \lesssim_{SoS} \mathcal{P}_2$, there exists an $\sos$ proof
    \begin{align}\label{p_1}
        p_1 + \frac{\varepsilon}{2} = \sum_{i} s_i^2 + \sum_{j=1}^{m_2}q_jh_j,    
    \end{align}
   of degree $O(d)$, and coefficients bounded by $2^{poly(n^d, \lg \frac{1}{\varepsilon})}$. Since $\mathcal{P}_2 \lesssim_{SoS} \mathcal{P}_3$, for every $\varepsilon'>0$ and $j\in[m_2]$, we have $\sos$ proofs
   \begin{align}\label{q_i}
   \begin{split}
    ``q_j + \varepsilon' \geq 0&"  \text{ from }\mathcal{P}_3 \text{  and}\\
    ``-q_j + \varepsilon' \geq 0&" \text{ from } \mathcal{P}_3      
   \end{split}
   \end{align}
    of degree $O(d)$, and coefficients bounded by $2^{poly(n^d, \lg \frac{1}{\varepsilon})}$. Then, by Lemma \ref{lemma-pg-e} we obtain that, for all $j\in [m_2]$ and all $h_j$ of degree $O(d)$, there exists an $\sos$ proof $``q_jh_j + \frac{\varepsilon}{2m_2}\geq 0"$ from $\mathcal{P}_3$ of degree $O(d)$, and coefficients bounded by $2^{poly(n^d, \lg \frac{2m_2}{\varepsilon})}$. However, recall that by \ref{assumption_3} we have that $m_2 = poly(n)$ and thus the coefficients are bounded by $2^{poly(n^d, \lg \frac{1}{\varepsilon})}$. By summing up these proofs for all $j\in [m_2]$, we obtain an $\sos$ proof of
    \begin{align}\label{qh}
        ``\sum_{j=1}^{m_2}q_jh_j +\frac{\varepsilon}{2} \geq 0" \text{ from }\mathcal{P}_3
    \end{align}
   of degree $O(d)$, and coefficients bounded by $2^{poly(n^d, \lg \frac{1}{\varepsilon})}$. Finally, we combine the proofs in (\ref{qh}) and in (\ref{p_1}) and obtain an $\sos$ proof
    \begin{align}\label{p_1-final}
        ``p_1 + \varepsilon \geq 0" \text{ from }\mathcal{P}_3
    \end{align}
    of degree $O(d)$, and coefficients bounded by $2^{poly(n^d, \lg \frac{1}{\varepsilon})}$.  
\end{proof}

Next we present a few relevant examples for which it is possible to show $\sos$-approximability. The first example focuses on the powers of polynomial systems. Namely, let $\mathcal{P} = \{p_1 = 0, \ldots, p_m = 0\}$ be a polynomial system of equations. We define the \emph{$\alpha$-power of $\mathcal{P}$} as the polynomial system $\mathcal{P}^{\alpha} = \{p_1^{\alpha_1} = 0, p_2^{\alpha_2} = 0, \ldots, p_m^{\alpha_m} =0\}$, where $\alpha$ is a multi-index $\alpha = (\alpha_1,\alpha_2, \ldots, \alpha_m) \in \mathbb{N}^m$. The next result shows that $\alpha$-powers of a polynomial system approximate the set itself.

\begin{proposition}\label{prop:alpha_powers}
    Let $\alpha\in \mathbb{N}^n$, with $|\alpha|=O(d)$. Then, $\mathcal P\lesssim_{SoS}\mathcal{P}^\alpha$.
\end{proposition}

\begin{proof}
    Let $\mathcal{P} = \{p_1 = 0, \dots, p_m=0\}$ and $\varepsilon > 0$. We have to show that there exists an $\sos$ proof of $``p_i + \varepsilon \geq 0"$ and of $``-p_i + \varepsilon \geq 0"$ (for $i \in [m]$) from $\mathcal{P}^{\alpha}$ of degree $O(d)$, with coefficients bounded by $2^{poly(n^d, \lg \frac{1}{\varepsilon})}$. We start by showing this for $``p_1 + \varepsilon \geq 0$.

    Let $\ell = \lceil \lg \alpha_1 \rceil$ so that $2^{\ell} > \alpha_1$. Then there exists an $\sos$ proof of $``p_1 + \varepsilon \geq 0"$ from $p_1^{\alpha_1}$ of degree $O(d)$ and coefficients bounded by $2^{poly(n^d, \lg \frac{1}{\varepsilon})}$. Indeed,
    \begin{align*}
        p_1 + \varepsilon = \left(\sqrt{\frac{\varepsilon}{\ell}} + \frac{1}{2\sqrt{\frac{\varepsilon}{\ell}}} p_1 \right)^2 + \sum_{i=1}^{\ell-1} \left( \sqrt{\frac{\varepsilon}{\ell}} - \left( \frac{1}{2 \sqrt{\frac{\varepsilon}{\ell}}} \right)^{c_i} p_1^{2^i} \right)^2 - \left( \frac{1}{2 \sqrt{\frac{\varepsilon}{\ell}}} \right)^{2c_{\ell-1}} p_1^{2^\ell},
    \end{align*}
    where
    \begin{equation*}
        c_i = 
        \begin{cases}
            1 & i = 0, \\
            3 & i = 1, \\
            2c_{i-1} + 1 & otherwise.
        \end{cases}
    \end{equation*}
\end{proof}

In the second example, we show that when the polynomials in a system of polynomials are multiplied by positively shifted sums-of-squares, approximation is possible. More precisely, let $\mathcal{P}=\{p_1 = 0, \dots, p_m = 0\}$ be a polynomial system. Let $g\in \mathbb{R}[x_1, \dots, x_n]$ be a polynomial of the form
\begin{align}\label{sos+p}
    g = \sum_{i=1}^t q_i^2 + c,
\end{align}
for some constant $c > 0$. We consider the polynomial system $\mathcal{P}'=\{gp_1=0, \dots, p_m=0\}$. Clearly, $\mathcal{P}$ and $\mathcal{P}'$ have the same variety. We also make the assumptions \ref{assumption_2} and \ref{assumption_3} for set $\mathcal{P}'$. We have the following result.
\begin{proposition}
    Let $\mathcal{P}$ and $\mathcal{P}'$ as defined above. Then, we have $\mathcal{P} \lesssim_{SoS} \mathcal{P}'$.
\end{proposition}
\begin{proof}
Define $\sigma := \sum_{i=1}^t q_i^2$. Observe that the following identities hold: 
\begin{align*}
    p_1 + \varepsilon &= \frac{1}{4c^2 \varepsilon} \left[(\sigma p - 2\varepsilon c)^2 + c \sigma p^2 + (-\sigma p + 4\varepsilon c)pg \right], \\
    -p_1 + \varepsilon &= \frac{1}{4c^2 \varepsilon} \left[(\sigma p + 2\varepsilon c)^2 + c \sigma p^2 + (-\sigma p - 4\varepsilon c)pg \right].
\end{align*}
Therefore we have $\sos$ proofs of $``p_1 + \varepsilon \geq 0"$ and $``-p_1 + \varepsilon \geq 0"$ from $\mathcal{P}'$ of degree $O(d)$ and coefficients bounded by $2^{poly(n^d, \lg \frac{1}{\varepsilon})}$. For $i=2, \dots, m$ the polynomials $p_i+\varepsilon$ and $-p_i+\varepsilon$ are already $\sos$ proof from $\mathcal{P}'$.
\end{proof}

\paragraph{Showing $\SOSe$ness.}

The main consequence of the concept of $\sos$-approximability is that it allows for the inheritance of $\SOSe$ness among different polynomial systems.

\begin{theorem}\label{th:richness_inheritance}
    Let $\mathcal{P}_1$ and $\mathcal{P}_2$ be polynomials systems with zero set $S$. Assume that $\mathcal P_1$ is $\SOSe$ and that $\mathcal{P}_2$ is explicitly Archimedean. If  $\mathcal{P}_1\lesssim_{SoS} \mathcal{P}_2$, then $\mathcal{P}_2$ is $\SOSe$. 
\end{theorem}
\begin{proof}
Let $\mathcal{P}_1=\{p_1 =0, \cdots, p_{m_1}=0\}$ and $\mathcal{P}_2=\{q_1=0, \cdots, q_{m_2}=0\}$ be the two polynomials systems, let $\varepsilon>0$ be a real number and consider $p\in \I_{2d}(S)$. Since $\mathcal{P}_1$ is $\SOSe$, then there exists an $\sos$ proof 
\begin{align}\label{her-1}
    p+\frac{\varepsilon}{2} = \sum_{i=1}^{t} s_i^2 + \sum_{i=1}^{m_1}h_ip_i
\end{align}
of degree $O(d)$ and coefficients bounded by $2^{poly(n^d, \lg \frac{1}{\varepsilon})}$.
Since $\mathcal{P}_1\lesssim_{SoS} \mathcal{P}_2$ we have, for all $i\in [m_1]$, $\sos$ proofs of
\begin{align*}
``p_i+\frac{\varepsilon}{2m_1} \geq 0&" \text{ from } \mathcal{P}_2 \text{ and of} \\
``-p_i+\frac{\varepsilon}{2m_1} \geq 0&" \text{ from } \mathcal{P}_2
\end{align*}
of degree $O(d)$ and coefficients bounded by $2^{poly(n^d, \lg \frac{2m_1}{\varepsilon})}$. By Lemma \ref{lemma-pg-e}, for all $i\in [m_1]$ there exist an $\sos$ proof of $``p_ih_i + \frac{\varepsilon}{2m_1}\geq 0"$ from $\mathcal{P}_2$ of degree $O(d)$ and coefficients bounded by $2^{poly(n^d, \lg\frac{2m_1}{\varepsilon})}$. However, recall that by \ref{assumption_3} we have that $m_1 = poly(n)$ and thus the coefficients are bounded by $2^{poly(n^d, \lg \frac{1}{\varepsilon})}$. By summing up these proofs we obtain the $\sos$ proof
\begin{align}\label{her-2}
    ``\sum_{i=1}^{m_1}h_ip_i + \frac{\varepsilon}{2} \geq 0" \text{ from } \mathcal{P}_2
\end{align}
of degree $O(d)$ and coefficients bounded by $2^{poly(n^d, \lg\frac{1}{\varepsilon})}$. Finally, combining the proofs in (\ref{her-1}) and (\ref{her-2}), we obtain an $\sos$ proof of 
$$``p+\varepsilon \geq 0" \text{ from } \mathcal{P}_2$$
of degree $O(d)$ and coefficients bounded by $2^{poly(n^d, \lg\frac{1}{\varepsilon})}$.
\end{proof}

\begin{corollary}\label{cor:completeness_inheritance}
    Let $\mathcal{P}_1, \dots, \mathcal{P}_k$ be polynomial systems for some integer $k = O(1)$. Assume that $\mathcal{P}_1 \lesssim_{SoS} \dots \lesssim_{SoS} \mathcal{P}_k$. If $\mathcal{P}_1$ is $\sos_{\varepsilon}$-complete and $\mathcal{P}_k$ is explicitly Archimedean, then $\mathcal{P}_k$ is $\sos_{\varepsilon}$-complete.
\end{corollary}


It follows that the problem of showing that an explicitly Archimedean system $\mathcal{P}$ is $\SOSe$ can be reduced to identifying $\sos_\varepsilon$-complete polynomial systems $\mathcal{Q}$ such that $\mathcal{Q} \lesssim_{SoS} \mathcal{P}$. Interestingly, this can be achieved in various instances.

A broad class of such reductions arises from Gröbner basis theory. We recall that \GB bases completely characterize polynomial ideals. Specifically, for polynomial rings ordered by the \emph{\grlexns} order, every polynomial in the $2d$-truncated ideal $q \in \I_{2d}$ has remainder 0 when reduced by the set $\mathcal{G}_{2d}$ of elements of degree at most $2d$ of a \GB basis $\mathcal{G}$ (see e.g.~\cite{Cox}). Therefore, we have the following result.

\begin{lemma}\label{th:Grobner_basis_SOS_completeness}
    Let $\mathcal{P}$ be a polynomial system with $S = \Variety{\mathcal{P}}$. Let $\mathcal{G}_{2d}$ be a $2d$-truncated \GB basis of $I(S)$ according to the $\grlexns$ order. Assume that $\| \mathcal{G}_{2d} \|_{\infty} \leq 2^{poly(n^d)}$. Then $\mathcal{G}_{2d}$ is $\sos_{\varepsilon}$-complete.
\end{lemma}
\begin{proof}
    Let $q \in \I_{2d}(S)$ and $\varepsilon > 0$. By assumption we have that
    \begin{align*}
        q = \sum_{g \in \mathcal{G}_{2d}} h_g g,
    \end{align*}
    which is an $\sos$ proof of $``q \geq 0"$ from $\mathcal{G}_{2d}$.
    
    Moreover, all the polynomials $h_g g$ are quotients arising from the polynomial division $q|_{\mathcal{G}_{2d}}$. Therefore, $\deg(h_g g) \leq 2d$ since $\mathbb{R}[x_1, \dots, x_n]$ is equipped with the $\grlexns$ order and the coefficients are bounded by $2^{poly(n^d)}$ (see \cref{sect:PC_bit}).
\end{proof}

Finally, we obtain a method for checking whether a system $\mathcal{P}$ is $\SOSe$.


\begin{corollary}\label{cor:powers-complete}
    Let $\mathcal{P}$ be an explicitly Archimedean polynomial system and let $\mathcal{G}_{2d}$ be a $2d$-truncated \GB basis of $\I(\Variety{\mathcal{P}})$ according to the $\grlexns$ order. Assume that $\| \mathcal{G}_{2d} \|_{\infty} \leq 2^{poly(n^d)}$. If there exists a multi-index $\alpha$ with $|\alpha| = O(d)$ such that $\mathcal{G}^\alpha \lesssim_{\sos} \mathcal{P}$, then $\mathcal{P}$ is $\SOSe$.
\end{corollary}

\begin{proof}
    By \cref{prop:alpha_powers} we have that $\mathcal{G}_{2d} \lesssim_{\sos} \mathcal{G}_{2d}^{\alpha}$, therefore $\mathcal{G}_{2d} \lesssim_{\sos} \mathcal{G}_{2d}^{\alpha} \lesssim_{\sos} \mathcal{P}$. The result follows by \cref{cor:completeness_inheritance}, since $\mathcal{P}$ is explicitly Archimedean and since $\mathcal{G}_{2d}$ is $\SOSe$ by \cref{th:Grobner_basis_SOS_completeness}.
\end{proof}

Furthermore, the relation $\lesssim_{SoS}$ induces an order structure over the set of explicitly Archimedean polynomials systems with the same variety.

\begin{proposition}
    Consider the set
    \begin{equation*}
        \mathfrak{P}_{S} = \{ \mathcal{P} \, | \, \text{$\mathcal{P}$ is explicitly Archimedean with $\Variety{\mathcal{P}} = S$}\}.
    \end{equation*}
    Then $\lesssim_{SoS}$ is a \emph{preorder} (i.e. a reflexive and transitive relation) of the set $\mathfrak{P}_S$.
\end{proposition}

\subsection[Separation between Nsatz and SoS]{Separation between $\Nsatz$ and $\sos$ criteria}\label{sect:separation_Nsatz_SOS}

We now highlight some differences distinguishing the two notions of completeness. We distinguish between two fundamental cases: non-radical and radical ideals (see \cref{def:radical}).

\textit{Non-radical ideals}. When the ideal generated by the input polynomials is not radical, the $\Nsatz$ criterion is inherently weak and does not apply, as the $d$-completeness property cannot be satisfied. In contrast, below we show that our criterion is more robust and it may be effective even for non-radical ideals. 
In the following, we show a concrete example of an application of the $\sos_{\varepsilon}$ criterion for polynomial systems with non-radical ideals.

\begin{example}\label{ex:separation2}
    Let $\mathcal{P} = \{ x_1^2 + x_2^2 + \ldots + x_n^2 =0\} $. We show first that $\mathcal{P}$ is explicitly Archimedean. It suffices to show that $\mathcal{P}$ has bounded variables by \cref{prop:bounded_variables_archimedeanity}. Indeed, let $i \in [n]$ and consider a variable $x_i$. We have that 
    \begin{align*}
        \frac{1}{4} - x_i &= \left(\frac{1}{2} - x_i\right)^2 + x_2^2 + \dots + x_n^2 - (x_1^2 + x_2^2 + \dots + x_n^2), \\
        \frac{1}{4} + x_i &= \left(\frac{1}{2} + x_i\right)^2 + x_2^2 + \dots + x_n^2 - (x_1^2 + x_2^2 + \dots + x_n^2).
    \end{align*}
    Hence, $\mathcal{P}$ is explicitly Archimedean.
    
     Observe now that $\Variety{\mathcal{P}} = \{(0,0 \ldots, 0)\}$. Therefore, the reduced \GB Basis for $I(\Variety{\mathcal{P}})$ is given by $\mathcal{G} = \{ x_1, x_2, \ldots, x_n \}$. Next we observe that $I(\Variety{\mathcal{P}})$ is not radical. Indeed, that there are no $\Nsatz$ proofs of polynomials $x_i$ (for every $i \in [n]$) from $\mathcal{P}$ since the polynomial $h(x_1^2 + \dots + x_n^2)$ is the zero polynomial or it has degree at least 2, for every $h\in \mathbb{R}[x_1, \dots, x_n]$. Hence the $\Nsatz$ criterion cannot be applied to $\mathcal{P}$. However, it is still possible to find $\sos$ proofs of $``-x_i^2 \geq 0"$ and $``x_i^2 \geq 0"$ from $\mathcal{P}$. Thus, by definition, $\mathcal{G}^2 = \{x_1^2, \dots, x_n^2\} \lesssim_{SoS} \mathcal{P}$. Also, by \cref{prop:alpha_powers}, we have  $\mathcal{G} \lesssim_{SoS} \mathcal{G}^2$ and by \cref{th:Grobner_basis_SOS_completeness} we have that $\mathcal{G}$ is $\SOSe$. Thus by \cref{cor:completeness_inheritance} we have that $\mathcal{P}$ is $\SOSe$.

    Lastly, we note that the moment matrix is
    \begin{align*}
        (M_{\Variety{\mathcal{P}},d})_{ij} = \begin{cases}
            1 & \text{for } (i,j) = (1,1), \\
            0 & \text{otherwise}, \\
        \end{cases}
    \end{align*}
    thus $\Variety{\mathcal{P}}$ is $1$-spectrally rich. Therefore, the $\sos_{\varepsilon}$ criterion of \cref{th:SoS_Criterion} applies, i.e. for every polynomial $r$ with a degree $2d$ $\sos$ proof from $\mathcal{P}$, there exists also a proof of $``r + \varepsilon \geq 0"$ of degree $O(d)$ and coefficients bounded by $2^ {poly(n^d,\lg \frac{1}{\varepsilon})}$ for any additive error $\varepsilon > 0$.
\end{example}

\textit{Radical ideals}. The previous separation example show advantages of the $\sos_{\varepsilon}$ criterion over the $\Nsatz$ criterion. But what happens in the case of radical ideals? 

For problems with a finite domain, the ideal is radical, and it is well known that the $\Nsatz$ proof system is complete for sufficiently large degrees. 
However, in general, a linear lower bound on the degree $O(n)$ is unavoidable, in the sense that there are instances of systems $\mathcal{P}$ and polynomials $r$ such that for proving that $``r=0"$ from $\mathcal{P}$ by the $\Nsatz$ proof system there is a lower bound $\Omega(n)$ on the degree \cite{Buss96}. Thus, if the degree is bounded by a constant $d$, both the $\Nsatz$ and the $\sos$ proof systems are again incomplete, even though we are in a radical setting. We address whether a separation can be established between the \(\sos_{\varepsilon}\) criterion and the \(\Nsatz\) criterion in this context. Specifically, we ask whether there exists a polynomial system \(\mathcal{P}\) and a polynomial \(r\) such that any \(\Nsatz\) proof of \(``r = 0"\) from \(\mathcal{P}\) necessarily has a degree that depends non-constantly on \(n\), while, for every additive error \(\varepsilon > 0\), \(\sos\) proofs of \(``r + \varepsilon \geq 0"\) and \(``-r + \varepsilon \geq 0"\) from \(\mathcal{P}\) can be achieved with degree \(O(d)\) and coefficients bounded by \(2^{poly(n^d)}\). We affirmatively answer this question.

To do so, in \cref{sect:applications} we will examine two natural families of problems, which have played a crucial role in the theory of $\CSP(\Gamma)$. In these cases, the ideal is radical because the variables take values from a finite domain. For these families, we show a strict separation between the two criteria.

\emph{Role of $\varepsilon$ in the criteria.} We highlight that our criterion asks for an {\em approximated} proof of the nonnegativity of the elements in the truncated ideal $I_{2d}(S)$. This seemingly subtle difference has a significant impact on the application of the criterion.
Indeed, \cref{ex:separation2} also shows that even if for some $q\in \I_{2d}(S)$ there is no $\sos$ proof for $``q\geq 0"$, there may be one for $``q+\varepsilon\geq 0"$ satisfying the criterion conditions. This shows that not only replacing the proof system with stronger one ($\Nsatz$ with $\sos$) plays a role, but also extending it to an approximate form. 
We further note that allowing this approximation in the condition has an impact in the result of the criterion. Whereas the $\Nsatz$ criterion shows the existence of an $\sos$ proof of $``r\geq 0"$ with bounded coefficients, the  $\sos_{\varepsilon}$ criterion guarantees the existence of a proof of $``r+\varepsilon\geq 0"$ with bounded coefficients. However, following the discussion in the introduction, this second property is enough to guarantee the polynomial-time computability (up to arbitrary precision) and essentially does not compromise the quality of the computed solutions. 

\subsection{The semialgebraic case}\label{sect:semialgebraic_sos_criterion}

We have seen the $\sos_{\varepsilon}$ criterion in the algebraic case with finite $S$. As noted, this setting is very general and covers a wide range of combinatorial problems. Moreover, all the separations we present in this paper are in this case.
Nonetheless, it is not hard to generalize the $\sos_{\varepsilon}$ criterion to the case where $S$ is \textit{infinite} and there are inequality constraints.
Indeed, Raghavendra and Weitz \cite{raghavendra_weitz2017} originally formulated the $\Nsatz$ criterion in this more general setting. For completeness of exposition, we proceed to formulate the $\sos_{\varepsilon}$ criterion in its full generality. The proof in the semialgebraic case, is similar, mutatis-mutandis, to the proof of the $\sos_{\varepsilon}$ criterion (see \cref{sect:sos_criterion_subsection}).

We begin by defining the $\sos$ proof system in this setting.

\begin{definition}
    Let $\mathcal{P} = \{p_1=0,\ldots,p_m=0\}$ be a set of polynomial equality constraints and $\mathcal{Q} = \{q_1 \geq 0, \ldots, q_\ell \geq 0\}$ be a set of polynomial inequality constraints. Consider a polynomial $r\in \mathbb{R}[x_1, \dots, x_n]$. An $\sos$ proof of $``r\geq 0"$ (over $S$) from $(\mathcal{P}, \mathcal{Q})$ is an identity of the form 
    \begin{align*}
        r = \sum_{i=1}^{t_0} s_i^2 + \sum_{i=1}^\ell \left( \sum_{j=1}^{t_i} \lambda_j^2 \right) q_i + \sum_{i=1}^m h_i p_i,
    \end{align*}
    where $s_i, \lambda_j, h_i\in \mathbb{R}[x_1, \ldots, x_n]$. Moreover, we say that the above $\sos$ proof 
    has \emph{degree} at most $d$ if $\max \{\deg(s_i^2), \deg(\lambda_j^2 q_i), \deg(h_ip_i)\} \leq d$.
\end{definition}

Next, we "adjust" various definitions to the semialgebraic case.

\begin{definition}
    Let $\mathcal{P} = \{p_1=0,\ldots,p_m=0\}$ be a set of polynomial equality constraints and $\mathcal{Q} = \{q_1 \geq 0, \ldots, q_\ell \geq 0\}$ be a set of polynomial inequality constraints. We define as
    \begin{align*}
        S = \{x \in \mathbb{R}^n \ | \ p_1(x) = \ldots = p_m(x) =0, \ q_1(x), \ldots, q_\ell(x) \geq 0\}
    \end{align*}
    as the \emph{feasibility set} (or \emph{zero set}) of $(\mathcal{P}, \mathcal{Q})$. Moreover, we recall that the moment matrix is defined as $M = M_{S,d} = \mathbb{E}_{\alpha \in S}[\mathbf{v}_d(\alpha) \mathbf{v}_d^{\sf{T}}(\alpha)]$, where the expectation is taken over the uniform distribution over $S$.
\end{definition}

Lastly, we introduce a new notion to relates to the set of inequality constraints $\mathcal{Q}$

\begin{definition}
    We say that $S$ is $\mu$-robust for $\mathcal{Q}$ if for all $q \in \mathcal{Q}$ and all $\alpha \in S$, it holds that $q(\alpha) > \mu$. 
\end{definition}

We are ready to state the $\sos_{\varepsilon}$ in the semialgebraic case.

\begin{theorem}[$\sos_{\varepsilon}$ criterion]\label{th:SoS_Criterion_semialgebraic}
    Let $\mathcal{P} = \{p_1=0,\ldots,p_m=0\}$ be a set of polynomial equality constraints and $\mathcal{Q} = \{q_1 \geq 0, \ldots, q_\ell \geq 0\}$ be a set of polynomial inequality constraints, with feasibility set $S$ such that $\| S \| \leq 2^{poly(n^d)}$.
    Assume that
    \begin{itemize}
        \item [1)] $S$ is $\delta$-spectrally rich up to degree $d$, 
        \item [2)] $\mathcal{P}$ is $\sos_\varepsilon$-complete over $S$,
        \item [3)] $S$ is $\mu$-robust for $\mathcal{Q}$.
    \end{itemize}
    Let $r$ be a polynomial. If $``r \geq 0"$ has a degree $2d$ $\sos$ proof from $(\mathcal{P}, \mathcal{Q})$
    then, for every $\varepsilon>0$, there exists an $\sos$ proof of $``r+\varepsilon \geq 0"$ of degree $O(d)$ from $(\mathcal{P}, \mathcal{Q})$
    such that the coefficients of every polynomial appearing in the proof are bounded by $2^{poly(n^d, \lg \frac{1}{\delta}, \lg \frac{1}{\mu} \lg \frac{1}{\varepsilon})}$.
\end{theorem}



\section{\sos\ and \PC\ for polynomials over finite domains}\label{sect:PC_criterion}

This section provides a complete exposition of the main technical results concerning the automatability of degree-$d$ $\sos$ proofs. It presents two primary results: the $\PC$ criterion, a sufficient condition for automatability based on the $\PC$ proof system, and the approximate simulation of $\PC$ by $\sos$ over finite domains. The simulation result is used in the proof of the $\PC$ criterion and may be of independent interest.

We begin by formally defining polynomial systems over finite domains. These are systems of polynomials where variables are restricted to take values over finite sets. Following this, we present \cref{th:sospc}, which states that $\sos$ approximately simulates $\PC$ in this finite domain setting. This result builds upon on \cref{th:sos_sim_PC_1}; its relation to prior work and its role in proving the main criterion are discussed. Finally, the section concludes with the presentation of the $\PC$ criterion (\cref{th:PC_criterion}). 

\subsection{Finite domains systems}
Consider a system of real polynomial equations

\begin{equation}\label{eq:Psystem}
\F=\{f_1=0,\ldots,f_m=0\}    
\end{equation}
over variables $x_1,\ldots,x_n$.
We consider the general case of polynomial equations over a finite domain $D$ of even size $2k$, with $k\in \N$, namely every variable can take a value from among $2k$ given rational values $\rho_{1},\rho_{2},\ldots,\rho_{2k}$. If the domain has an odd number of (distinct) elements, repeat an element so the resulting domain has an even number of (not distinct) elements.
We define the univariate rational \emph{domain polynomials} $D_{2k}(x_j)$ of degree $2k$ for each variable $x_j$ as follows: 
\begin{align}
D_{2k}(x_j)&=(x_j-\rho_{1})(x_j-\rho_{2})\cdots (x_j-\rho_{2k}) \qquad j\in[n], \textrm{ or equivalently,}    \label{eq:domain_f_roots}\\
D_{2k}(x_j)&=x_j^{2k}+\alpha_{2k-1}x_j^{2k-1}+\dots + \alpha_1 x_j + \alpha_0 \qquad j\in[n], \label{eq:domain_f}   
\end{align}
where the correspondence between $\{\alpha_0,\ldots,\alpha_{2k-1}\}$ and $\{\rho_1,\ldots,\rho_{2k}\}$ is given by the well-known Vieta's formulas.
To enforce finite domain variables, the axioms
\begin{equation}\label{eq:domain}
D_{2k}(x_j)=0 \qquad j\in[n],    
\end{equation}
are included in the proof systems. 
Hence every variable $x_j$ can take $2k$ possible values which are the roots of \cref{eq:domain}. We denote by $D=\{\rho_1, \rho_2, \dots, \rho_{2k}\}$ the set of domain values, i.e. $D_{2k}(v)=0$ if and only if $v\in D$, and we set
\begin{equation}\label{eq:D}
    \Do=\{D_{2k}(x_j)=0\mid j\in[n]\}.
\end{equation}
For example, to enforce Boolean variables, the axioms $x^2_j - x_j=0$ are always included in the proof systems and in this case $D=\{0,1\}$. 

In summary, we will consider polynomial systems of equations of the following form.

\begin{definition}
    Let $D$ be a finite domain. A \emph{polynomial system over a finite domain $D$} is defined as a set $\mathcal{P}$ of the form
    \begin{equation}\label{eq:axiomsF}
        \mathcal{P} = \F \cup \Do =\{f_1=0,\ldots,f_m=0\} \cup \{D_{2k}(x_j)=0\mid j\in[n]\}.
    \end{equation}
\end{definition}

Recall that we are considering rational values for $\rho_{1},\rho_{2},\ldots,\rho_{2k}$ in \eqref{eq:domain_f_roots}. It follows that $\alpha_{2k-1}, \ldots,$ $\alpha_0$ are also rational. Let $\beta$ be the the minimum number of bits needed to encode each of the numbers $\alpha_{2k-1}, \ldots,$ $\alpha_0,$  and the values $\rho_1, \rho_2, \dots, \rho_{2k}$ when the rational coefficients are represented with their reduced fractions written in binary.
The \emph{polynomial domain size} is the bit length of the binary encoding of $D_{2k}(x)$, namely $O(k\beta)$. 


Now, we recall the following result claiming that every globally nonnegative univariate polynomial has an $\sos$ decomposition with well structured coefficients.

\begin{lemma}\label{lemma-rational-sq}\cite[Section 4, Thm 23]{magron-schwei}
    Let $p\in \mathbb{R}[x]$ be a univariate polynomial of degree $2d$ with rational coefficients, such that each of them can be encoded with $\tau$ bits. Assume that $p(x)\geq 0$ for all $x\in \mathbb{R}$. Then, we have 
    $$p=\sum_{i=1}^{2d+3}a_iq_i^2,$$
    for some nonnegative rational constants $a_i$ and some polynomials $q_i$ of degree $d$ with rational coefficients. All coefficients in this representation can be encoded with ${O}(d^3 +d^2\tau)$ bits.
\end{lemma}
The following result demonstrates that, for any polynomial system \(\mathcal{P}\) over a finite domain \(D\), polynomial-size \(\sos\) proofs can be constructed to establish that the variables are bounded.
\begin{lemma}\label{lemma-new-t}
    There exists a positive rational number $t>\max_{i\in [2k]}\{2,|\rho_i|\}$ that can be encoded with ${O}(k\beta)$ bits such that there exist $\sos$ proofs of degree $2k$
    \begin{align}
        t-x = \sum_{i=1}^{2k+3}a_iq_i^2 - D_{2k}(x), \\
        t+x = \sum_{i=1}^{2k+3}\tilde{a_i}\tilde{q}_i^2 - D_{2k}(x),
    \end{align}
 for some rational constants $a_i$ and some polynomials $q_i$ with rational coefficients. All coefficients in this representations can be encoded with $poly(k,\beta)$ bits.
\end{lemma}
\begin{proof}
    We show only the existence of the first $\sos$ proof, as the second follows with a similar argument. We consider the polynomial $p(x):= D_{2k}(x)-x$. We will find a lower bound for $\min_{x\in \mathbb{R}} p(x)$. Let $\rho := \max_{i\in [2k]}|\rho_i|$ and let $a := \max\{2, \rho\}$. Consider the function $f(x) = x^{2k} - x$. Observe that $f$ is monotonically increasing for every $x \geq 1$ and that $f(a) = a^{2k} - a \geq \rho$, thus $f(x) - \rho = x^{2k} - (\rho + x) \geq 0$ for all $x \geq a$. Moreover, $D_{2k}(\rho + x) \geq x^{2k}$ for all $x \geq a$. Hence, $D_{2k}(\rho + x)-(\rho+x) \geq x^{2k} - (\rho+x) \geq 0$ for all $x\geq a$. On the other hand, since $D_{2k}$ has even degree and positive leading coefficient, it follows immediately that for every $y \geq 0$, we have that $D_{2k}(-\rho-y)-(-\rho-y)\geq 0$. Therefore, $p(x)$ can just take negative values in the interval $[-\rho,\rho+a] \subseteq [-2a, 2a]$. Also, for $x\in [-2a,2a]$ we have that $|x-\rho_i|<3a$ for all $i\in[2k]$. Hence, we have $|D_{2k}(x)|\leq (3a)^{2k}$, and thus for $x\in [-2a,2a]$ we have that
    $$p(x)= D_{2k}(x)- x \geq -((3a)^{2k}+2a)=:-t.$$
    Observe that $t$ can be encoded with $\mathcal{O}(k\beta)$ bits.

    Now, the polynomial $D_{2k}(x)-x+t$ is globally nonnegative. Also, it has rational coefficients that can be encoded with $O(k\beta)$ bits. Then, by Lemma \ref{lemma-rational-sq}, we obtain that 
    $$D_{2k}(x)-x+t = \sum_{i=1}^{2k+3}a_iq_i^2 $$
    for some constants $a_i$ and polynomials $q_i$ that can be encoded in $O(k^3 + k^3\beta)$ bits.
\end{proof}

The results of \cref{lemma-new-t} and \cref{prop:bounded_variables_archimedeanity} allow us to conclude that any polynomial system over a finite domain is explicitly Archimedean.
\begin{proposition}\label{prop:finite-archi}
    Let $\mathcal{P}$ be a polynomial system over a finite domain $D$. Then, $\mathcal{P}$ is explicitly Archimedean.
\end{proposition} 


\subsection{Approximate simulation of \PC\ by \sos}\label{sect:proof_sos_sim_PC}

Berkholz~\cite{berkholz18} related different approaches for proving the unsatisfiability of a system of real polynomial equations. Over Boolean variables, he showed that SoS simulates $\PC$ refutation: any \(\PC\) refutation of degree \(d\) can be converted into an \(\sos\) refutation of degree \(2d\), with only a polynomial increase in size.
In the non-Boolean setting, there are systems of equations that are easier to refute for $\PC$ than for SoS \cite{GrigorievV01}. Grigoriev and Vorobjov \cite{GrigorievV01} show that the simulation of $\PC$ by SoS does not hold in the non-Boolean case, namely when the Boolean axioms $x_j^2-x_j=0$ are omitted. For example, the so-called telescopic system of equations, $\{ yx_1=1, x_1^2=x_2, \ldots, x_{n-1}^2=x_n, x_n=0\}$, has a $\PC$ refutation of degree $n$, but it requires exponential refutation degree in SoS \cite{GrigorievV01}.
However, it is not known if SoS can simulate $\PC$ in the (non-Boolean) general domain setting, namely when variables can take values from a general finite set of values. 
In this section, we address the existing knowledge gap by extending Berkholz's result to general domains. This complements the results in~\cite{GrigorievV01,berkholz18}.
%
%
Recall that we are considering a polynomial systems of the form 
\begin{equation}
    \F=\{f_1=0,\ldots,f_m=0\} \cup \Do=\{D_{2k}(x_j)=0\mid j\in[n]\}
\end{equation}
over variables $x_1,\ldots,x_n$. We prove the following \cref{th:sos_sim_PC_1}, which is a generalization of (\cite[Lemma 1]{berkholz18}). The overall structure of the proof partially follows from the one in \cite{berkholz18} but with some significant differences that will be emphasized below in the proof (see Case~4).

%

%

\begin{lemma}\label{th:sos_sim_PC_1}
    Let $(r_1, r_2, \dots r_L)$ be a \PC\ derivation from $\mathcal{F}\cup \mathcal{D}$ of degree $d$ and size $S$. Then, for every $H\leq L$ there exists an $\sos$ proof of $-(r_H)^2$ of degree $2(d+k-1)$ of the following form
    \begin{equation}\label{eq:sosstep}
         \sum_{i=1}^m (-a_i f_i) f_i + \sum_{j=1}^{n} q_j D_{2k}(x_j) +\sum_{j=1}^{m_1} c_{j}(p_{j})^2=-(r_H)^2,
    \end{equation}
    such that:
    \begin{itemize}
        \item $m_1$ is constant bounded by $O(k, H)$,
        \item $a_i, c_j$ are nonnegative constants, and $q_j$ (for $j\in [n])$, $p_j$ (for $j\in [m_1]$) are polynomials,
        
        \item all coefficients in the proof can be encoded with $poly(k, \beta, S)$ bits.
    \end{itemize}
\end{lemma}

\begin{proof} 
The proof is by induction on $H$. We make a case analysis on the 4 derivation rules \eqref{eq:PCrules} which, respectively, correspond to the following 4 cases. Assume the claim holds for all $H<L$ and we prove that an $\sos$ proof of the form \cref{eq:sosstep} exists for $r_L$. Recall that $d$ is the degree of the $\PC$ proof.

\textbf{Case 1:} If $r_L$ is an axiom from $\F$ (see \eqref{eq:axiomsF}), namely $r_L=f_i$ for some $i\in [m]$.
Then a $\sos$ proof of $-(r_L)^2$ is obtained by setting $a_i=1$, and the other coefficients and polynomials equal to zero. This $\sos$ proof requires degree at most $2d$.

\textbf{Case 2:} If $r_L=D_{2k}(x_j)$ for some $j\in [n]$. Then a $\sos$ proof of $-(r_L)^2$ is obtained by setting $q_L=-D_{2k}(x_j)$, and all other coefficients  and polynomials to zero. This $\sos$ proof requires degree at most $2d$ (note that the $\PC$ proof degree of this case is $2k\leq d$).

\textbf{Case 3:} If $r_L= a\cdot r_{H_1} + b\cdot r_{H_2}$ for some $H_1,H_2<L$, where $r_{H_1}, r_{H_2}$ are previously derived polynomials and $a,b\in \mathbb{R}$. This case and the corresponding analysis are the same as in \cite{berkholz18}. We report the proof for the sake of completeness. By assumption, we have an $\sos$ proof of ``$-r_{H_1}^2\geq 0"$ and of ``$-r_{H_2}^2\geq 0"$ as in (\ref{eq:sosstep}) of degree $2(d+k-1)$ and rational coefficients of the right bit size.  Let $p_L:=ar_{H_1} - br_{H_2}$. Then, using the following identity  
$$-r_L^2 = p_L^2 -2a^2{r_{H_1}}^2 -2b^2{r_{H_2}}^2,$$
we obtain the desired proof of ``$-r_L^2\geq 0$" of degree $2(d+k-1)$ with coefficients of the claimed size.

\textbf{Case 4:} If $r_L=x_j r_H$ for some $H<L$ and $j\in [n]$, where $r_H$ is a previously derived polynomial.
By the induction hypothesis, there is a $\sos$ proof of $-r_H^2$ of degree $\db$ of the form given by \eqref{eq:sosstep}.
The goal is to transform this proof into a proof of $``-(x_j r_H)^2\geq 0"$. Note that multiplying everything by $x_j^2$ does not work since this would increase the degree of the proof to $\db+2$. Instead we want to simulate the multiplication rule in $\PC$ by maintaining the total degree bounded by $\db$. Here is where our approach differs significantly from the one in \cite{berkholz18}, thus improving and extending the result.

Let $t$ be as in Lemma \ref{lemma-new-t}, we observe that it suffices to find an $\sos$ proof of 
\begin{align}\label{aux-1}
``-(x_j-t)^2r_H^2 \geq 0 \ "
\end{align}
as in (\ref{eq:sosstep}) (of degree $2(d+k-1)$ with the claimed bit size). Indeed, by assumption we have a proof of $``-r_H^2\geq 0"$ as in (\ref{eq:sosstep}), and therefore we have a proof of 
\begin{align}\label{aux-2}
``-t^2r_H^2\geq 0 \ "
\end{align}
of degree $2(d+k-1)$ with coefficients of size bit size $poly(k,\beta,S)$. Additionally, by Lemma~\ref{lemma-new-t}, we have an $\sos$ proof of 
\begin{align}
    t - x_j = \sum_{i=1}^{2k+3}a_iq_i^2 - D_{2k}(x_j) 
\end{align}
of degree $2k$ with coefficients of bit size $poly(k,\beta)$. Then, by multiplying this proof by $2t r_H^2$, we obtain an $\sos$ proof 
\begin{align} \label{aux-3} 
   2t^2r_H^2- 2t x_j r_H^2  = \sum_{i=1}^{2k+3}2ta_i(r_Hq_i)^2  - 2tr_H^2D_{2k}(x_j)  
\end{align}
of degree $2(d+k-1)$ with coefficients of bit size $poly(k,\beta, S)$ (recall that $r_H$ has degree at most $d-1$, as $x_jr_H$ was obtained in the $\PC$ derivation of degree $d$). Finally, summing up the proofs in (\ref{aux-1}), (\ref{aux-2}) and (\ref{aux-3}) we obtain an $\sos$ proof of $``-x_j^2r_H^2\geq0"$ as desired.
\\

The rest of the proof is devoted to finding an $\sos$ proof of (\ref{aux-1}). For convenience of notation, let us  use $x$ to denote $x_j$ and $r$ to denote $r_{H}$.    
It follows that our goal is to obtain a $\sos$ proof of $``-((x-t)r)^2\geq 0"$ starting from one as given by \eqref{eq:sosstep} for $``-r^2\geq 0"$. 

Now, consider the following univariate polynomial 
$$p:= 4D_{2k}(x) - (x-t)^2.$$
We think of $D_{2k}(x)$ written as follows:
\begin{align}
D_{2k}(x)&=2(x-t+t-\rho_{1})(x-t+t-\rho_{2})\cdots (x-t+t-\rho_{2k}).     \label{eq:domain_f_rootsM}
%
\end{align}

Recall that, for  our choice of $t$,  we have $t\geq \max\{2,|\rho_i|\}$. Now, we will lower bound the minimum of the polynomial $p$ over $\mathbb{R}$. We prove that for $x\geq t +1$ (i.e., $x-t \geq 1$) we have $p(x) \geq 0$. Notice that if $x-t \geq 1$, then $4D_{2k}(x) \geq 4(x-t)^{2k}$ since each factor in (\ref{eq:domain_f_rootsM}) is positive and at least $(x-t)$. Then, $p(x)\geq 4(x-t)^{2k}-(x-t)^2 = 4(x-t)^2((x-t)^{2k-2}-1)\geq 0$, where the last inequality holds as $x-t \geq 1$. Now, we show that $p(x)\geq 0$ for $x\leq -3t$. Let $x=-3t-a$, for some $a\geq 0$. Then, we have that $x-\rho_i\leq-2t-a$ for $i\in [2k]$, and thus $4D_{2k}(x)\geq 4(2t+a)^{2k} = (4t+2a)^2(2t+a)^{2k-2}\geq (4t+a)^2 = (x-t)^2$. This implies that $p$ can take negative values just in the interval $x\in [-3t, t+1]$. Also, for $x\in [-3t, t+1]$, we have $|x-\rho_i|\leq 4t$, and thus we have
\begin{align} 
\min_{x\in \mathbb{R}} p(x) \geq - (4(4t)^{2k} + 16t^2)=:-C.
\end{align}

It is easy to observe that $C$ is rational and can be encoded with $O(k\beta)$ bits. Then, the univariate polynomial $4D_{2k}(x) - (x-t)^2 + C$ is globally nonnegative. Therefore, by Lemma \ref{lemma-rational-sq}, it can be written as 
\begin{align}\label{eq-p_i}
    4D_{2k}(x) - (x-t)^2 + C = \sum_{i=1}^{2k+3}a_ip_i^2,
\end{align}
where $a_i$ is a constant, $p_i$ (for $i\in [2k+3])$ is a polynomial of degree $k$ and each coefficients in the representation can be encoded in $poly(k,\beta)$. Now, we multiply Equation (\ref{eq-p_i}) by $r^2$ and obtain 

\begin{align}\label{eq-proof}
    - r^2(x-t)^2 = \sum_{i=1}^{2k+3}a_i(rp_i)^2 -  4r^2 D_{2k}(x) -Cr^2.
\end{align}
Since $\deg(r) \leq d-1$, it follows that $-4r^2 D_{2k}(x)$ and all the terms in the sum have degree at most $2(d+k-1)$. On the other hand, by assumption, there is a proof of ``$-r^2\geq 0$" as in (\ref{eq:sosstep}) of degree $2(d+k-1)$ and coefficients of bit size $poly(k,\beta,S)$.  We substitute $-Cr^2$ in (\ref{eq-proof}) with this proof multiplied by $C$ (recall that $C$ is a positive constant that can be encoded with $O(k\beta)$ bits). This yields the desired proof with the claimed bit complexity.
\end{proof}

\begin{remark}[Finite domains of odd size]
    The result of Lemma~\ref{th:sos_sim_PC_1} extends to domains with an odd number of elements. Specifically, if the domain D has size $|D| = |{\rho_1, \rho_2, \ldots, \rho_{2k - 1}}| = 2k - 1$, the same conclusion follows. To show this, we employ the same proof strategy as used in Cases 1, 2, and 3. However, in Case 4, instead of the domain polynomials defined as
    \begin{align*}
        D_{2k-1}(x_i) = (x_i - \rho_1)\cdots(x_i - \rho_{2k-1}),
    \end{align*}
    one can consider a modified set of polynomials, denoted by $\tilde{D}$, where each polynomial is defined as
    \begin{align*}
        \tilde{D}_{2k-1}(x_i) = (x_i - \rho_1)^2(x_i - \rho_2)\cdots(x_i - \rho_{2k-1}).
    \end{align*}
    Here, one root is repeated to ensure an even degree for the polynomials in $\tilde{D}$. With this adjustment, the same arguments apply, yielding an SoS proof from $\mathcal{F} \cup \mathcal{D} \cup \tilde{\mathcal{D}}$, and hence from $\mathcal{F} \cup \mathcal{D}$, of degree $2(d + k - 1)$.
\end{remark}

\begin{remark}
   In the previous result, we highlighted the dependence of the coefficients present in the $\sos$ proof on the parameter $\beta$. Recall that $\beta$ corresponds to the number of bits needed to encode the coefficients of the polynomials $D_{2k}(x)$ and its roots $\rho_1, \dots, \rho_{2k}$. This parameter can be eliminated and implicitly linked to the size of the $\PC$ proof $S$. Specifically, it suffices to assume that the $\PC$ proof begins by deriving all the polynomials $D_{2k}(x_j)$ for $j \in [n]$.
\end{remark}

While \cref{th:sos_sim_PC_1} immediately establishes a simulation of $\PC$ by $\sos$ as \emph{refutation} systems, it remains unclear whether $\sos$ can also simulate $\PC$ as a \emph{derivation} system. Specifically, the existence of an $\sos$ proof of $``-r^2 \geq 0"$ does not immediately guarantee the existence of an $\sos$ proof of $``\pm r \geq 0"$. Further, it is not hard to find polynomial systems for which the latter does not hold. Consider the simple polynomial system $\{ x^2\}$ from which it is trivial to derive $``-x^2 \geq 0"$. However, there are no $\sos$ proofs of $``\pm x \geq 0"$ from such premise.

Interestingly, by the use of $\sos$ approximability techniques developed in \cref{sect:SOS_completeness}, we are able to work around and resolve this issue. Provided that, given a statement derived by $\PC$ $``r = 0"$, an (arbitrarily) small approximation $\varepsilon$ of the statement is allowed, the simulation holds. That is, there exist $\sos$ proofs of $`` r + \varepsilon \geq 0"$ and of $``- r + \varepsilon \geq 0"$.

\begin{theorem}\label{th:sospc}
    $\sos$ \emph{approximates} $\PC$ with degree linear in the domain size $k$ over general finite domains. That is, if there exists a \(\PC\) derivation of \( ``r=0" \) with degree \( d \) and size \( S \), then for every \( \varepsilon > 0 \), we have \(\sos\) proofs of \( ``r+\varepsilon \geq 0" \) and \( ``-r+\varepsilon \geq 0" \) with degree \( O(d+k) \) and coefficients bounded by \( 2^{\text{poly}(k, S, \lg \frac{1}{\varepsilon})} \).
\end{theorem}

\begin{proof}
    Assume there is a \PC\ derivation of $``r=0"$ from $\mathcal{P}$. Then the polynomial systems $\mathcal{P}$, $\mathcal{P}\cup\{r\}$ and $\mathcal{P}\cup\{r^2\}$ have the same zero set. Moreover, by \cref{th:sos_sim_PC_1}, it follows that $\mathcal{P}\cup \{r^2\} \lesssim_{\sos} \mathcal{P}$. Further, by \cref{prop:alpha_powers}, it follows that $\mathcal{P}\cup \{r\} \lesssim_{\sos}\mathcal{P}\cup \{r^2\}$. Thus, by the transitivity of \cref{th:transitivity_of_approximations}, we have that $\mathcal{P}\cup \{r\} \lesssim_{\sos} \mathcal{P}$, i.e. there exist $\sos$ proofs of $``r+\varepsilon\geq 0"$ and $``-r+\varepsilon\geq 0"$ from $\mathcal{P}$ of degree $O(2(d+k-1))$ and coefficients bounded by $2^{poly(k, S, \lg \frac{1}{\varepsilon})}$.
\end{proof}

\subsection[PC criterion]{$\PC$ criterion}\label{sect:pc_crit}

In this section we show that, within the context of finite domains, \cref{th:sos_sim_PC_1} can be combined with the $\sos_{\varepsilon}$ criterion to formulate a new criterion, called \emph{\PC\ criterion}, based on the \PC\ proof system. 
While, in general, \PC\ is weaker than \sos\ as a proof system, it naturally connects to the theory of \GB basis, in particular to Buchberger's algorithm for their computation (see \cite{BuchbergerThesis}).
As we will see in \cref{sect:applications}, this connection enables the application of the \PC\ criterion to certain families of problems arising from $\CSP$s, for which the \Nsatz\ criterion is not satisfied.

\begin{theorem}[PC criterion]\label{th:PC_criterion}
    Let $\mathcal{P} = \{p_1 = 0, \dots, p_m = 0\}$ polynomial system over a finite domain $D$ of $2k$ rational values, let $S=\Variety{\mathcal{P}}$ be its variety and let $r\in\mathbb{R}[x_1, \dots, x_n]$ be a polynomial nonnegative over $S$. Assume there exists an $\sos$ proof of $``r\geq 0"$ from $\mathcal{P}$ of degree $2d$
  \begin{equation*}
        r = \sum_{i=1}^{t_0} \sigma_i^2 + \sum_{i=1}^m h_i p_i.
    \end{equation*}
   Let $\mathcal{G}_{2d}$ be a $2d$-truncated \GB basis of $I(S)$ according to the $\grlexns$ order such that $\|\mathcal{G}_{2d}\|_{\infty}\leq 2^{poly(n^d)}$. Assume that, for every $g\in \mathcal{G}_{2d}$, there exist a $\PC$ derivation of $g$ from $\mathcal{P}$ of size $poly(n^d)$ and degree $O(d)$. Then, for every $\varepsilon>0$, the polynomial $``r+\varepsilon \geq 0"$ has a degree-$O(d)$ $\sos$ proof
    \begin{equation*}
        r +\varepsilon = \sum_{i=1}^{t} \tilde{\sigma}_i^2 + \sum_{i=1}^m \tilde{h}_i p_i,
    \end{equation*}
    where the coefficients of every polynomial appearing in the proof are bounded by $2^{poly(n^d, \lg \frac{1}{\varepsilon})}$.   
\end{theorem}
\begin{proof}
    We divide the proof into two cases depending whether $S$ is empty or not: 
    \begin{enumerate}
        \item If $S = \emptyset$, then $\mathcal{G}_{2d} = \{ 1 \}$, which corresponds to the case of refutations. By assumption, there exist a $\PC$ derivation of $``1=0"$ from $\mathcal{P}$ of size $poly(n^d)$ and degree $O(d)$. Then, by \cref{th:sos_sim_PC_1}, we have a proof of $``-1\geq 0"$ from $\mathcal{P}$ of degree $O(d)$, and coefficients bounded by $2^{poly(n^d)}$, as desired.
        \item If $S\neq \emptyset$, then we apply \cref{th:SoS_Criterion} ($\sos_\varepsilon$ criterion). First, by Corollary \ref{cor:rich-finite}, $S$ is $\delta$-spectrally rich up to degree $d$ for some $\delta>2^{-poly(n^d)}$. It remains to prove that $\mathcal{P}$ is $\sos_{\varepsilon}$-complete over $S$. By assumption, there are $\PC$ derivations for all elements in $\mathcal G_{2d}$ of size $poly(n^d)$. Then, by \cref{th:sos_sim_PC_1}, for all $g\in \mathcal{G}_{2d}$, we have an $\sos$ proof of $``-g^2\geq 0"$ from $\mathcal{P}$ of degree $O(d)$, and coefficients bounded by $2^{poly(n^d)}$. Clearly, we also have a proof of $``g^2\geq 0"$ of degree $O(d)$ and coefficients bounded by $2^{poly(n^d)}$. Then, we have that $\mathcal{G}^{(2, \dots, 2)}\lesssim_{\sos} \mathcal{P}$. Moreover, by Proposition \ref{prop:finite-archi}, $\mathcal{P}$ is explicitly Archimedean. Then, by Corollary \ref{cor:powers-complete}, we obtain that $\mathcal{P}$ is $SOS_{\varepsilon}$-complete, completing the proof.
        \end{enumerate}
\end{proof}

\section{Strong Separation for certain Constraint Satisfaction Problems}\label{sect:applications}
 In what follows, we establish \cref{th:semilattice+dualdiscr} by demonstrating and utilizing the ability of \(\sos\) to approximate a dynamic proof system, such as \PC\ (see \cref{th:sospc}).
 In light of \cref{th:PC_criterion}, it is sufficient to show that $\PC$ can solve in polynomial time $\IMP_d(\Gamma)$ when $\Gamma$ is a finite constraint language closed under a semilattice polymorphism (see \cref{th:semilattice}), and in the case it is closed under a dual-discriminator polymorphism (see \cref{th:dual_discriminator}). 
 The degree lower bound for \Nsatz\ given in \cite{BussP98}, along with the results of this section, and \cref{th:PC_criterion} gives the claimed separation among the $\sos_{\varepsilon}$ and $\Nsatz$ criteria. 

The structure of the following sections is outlined as follows. The literature review, along with essential background and notation, is presented in \cref{sect:csp_literature} and \cref{sect:CSP_IMP_background}, respectively. The proofs of \cref{th:semilattice} and \cref{th:dual_discriminator} are provided in \cref{sect:semilattice-proof} and \cref{sect:dual-proof}, respectively.
\subsection{Related results} \label{sect:csp_literature}
In \cite{Mastrolilli21TALG, BharathiM21}, Mastrolilli and Bharathi initiated a systematic study of the IMP$_d$ tractability for combinatorial ideals arising from Constraint Satisfaction Problems $\CSP(\Gamma)$ in which the type of constraints is restricted to relations from a set $\Gamma$ over the Boolean domain.
Note that $\CSP(\Gamma)$ is just the special case of not-$\IMP_0(\Gamma)$ with $r=1$.
The main results of \cite{Mastrolilli21TALG, BharathiM21} identified the borderline of tractability of $\IMP_d(\Gamma)$ for languages $\Gamma$ over the Boolean domain. By using \GB bases techniques, they expanded Schaefer's dichotomy theorem \cite{Schaefer78} which classifies all CSPs of the form $\CSP(\Gamma)$ over the Boolean domain to be either in P or NP-complete. Recently, Bulatov and Rafiey \cite{BulatovRSTOC22, BulatovRSTACS22} continued this line of research by extending \cite{Mastrolilli21TALG, BharathiM21} beyond Boolean domains in several ways. 

With the aim of expanding the class of $\IMP_d(\Gamma)$s tractable by $\PC$, we observe that some of the algorithms that are considered in \cite{BulatovRSTOC22, BulatovRSTACS22, Mastrolilli21TALG, BharathiM21} for solving the $\IMP_d(\Gamma)$ are known to not being simulable by $\PC$ and by $\sos$.  For example, when $\Gamma$ is closed under the minority polymorphism, in \cite{BharathiM21} it is shown that the membership proof for $\IMP_d(\Gamma)$ can be computed in $n^{O(d)}$ time for any $d\in \mathbb{N}$.
Note that \textsc{3Lin(2)} is a special case of this class of problems. However,  linear (thereby, sharp) lower bounds on degrees for $\sos$ refutations are known \cite{GRIGORIEV2001613} for \textsc{3Lin(2)}.
It follows that bounded degree $\sos$ and $\PC$ over the reals cannot simulate the algorithm in \cite{BharathiM21}. 
The approach in \cite{BharathiM21} has been generalized by \cite{BulatovRSTOC22} by showing that constructing a $d$-truncated \GB Basis for an ideal $\I$ is reducible to solving $\chi\IMP_d$ for the ideal $\I$ (see \cite{BulatovRSTOC22} for details). With this reduction at hand, they designed a general algorithmic approach,
inspired by the famous FGLM algorithm \cite{FAUGERE1993329} and the conversion algorithm in \cite{BharathiM21}, to construct $d$-truncated \GB Basis for many combinatorial ideals, in particular, combinatorial ideals arising from languages invariant under a semilattice, or the dual-discriminator, or languages expressible as linear equations over $GF(p)$. In light of the impossibility result for the particular case of \textsc{3Lin(2)} discussed earlier, the general approach presented by \cite{BulatovRSTOC22}, which also works for \textsc{3Lin(2)}, cannot in general be simulated by $\PC$.

In \cref{sect:applications}, we complement the aforementioned impossibility result with some positive results. 
More precisely, we show that $\PC$ is powerful enough to solve $\IMP_d(\Gamma)$ when $\Gamma$ is closed under a semilattice polymorphism or the dual discriminator. 
As a result of the aforementioned considerations, our approach differs fundamentally from the general methodology employed in \cite{BulatovRSTOC22} (see also the discussion in \cref{rm:semilattice}).
%

Furthermore, strategies in \cite{BharathiM21,BulatovRSTOC22,BulatovARXIV21, Mastrolilli21TALG} to address the problem of \sos\ bit complexity involve replacing the original input polynomial constraints \(\mathcal{P}\) (see \cref{def:SOS_proof}) with a new set of polynomials $\mathcal{P}^{(d)}$ that satisfies the \(\Nsatz\) criterion, and generally depends on the \sos\ degree \(d\). This set $\mathcal{P}^{(d)}$ is computed externally (by an algorithm specifically designed for this purpose), serving as the input for \(\sos\) in place of $\mathcal{P}$. For example, in the semilattice case, if $\mathcal{P}$ consists of $m$ polynomials, the set $\mathcal{P}^{(d)}$, used in \cite{BulatovRSTOC22,Mastrolilli21TALG}, is generated by a specific algorithm and has a size of \( m^{O(d)} \); that is, $\mathcal{P}^{(d)}$ depends on \( d \) and grows exponentially with the \(\sos\) degree \( d \). This preprocessing step ensures that \(\sos\) retains ``low`' bit complexity, but only if $\mathcal{P}$ is substituted with $\mathcal{P}^{(d)}$. Essentially, the approach utilized in \cite{BharathiM21,BulatovRSTOC22,BulatovARXIV21, Mastrolilli21TALG} is to apply the \(\Nsatz\) criterion without enhancing or extending it, with the goal of replacing the initial input polynomial system with a new one that is computed externally and satisfies the \Nsatz\ criterion.
Our results demonstrate that all preprocessing steps employed in \cite{BharathiM21,BulatovRSTOC22,Mastrolilli21TALG} are unnecessary, as \(\sos\) achieves low bit complexity for any fixed \( d \) when \(\mathcal{P}\) is provided directly as input.

\subsection[Background and notation for CSP's]{Background and notation for \(\CSP(\Gamma)\)}\label{sect:CSP_IMP_background}
In this section we give the basic definitions and results that we will need later.
We refer to~\cite{barto_et_al:DFU:2017:6959,2017dfu7,Chen09, Mastrolilli21TALG, BulatovRSTOC22} for more details.

Let $D$ denote a finite set called the  \textbf{\emph{domain}}.
By a $k$-ary \textbf{\emph{relation}} $R$ on a domain $D$ we mean a subset of the $k$-th cartesian power $D^k$; $k$ is said to be the \textbf{\emph{arity}} of the relation. We often use relations and (affine) varieties interchangeably since both are subsets of $D^k$ (we will not refer to varieties from universal algebra in this paper). A \textbf{\emph{constraint language}} $\Gamma$ over $D$ is a set of relations over $D$. A constraint language is \textbf{\emph{finite}} if it contains finitely many relations, and is \textbf{\emph{Boolean}} if it is over the two-element domain $\{0,1\}$. 

A \emph{\textbf{constraint}} over a constraint language $\Gamma$ is an expression of the form $R(x_{i_1},\ldots, x_{i_k})$ where $R$ is a relation of arity $k$ contained in $\Gamma$, and $x_{i_1},\ldots, x_{i_k}$ are variables that belong to the variable set $X$. A constraint is satisfied by a mapping $\phi$ defined on the variables if $(\phi(x_{i_1}),\ldots, \phi(x_{i_k}))\in R$.

\begin{definition}\label{def:csp}
 The \emph{(nonuniform) \textsc{Constraint Satisfaction Problem} ($\CSP$)} associated with language $\Gamma$ over $D$ is the problem $\CSP(\Gamma)$ in which: an instance is a triple $\Cc=(X,D,C)$ where $X=\{x_1,\ldots,x_n\}$ is a set of $n$ variables and $C$  is a set of constraints over $\Gamma$ with variables from $X$. The goal is to decide whether or not there exists a solution, i.e. a mapping $\phi: X\rightarrow D$ satisfying all of the constraints. We will use $Sol(\Cc)\subseteq D^n$ to denote the set of solutions of $\Cc$.
\end{definition}
Moreover, we follow the algebraic approach to Schaefer's dichotomy result \cite{Schaefer78} formulated by Jeavons \cite{JEAVONS1998185} where each class of CSPs that are polynomial time solvable is associated with a polymorphism.
Recall that a polymorphism of a constraint language $\Gamma$ over a set $D$ is a multi-ary operation on $D$ that can be viewed as a multidimensional symmetry of relations from $\Gamma$ (see e.g.~\cite{barto_et_al:DFU:2017:6959}).
\begin{definition}\label{def:polymorph}
An operation $f:D^m \rightarrow D$ is a \textbf{polymorphism} of a relation $R\subseteq D^k$ if for any choice of $m$ tuples $(t_{11},\dots,t_{1k}),\dots,(t_{m1},\dots,t_{mk})$ from $R$ (allowing repetitions), it holds that the tuple obtained from these $m$ tuples by applying $f$ coordinate-wise, $(f(t_{11},\dots,t_{m1}),\dots,f(t_{1k},\dots,t_{mk}))$, is in $R$. We also say that $f$ \textbf{preserves} $R$, or that $R$ is \textbf{invariant} or \textbf{closed} with respect to $f$. A polymorphism of a constraint language $\Gamma$ is an operation that is a polymorphism of every $R\in \Gamma$. By $\Pol(\Gamma)$ we denote the set of all polymorphisms of $\Gamma$. 
\end{definition}

\subsubsection[The ideal membership problem of a constraint language IMP(Gamma)]{The ideal membership problem of a constraint language $IMP(\Gamma)$}\label{sec:IMPdef}

The polynomial \textsc{Ideal Membership Problem} (\IMP) is the following computational task.
Let $\Field[x_1, \ldots, x_n]$ be the ring of polynomials over the field $\Field$ and indeterminates $\{x_1,\ldots, x_n\}$ ordered according to the \grlex order (see \cref{sect:prelim}).
Given $f_0,f_1,\ldots,f_r\in \Field[x_1, \ldots, x_n]$ we want to decide if $f_0\in \I= \GIdeal{f_1,\ldots, f_r}$, where $\I$ is the ideal generated by $F=\{f_1,\ldots , f_r\}$.
If the ideal $\I$ corresponds to a $\CSP$ instance we can be specific on its structure. 
Here, we explain how to construct an ideal corresponding to a given CSP($\Gamma$) instance $\Cc$ by following \cite{Mastrolilli21TALG}. Constraints are in essence varieties (see e.g.~\cite{vandongenPhd,JeffersonJGD13}). 
\begin{definition}\label{def:combinatorial_ideal}
For any given $\CSP(\Gamma)$ instance $\Cc=(X,D,C)$, the \textbf{\emph{combinatorial ideal}} 
\begin{align}\label{eq:comb_Id}
    \I_\Cc=\langle f_{R_1}(X_{R_1}),\ldots,f_{R_\ell}(X_{R_\ell}),f_D(x_1),\ldots,f_D(x_n)\rangle
\end{align} 
is defined as the vanishing ideal of the set $Sol(\Cc)$ and it is constructed as follows.
\begin{itemize}
    \item For every $x_i\in X$ the ideal $I_\Cc$ contains a domain polynomial $f_D(x_i)$ whose zeroes are precisely the elements of the domain $D$.
    \item For every constraint $R_j(X_{R_j})\in C$, where $X_{R_j}$ is a tuple of variables from $X$, the ideal $\I_\Cc$ contains a polynomial $f_{R_j}(X_{R_j})$ such that for $X_{R_j}\in D^{|X_{R_j}|}$ it holds $f_{R_j}(X_{R_j})=0$ if and only if $R_j(X_{R_j})$ is true.
\end{itemize}   
\end{definition}
  See \cite{Mastrolilli21TALG} for more details and properties. 
\begin{definition}\label{def:IMP}
 The {\emph{\textsc{Ideal Membership Problem}}} associated with language $\Gamma$ is the problem $\IMP(\Gamma)$ in which
 the input consists of a polynomial $f\in \Field[x_1, \ldots, x_n]$ and a $\CSP(\Gamma)$ instance $\Cc=(X,D,C)$ where $D\subset \Field$. The goal is to decide whether $f$ lies in the combinatorial ideal~$I_\Cc$. We use $\IMP_d(\Gamma)$ to denote $\IMP(\Gamma)$ when the input polynomial $f$ has degree at most $d$.
\end{definition}

Ideal membership testing can be performed by means of \GB bases. Indeed, if we can compute the $d$-truncated \GB basis $\mathcal{G}_d$ of $\I_\Cc$ in $n^{poly(n^d)}$ time, then we can solve $\IMP_d(\Gamma)$ in polynomial time (see \cref{sect:prelim}).

As in the case of the CSP, polymorphisms of $\Gamma$ are what determines the complexity of $\IMP_d(\Gamma)$ (see \cite{Mastrolilli21TALG,BharathiM21, BulatovRSTOC22}). 
 

\subsection{Polynomial Calculus and semilattice polymorphism}\label{sect:PCsemilattice}

We consider the complexity of $\IMP_d(\Gamma)$ for constraint languages $\Gamma$ closed under a \emph{semilattice} operation $\psi$ (either meet or join). 
There are two kinds of semilattice operations (see e.g.~\cite{Davey_Priestley_2002}). A \emph{join}-semilattice, also known as an \emph{upper}-semilattice, refers to a partially ordered set that possesses a \emph{join} (or least upper bound) for every nonempty finite subset. Conversely, a \emph{meet}-semilattice, or \emph{lower}-semilattice, is a partially ordered set characterized by having a \emph{meet} (or greatest lower bound) for any nonempty finite subset. 
Algebraically, semilattices can be defined {as pairs $\mathcal{D} = (D,\phi)$, where $D$ is a domain and $\phi$ is the semilattice operation \textit{join} or \textit{meet}. Note that both operations are associative, commutative and idempotent binary operations.} 

In the following, we show that standard $\PC$ is $d$-complete and efficient for constraint languages that are closed under a semilattice polymorphism. Our result greatly simplifies known approaches \cite{BulatovRSTOC22, Mastrolilli21TALG} and unifies them into one simple
PC-based algorithm. Further details explaining the substantial differences with what is already known are given and discussed in \cref{rm:semilattice}. 
Our main technical result is as follows.
\begin{theorem}\label{th:semilattice}
Let $\Gamma$ be a finite constraint language over a domain $D$. Consider an instance $\Cc$ of $\CSP(\Gamma)$. If $\Gamma$ is closed under a semilattice polymorphism, then $O(d)$-degree \PC\  can compute in $n^{O(d)}$ time the reduced $d$-truncated \GB basis $\mathcal{G}_d$ (in \grlex order) of the combinatorial ideal $I_\Cc$, for any degree $d \in \mathbb{N}$ and where $n$ is the number of variables.
\end{theorem}
\begin{proof}
    See \cref{sect:semilattice-proof}.
\end{proof}
%
%
\cref{th:semilattice} in conjunction with \cref{th:PC_criterion} implies \cref{th:semilattice+dualdiscr} for semilattice structures.
%

\subsection{Polynomial Calculus and dual discriminator polymorphism}\label{sect:PCdual}
We consider the complexity of $\IMP_d(\Gamma)$ for constraint languages where the dual discriminator operation is a polymorphism of $\Gamma$. 
The dual discriminator is a well-known majority operation \cite{Jeavons:1997:CPC,barto_et_al:DFU:2017:6959} and is often used as a starting point in many $\CSP$ related classifications \cite{barto_et_al:DFU:2017:6959}. For a finite domain $D$, a ternary operation $f$ is called a majority operation if  $f(a,a,b)=f(a,b,a)=f(b,a,a)=a$ for all $a,b\in D$. The \emph{dual discriminator} $\nabla$ on a domain $D$, is a majority operation such that $\nabla(a,b,c)=a$ for pairwise distinct $a,b,c\in D$.

In \cite{BulatovRSTOC22} it is shown that $\IMP_d(\Gamma)$ is solvable in polynomial time for any fixed $d$. The work in \cite{BharathiM21} complements the result in \cite{BulatovRSTOC22} by proving that the \textit{full} (as opposed to \textit{truncated}) \GB basis in graded lexicographic order can be computed in polynomial time and with bounded degree, thus proving polynomial time efficiency for solving the general $\IMP(\Gamma)$ (see also \cref{sect:PC_bit}).

In the following, we again show the power of the $\PC$ by demonstrating that the ad hoc algorithm presented in \cite{BharathiM21} is simulable by $\PC$. 
This greatly simplifies previous algorithms \cite{BulatovRSTOC22, BharathiM21, BharathiM22} and provides another family of problems for which the $\sos_{\varepsilon}$ criterion is provably stronger than the $\Nsatz$ criterion.
Our main technical result is as follows.
\begin{theorem}\label{th:dual_discriminator}
    Let $\Gamma$ be a finite constraint language over a domain $D$. Consider an instance $\Cc$ of $\CSP(\Gamma)$. If $\Gamma$ is closed under a dual discriminator polymorphism, then \PC\ can compute in $n^{O(1)}$ time the reduced \GB basis $\mathcal{G}$ (in \grlex order) of the combinatorial ideal $I_\Cc$, where $n$ is the number of variables and $|D|=O(1)$.
\end{theorem}
\begin{proof}
    See \cref{sect:dual-proof}.
\end{proof}
\cref{th:dual_discriminator}, along with \cref{th:PC_criterion}, implies \cref{th:semilattice+dualdiscr} for dual discriminator structures. 

\section{Proof of \cref{th:semilattice} }\label{sect:semilattice-proof}
We consider the complexity of $\IMP_d(\Gamma)$ for constraint languages $\Gamma$ where $\Pol(\Gamma)$ (see \cref{def:polymorph}) includes a \emph{semilattice} operation $\psi$ (either meet or join). 
There are two kinds of semilattice operations (see e.g.~\cite{Davey_Priestley_2002}). A \emph{join}-semilattice, also known as an \emph{upper}-semilattice, refers to a partially ordered set that possesses a \emph{join} (or least upper bound) for every nonempty finite subset. Conversely, a \emph{meet}-semilattice, or \emph{lower}-semilattice, is a partially ordered set characterized by having a \emph{meet} (or greatest lower bound) for any nonempty finite subset. 
Algebraically, semilattices can be defined {as pairs $\mathcal{D} = (D,\phi)$, where $D$ is a domain and $\phi$ is the semilattice operation \textit{join} or \textit{meet}. Note that both operations are associative, commutative and idempotent binary operations.} 
In mathematics, the symbol for the join (meet) operation in a semilattice is often denoted by the symbol $\vee$ ($\wedge$).
Any such operation induces a partial order ($\preceq$) (and its corresponding inverse order) in which the result of the operation for any two elements represents the least upper bound (or greatest lower bound) of those elements in relation to the established partial order.


The input to $\IMP_d(\Gamma)$ consists of any given set of polynomials that defines the combinatorial ideal $\I_\Cc$ (see \cref{def:combinatorial_ideal}) corresponding to a semilattice closed language $\Gamma$: 
\begin{align}\label{eqn:polynomial_constraints_CSPs}
    f_{R_1}(X_{R_1}),\ldots,f_{R_\ell}(X_{R_\ell}),f_D(x_1),\ldots,f_D(x_n).
\end{align} We want to show that $\PC$ is capable of computing the $d$-truncated \GB basis (in \grlex order) in polynomial time for any fixed $d$.

\paragraph{\cref{th:semilattice} proof outline.}\label{sect:semilattice_structure}
%
Schematically, \cref{th:semilattice} is proven by the following arguments:
\begin{enumerate}[(i)]
    \item First we prove \cref{th:semilattice} for the Boolean case, where the domain $D = \{0,1\}$.
    That is, we show that bounded-degree $\PC$ computes the $d$-truncated \GB basis for the Boolean domain. The known \cite{Mastrolilli21TALG} algorithm to efficiently compute the $d$-truncated \GB basis consists of ``guessing'' the truncated \GB basis in polynomial time. Here, the main technical difficulty is that this guessing ``trick'' is not immediately simulable in an efficient way by \PC. We show that the latter is possible. The algorithm in \cite{Mastrolilli21TALG} essentially reduces the \IMP\ for a given polynomial $f$ in the $2d$-truncated \GB basis to the (contrapositive) problem of checking whether ``non-vanishing assignments" of variables for $f$ belong to the variety. In this work, we are able to $\PC$-derive $f$ by polynomially formulating ``non-vanishing assignments" into an infeasible system of Horn-type polynomials. We then combine algebraic and logical reasoning, leading us to efficient \PC\ refutations of the new system by means of simulation of refutation proofs. By accurately using the \PC\ ability for refutation, we can then retrieve a \PC\ derivation of $f$. This technique may be of independent interest.
    \item We reduce the general case (with arbitrary finite domain $D$) to the Boolean case. The reduction is achieved by encoding the domain $D$ using strings over $\{0,1\}$. The encoding is given by a novel bijective map that preserves the semilattice structure. The strength of our bijection is that it ensures a one-to-one correspondence between the solution spaces of the original $\CSP$ over $D$ and the reduced Boolean problem, which allows us to reduce the (search version of) $\IMP_d(\Gamma)$ to the (search version of) $\IMP_{O(d)}(\Gamma^{01})$, where $\Gamma^{01}$ is a Boolean constraint language derived from $\Gamma$. Crucially, the preservation of the semilattice structure ensures that $\IMP_{O(d)}(\Gamma^{01})$ remains solvable in polynomial time by $\PC$.
 
\end{enumerate}

More details on the second point are given below.

\begin{enumerate}
    \item Show that any instance $\Cc = (X,D,C)$ of $\CSP(\Gamma)$ is reducible to an instance $C^{01}$ of $\CSP(\Gamma^{01})$, where $\Gamma^{01}$ is a finite constraint language over $\{0,1\}$ and so that there exists a $\phi \in \Pol(\Gamma^{01})$ that is a semilattice ($\Min$ or $\Max$ polymorphisms) (see \cref{sect:mapping2bool}). The idea is that we can "encode" $\Cc$ in binary and that the encoding function is invertible (see \cref{sect:binencoding}).
    \item Show $\IMP_d(\Gamma)$ is reducible to $\IMP_{\uval d}(\Gamma^{01})$, where $\Gamma^{01}$ is a constraint language over the Boolean domain $\{0,1\}$ and there exists a semilattice $\phi \in \Pol(\Gamma^{01})$ ($\Min$ or $\Max$ polymorphism) (see \cref{sect:2Bool}). In addition, the reduction ensures that the varieties in the two different domains are in one-to-one correspondence.
    \item By \cref{th:2semilattice}, we can solve $\IMP_{\uval d}(\Gamma^{01})$ by bounded-degree $\PC$.
    \item Our reduction guarantees that we can recover a bounded-degree $\PC$ proof in the finite domain from the bounded-degree $\PC$ proof over the Boolean domain (see \cref{sect:mappingBack}). More precisely, we show how $\PC$ proofs of degree $\uval d$ in the Boolean domain translate into the $\PC$ proofs of degree $O(d)$ in the finite domain, thus proving \cref{th:semilattice}.
\end{enumerate}

\begin{remark}\label{rm:semilattice}
We emphasize that a reduction to the Boolean domain case has also been used to prove that the decision version of $\IMP_d(\Gamma)$ is tractable for such constraint languages $\Gamma$ (see \cite[Th. 5.10]{BulatovRSTOC22} for more details). However the mapping used in \cite{BulatovRSTOC22} is by the means of pp-interpretability (see \cite{DonaPapert1964}), and it is not guaranteed that one can recover proofs for the finite domain under this reduction. In particular, the very first obstacle is given by the fact the mapping $\pi$ in the definition of pp-interpretability \cite[Def.~3.12]{BulatovRSTOC22} is not guaranteed to be a bijection. In fact, this difficulty of transforming the Boolean case proof to the finite general domain case led to the development of a specific method \cite[Th.~6.5]{BulatovRSTOC22} for the search version of the problem. Moreover, as previously noted at the beginning of \cref{sect:applications}, it is far from evident that it can be simulated by $\PC$. In the following, we show that the standard bounded-degree $\PC$ approach is sufficient. This simplifies known approaches and unifies them into one simple $\PC$-based approach.
\end{remark}

\subsection[Min/Max polymorphisms]{$\Min$/$\Max$ polymorphisms}



Two important classes of polymorphisms that played a fundamental role in the celebrated dichotomy theorem by Schaefer~\cite{Schaefer78} are the $\Min/\Max$ polymorphisms. In fact, $\Min$ ($\Max$) is a polymorphism of the (dual) problem \textsc{Horn-SAT} \cite{JEAVONS_TRACTABLE_CONSTRAINTS}. A Boolean language $\Gamma$ is invariant under semilattice operations given by (component-wise) the $\Max$ operation (logical \emph{OR}) or the $\Min$ operation (logical \emph{AND}). The semilattice polymorphism is a well-known generalization of $\Min/\Max$ polymorphisms for the general finite domain.
In this section we show that $\PC$ is $d$-complete and efficient for polynomial systems that are closed with respect to the $\Min$ polymorphism (a similar proof holds for $\Max$). 

In the remainder of this section we focus on system of polynomials $C_1,\ldots, C_m$ that are $\Min$ closed, namely each polynomial constraint $C_i$, along with domain polynomials, has solutions that are closed with respect to the $\Min$ polymorphism. These polynomials are defined in \cref{def:combinatorial_ideal}.

%

\subsubsection[Min polymorphism]{$\Min$ polymorphism}

We will make use of the following definition from \cite{Mastrolilli21TALG}.
\begin{definition}\label{def:2terms}
  For a given set $X=\{x_1,\ldots,x_n\}$ of variables and for any set $S\subseteq [n]$ possibly empty, $\alpha\in \{0, \pm 1\}$, let a \textbf{\emph{term}} be defined as \footnote{The empty product has the value 1.}
  \begin{align*}
    \tau^+(S)&\mydef \alpha\prod_{i\in S}x_i, \quad \textsc{*positive term*} \\
    \tau^-(S)&\mydef\alpha\prod_{i\in S}(x_i-1). \quad \textsc{*negative term*}
  \end{align*}
  For $S_1,S_2\subseteq [n]$ and $i\in [n]$, let a \textbf{\emph{2-terms polynomial}} be a polynomial that is the sum of two terms or it is $\pm(x_i^2-x_i)$.
  We say that a set $G$ of polynomials is \textbf{\emph{2-terms structured}} if each polynomial from $G$ is a 2-terms polynomial.

  We further distinguish between the following special 2-terms polynomials:
  \begin{align*}
    \T^+ & \mydef \{\tau^+(S_1)+\tau^+(S_2)\mid S_1,S_2\subseteq[n]\} \cup\{\pm(x_i^2-x_i)\mid i\in [n]\},  \quad \textsc{*positive 2-terms*}\\
    \T^- & \mydef \{\tau^-(S_1)+\tau^-(S_2)\mid S_1,S_2\subseteq[n]\} \cup\{\pm(x_i^2-x_i)\mid i\in [n]\}. \quad \textsc{*negative 2-terms*}
  \end{align*}
  \end{definition}

\begin{remark}
$\Gamma$ is a finite language. It follows that each given polynomial constraint $C_i$ (as defined in \cref{def:combinatorial_ideal}) has a constant number of feasible solutions.
Therefore \PC\ can efficiently derive any polynomial vanishing over the solutions of $C_i$, and in particular it can derive any vanishing $2$-term polynomial.
\end{remark}

\begin{lemma}\label{th:Min}
If $\Min\in\Pol(\Gamma)$ then $\PC$ can compute the truncated reduced \GB basis $\mathcal{G}_d$ in $n^{O(d)}$ time (where $n$ is the number of variables), for any degree $d\in \mathbb{N}$.
\end{lemma}
\begin{proof}
We know from \cite{Mastrolilli21TALG} that the reduced \GB basis $\mathcal{G}$ has a positive 2-term structure (see \cref{def:2terms}). \footnote{Negative 2-term structure if $\Max\in\Pol(\Gamma)$.} Let $P,Q\subseteq [n]$.  Consider the following positive 2-term polynomial~$f$:
\begin{equation*}
f = p +\alpha q, \;
\text{where } p = \prod_{i\in P} x_i,\; q = \prod_{j\in Q} x_j \text{ and } \alpha \in \{0,\pm 1\}.
\end{equation*}
Assume $\alpha = -1$. The other cases are similar. The claim follows by showing that if $f$ belongs to $\mathcal{G}_d$ then there is a degree-$O(d)$ $\PC$ proof of $f=0$ with $n^{O(d)}$ size. 

Let $X_f$ be the set of variables appearing in $f$. Consider a subset $Y\subseteq X_f$ and a mapping $\phi: Y\rightarrow \{0,1\}$. We say that $(Y,\phi)$ is a \emph{non-vanishing partial assignment} of $f$ if there exists no assignment of the variables in $X_f\setminus Y$ that makes $f$ equal to zero while {$\phi(x_i)$ is assigned to $x_i$}, for $i\in Y$; moreover, $(Y,\phi)$ is \emph{minimal} with respect to set inclusion if by removing any variable $x_j$ from $Y$ there is an assignment of the variables in $X_f\setminus (Y\setminus\{x_j\})$ that makes $f$ equal to zero while {$\phi(x_i)$ is assigned to $x_i$} for $i\in Y\setminus\{x_j\}$.

A minimal non-vanishing partial assignment {$\psi$} of $f$ implies that either $p=1$ and $q=0$, or $q=1$ and $p=0$. {In the former case, $\psi:P\cup\{j\}\rightarrow \{0,1\}$ such that $\psi(x_i)=1$ for every $i\in P$, and $\psi(x_j)=0$ for some $j \in Q$ is such an assignment.
In the latter, $\psi:Q\cup\{i\}\rightarrow \{0,1\}$ is of the form $\psi(x_j)=1$ for every $j\in Q$, and $\psi(x_i)=0$ for some $i \in P$ is such an assignment.}

The claim of \cref{th:Min} follows by \cref{th:MiNVA_sos} and \cref{th:f_sos}.
\end{proof}
\begin{lemma}\label{th:MiNVA_sos}
If $f=p-q$ belongs to $\mathcal{G}_d$ then the following polynomials 
%
\begin{equation}\label{MINVA1}
    {p_j = p(x_j-1) \quad \forall j \in Q, \quad q_i = q(x_i-1) \quad \forall i \in P}
\end{equation}
%
belong to the combinatorial ideal $\I_\Cc$ and there is a $O(d)$-PC proof of this fact.
\end{lemma}
\begin{proof}
Note that if the instance $\Cc$ has no solution, i.e. $\Variety{\I_\Cc}=\emptyset$, then the claim is vacuously true. If $P=\emptyset$ or $Q=\emptyset$ then again the claim is vacuously true. So in the following we will assume, w.l.o.g., that $\Variety{\I_\Cc}$, $P$ and $Q$ are not empty sets.

Since $f=p-q$ belongs to $\mathcal{G}_d$, it follows that the polynomials in \eqref{MINVA1} must vanish at every solution from $\Variety{\I_\Cc}$. In fact, if there {was} a solution $s$ from $\Variety{\I_\Cc}$ that would make at least one of the polynomials in \eqref{MINVA1} nonzero, then $f(s)\not=0$ contradicting our hypothesis.

Now consider any minimal non-vanishing partial assignment $\phi_j:P\cup\{j\}\rightarrow \{0,1\}$ such that
\begin{equation}\label{eq:MINVA_SOL}
    \phi_j(x_i)=1 \text{ for } i\in P\text{, and  } {\phi_j(x_j) = 0} \text{ for some }j \in Q. 
\end{equation}
Note that the argument works symmetrically if we exchange $P$ with $Q$. If we set variables {$x_j$} according to {$\phi_j(x_k)$} for every $k \in P \cup \{j\}$, then the set of feasible solutions becomes empty, since every feasible solution makes $f$ vanishing by assumption. It follows that there is no feasible solution that satisfies \eqref{eq:MINVA_SOL}. This new CSP instance, i.e., $\Cc$ augmented with the polynomials $x_i - 1 = 0$ for $i \in P$ and $x_j = 0$ for some $j \in Q$ arising from \cref{eq:MINVA_SOL}, is unsatisfiable.
In particular, $p_j$ together with the assignments in \cref{eq:MINVA_SOL} can be interpreted as an infeasible Horn formula. It is well known that Horn clauses admit an efficient refutation by resolution, and thus a degree-$d$ $\PC$ refutation (see \cref{sect:ref_deg} and \cite{FlemingKothariPitassi19}). Suppose the refutation is given by the sequence of polynomials $(r_1=0,\dots,r_L=0)$ of degree at most $d$ where $r_L=1$. Multiplying each polynomial in this sequence by $p_j$, we can prove that the augmented system has a $\PC$ proof of $p_j$ of degree at most $2d$. Interestingly, the same proof is also valid in the original system: if some $r_k$ is an assignment corresponding to \cref{eq:MINVA_SOL} i.e., $r_k = x_k - \phi_j(x_k)$, then $p_j r_k = (x_k - \phi_j(x_k))p_j$. If $k = i \in P$ then $p_jr_k=(x_i-1)p_j$ which is a multiple of the domain polynomial $(x_i-1)x_i$. If $k = j \in Q$ then $p_jr_k=(x_j)p_j$ which is a multiple of the domain polynomial $x_j(x_j-1)$.
A similar proof can be obtained for $q_i$ and the claim follows.
\end{proof}

\begin{lemma}\label{th:f_sos}
If $f=p-q$ belongs to $\mathcal{G}_d$, then $f=0$ admits a degree-$O(d)$ $\PC$ proof.
\end{lemma}
\begin{proof}
    We first consider the case of $f=p-q$ where $\gcd(p,q)=1$, that is, $p$ and $q$ have no variables in common. Without loss of generality, we have
\begin{equation*}
    p = x_{i_1}x_{i_2}\cdots x_{i_{|P|}} \textrm{ and } q = x_{j_1}x_{j_2}\cdots x_{j_{|Q|}}.
\end{equation*}
Let 
$p_{j_k} = p(x_{j_k}-1) \;\;\forall k\in [|Q|] \textrm{ and } q_{i_k} = q(x_{i_k}-1) \;\;\forall k\in [|P|].$
Then,
\begin{equation*}
    - p_{j_{|Q|}} - p_{j_{|Q|-1}}(x_{j_{|Q|}}) - p_{j_{|Q|-2}}(x_{j_{|Q|-1}})(x_{j_{|Q|}}) - \cdots - p_{j_{1}}(x_{j_2})\cdots(x_{j_{|Q|-1}})(x_{j_{|Q|}}) = p - pq.
\end{equation*}
Similarly, 
\begin{equation*}
    q_{i_{|P|}} + q_{i_{|P|-1}}(x_{i_{|P|}}) + q_{i_{|P|-2}}(x_{i_{|P|-1}})(x_{i_{|P|}}) + \cdots + q_{i_{1}}(x_{i_2})\cdots(x_{i_{|P|-1}})(x_{i_{|P|}}) = -q + pq.
\end{equation*}
Adding the two equations above, for particular polynomials $h_{j_k}$, $h_{i_k}$  we have
\begin{equation*}
    p - q = \sum_{k\in [|Q|]} h_{j_{k}}p_{j_{k}} + \sum_{k \in [|P|]} h_{i_{k}}q_{i_{k}}.
\end{equation*}
Note that, since $\deg(f) \leq d$, then $|P|,|Q| \leq d+1$. Therefore, $p-q$ admits a degree-$O(d)$ \PC\ proof.

On the other hand, suppose $f=m(p-q)$ where $\gcd(p,q)=1$, and $m$ is some multilinear monomial (in the previous case, $m=1$). Continuing the notations for $p$ and $q$, the new minimal non-vanishing partial assignments imply
\begin{equation*}
    p_{j_k} = mp(x_{j_k}-1) \;\;\forall k\in [|Q|] \textrm{ and } q_{i_k} = mq(x_{i_k}-1) \;\;\forall k\in [|P|].
\end{equation*}
Then we have
\begin{equation*}
    m(p - q) = \sum_{k\in [|Q|]} h_{j_{k}}p_{j_{k}} + \sum_{k \in [|P|]} h_{i_{k}}q_{i_{k}},
\end{equation*}
which proves that $m(p-q)$ admits a degree-$O(d)$ \PC\ proof.
\end{proof}
By symmetric arguments we obtain the following.
\begin{lemma}\label{th:Max}
If $\Max\in\Pol(\Gamma)$ then $\PC$ can calculate the reduced truncated basis \GB $\mathcal{G}_d$ in $n^{O(d)}$ time (where $n$ is the number of variables), for any degree $d\in \mathbb{N}$.
\end{lemma}
\begin{corollary}\label{th:2semilattice}
    Let $\Gamma$ be a finite constraint language over $\{0, 1\}$ that is closed under  a 2-element semilattice operation polymorphism. Then $\PC$ can compute the reduced truncated \GB basis $\mathcal{G}_d$ in $n^{O(d)}$ time (where $n$ is the number of variables), for any degree $d\in \mathbb{N}$.
\end{corollary}

\subsection{Generalizing to finite domain semilattice}
In the following, we generalize \cref{th:2semilattice} to constraint languages over finite domains that are closed under a semilattice polymorphism and obtain the proof of \cref{th:semilattice}. We begin by recalling the definition of semilattice operations. We then present the main arguments used to prove \cref{th:semilattice}, followed by their details.

\subsubsection{Binary encoding}\label{sect:binencoding}
{Let} $\mathcal{D} = (D,\psi)$ {be} a semilattice{, with} $\psi$ {being} a semilattice polymorphism (see \cref{def:polymorph}) of the constraint language $\Gamma$.
Then, it is known that $\mathcal{D}$ can be encoded in binary form in such a way that it is a subalgebra of $\mathcal{B}^k$ {for some $k \in \mathbb{N}$}, where $\mathcal{B} = (\{0, 1\}, \phi)$ is a 2-element semilattice \cite{DonaPapert1964}.
In the following we show how to encode the elements of $D$ in binary in a proper form such that (i) the just mentioned property \cite{DonaPapert1964} is satisfied, and (ii) it will allow us to recover proofs {over the finite domain $D$} from the Boolean domain, a property that is not guaranteed by the approach considered in \cite{BulatovRSTOC22}.
The encoding $\mu$ is very ``natural'' and it is described in the proof of the following lemma.
\begin{lemma}\label{th:bin_encoding}
    Let $\mathcal{D} = (D,\psi)$ be a finite semilattice where $\psi$ is a meet-semilattice (join-semilattice) operation. Then there is a mapping $\mu:D\rightarrow \{0,1\}^{\uval}$ such that $\Min$ ($\Max$) is a polymorphism of {the Boolean relation} $D^{01}=\{\mu(d_1),\ldots,\mu(d_{|D|})\}$ and $\mu:D\rightarrow D^{01}$ is bijective.
\end{lemma}
\begin{proof}
    Assume that $\mathcal{D}$ is a meet-semilattice. Note that every join-semilattice is a meet-semilattice in the inverse order and vice versa, so the construction that we describe below can be easily adapted for join-semilattice. Let $D=\{d_1,\ldots,d_{|D|}\}$. Let us start by encoding every element $d_i$ of $D$ by using $|D|$ bits $(b_{i1},\ldots,b_{i{|D|}})$ such that the $j$-th bit $b_{ij}$ is 1 if and only if $d_{j}\preceq d_i$ (recall any semilattice operation, meet or join, induces a partial order $\preceq$), and $0$ otherwise. An easy argument will show that we can remove one of the $|D|$ bits of the proposed encoding and still retain the same properties.  
    
    
    {We call he above binary encoding $\mu$. It maps every element of $D$ to one element from $\{0,1\}^{|D|}$. Let $D^{01}=\{\mu(d_1),\ldots,\mu(d_{|D|})\}$. It is easy to observe that each element in $D^{01}$ is mapped to from at most one element of the domain, namely $\mu:D\rightarrow D^{01}$ is a bijection. Now we observe that the Boolean relation $D^{01} = \{\mu(d_1), \dots, \mu(d_s)\}$ is closed under the (bitwise) \textit{and} (or equivalently, $\Min \in \Pol(D^{01})$). Indeed, let $\phi$ be the binary $\Min$ over the Boolean domain and consider $\mu(d_{i_1}) = (b_{i_1 1}, \dots, b_{i_1 |D|})$ and $\mu(d_{i_2}) = (b_{i_2 1}, \dots, b_{i_2 |D|})$ for some $i_1, i_2 \in [|D|]$. By applying component-wise the map $\phi$ to $\mu(d_{i_1})$ and $\mu(d_{i_2})$ we obtain a string $(c_1, \dots, c_{|D|})$. First note that there is some $\ell \in {|D|}$ such that $c_k = b_{i_1 k} = b_{i_2 k}$ for all $k \leq \ell$ and $c_{\ell} = 1$, while $c_k = 0$ for all $\ell < k \leq |D|$. Thus $(c_1, \dots, c_{|D|}) \in D^{01}$: indeed any string is in $D^{01}$ if it is composed by a contiguous subsequence of $(b_{i_1 1}, \dots, b_{i_1 |D|})$ (or of $(b_{i_2 1}, \dots, b_{i_2 |D|})$) starting from $b_{i_1 1}$ ($b_{i_2 1}$) such that the last element of the subsequence is a 1 followed by only zeros. Moreover, it is easy to see that $(c_1, \dots, c_{|D|}) = \mu(d_{i_1} \wedge d_{i_2})$.
    }

    Finally, recall that a meet-semilattice is a partially ordered set characterized by having a greatest lower bound GLB with respect to the induced partial order $\preceq$ (also simply called \textit{meet}). By the previous construction we see that the column headed by GLB contains only 1's, since every other element is greater than GLB. So a slightly more compact binary encoding is obtained by dropping the bit $b_{iGLB}$ for every $i$ without any loss.
\end{proof}


 
\subsubsection[Reducing CSP(Gamma) over a finite domain to the Boolean domain]{Reducing $\CSP(\Gamma)$ over a finite domain to the Boolean domain}\label{sect:mapping2bool}
By \cref{th:bin_encoding} we obtain the following.%

\begin{lemma}\label{th:reduction_csp}
Let $\Gamma$ be a constraint language over $D$ that is closed with respect to a meet (join) semilattice operation. Let $\Gamma^{01}$ be the constraint language over $\{0,1\}$ that is obtained from $\Gamma$ by replacing the values from $D$ appearing in the relations from $\Gamma$ with their corresponding binary encoding, as given by the mapping $\mu$ in \cref{th:bin_encoding}. Then 
\begin{itemize}
    \item $\Min$ ($\Max$) is a polymorphism of $\Gamma^{01}$.
    \item Any given instance $\Cc=(X,D,C)$ of $\CSP(\Gamma)$ is polynomial time reducible to an instance $\Cc^{01}=(Y,\{0,1\},C^{01})$ of $\CSP(\Gamma^{01})$. Moreover, the solution sets $Sol(\Cc)$ and $Sol(\Cc^{01})$ are in one-to-one correspondence.  
\end{itemize} 
\end{lemma}
\begin{proof}
The claim that $\Min$ ($\Max$) is a polymorphism of $\Gamma^{01}$ is an immediate consequence of \cref{th:bin_encoding}. The claimed polynomial time reduction is obtained by the following construction.
\begin{enumerate} 
    \item For every variable $x_i\in X$ introduce $\uval$ new binary variables $y_{i1},\ldots, y_{i\uval}$. Let $Y_i=y_{i1},\ldots, y_{i\uval}$ and $Y=\{Y_1,\ldots, Y_n\}$.
    \item Replace the values from $D$ appearing in the relations from $\Gamma$ as described in \cref{th:bin_encoding}. This reduces the constraint language $\Gamma$ to its corresponding binary encoded constraint language $\Gamma^{01}$. 
    Indeed, any $k$-ary relation $R$ on a domain $D$ becomes a $\uval k$ relation $R^{01}$ on domain $\{0,1\}$.
    \item Replace every variable $x_i$ with $Y_i$. Then every constraint $R(x_1,\ldots, x_k)$ reduces to $R^{01}(Y_1,\ldots, Y_k)$.
    This maps any given instance $\Cc=(X,D,C)$ of $\CSP(\Gamma)$ to an instance $\Cc^{01}=(Y,\{0,1\},C^{01})$ of $\CSP(\Gamma^{01})$, where $C^{01}$ is essentially the binary encoding representation of $C$. The solution sets $Sol(\Cc)$ and $Sol(\Cc^{01})$ are in one-to-one correspondence by construction. 
\end{enumerate}    
\end{proof}
%
\subsubsection[Reducing IMPd(Gamma) over a finite domain to the Boolean domain]{Reducing $\IMP_{d}(\Gamma)$ over a finite domain to the Boolean domain} \label{sect:2Bool}
%
In the reduction considered in \cref{th:reduction_csp}  every variable $x_i\in X$ is mapped to $\uval$ new binary variables $y_{i1},\ldots, y_{i\uval}$.
In the following we reduce $\IMP_{d}(\Gamma)$ to $\IMP_{\uval d}(\Gamma^{01})$. This reduction, along with its corresponding ``inversion'' (see \cref{sect:mappingBack}), will be used to prove \cref{th:semilattice}, as summarized in \cref{sect:semilattice_structure}.

\paragraph{The interpolating polynomial $\p$ in the Boolean domain.} 
By a straightforward generalization of Lagrange interpolating polynomials (see e.g. \cite{phillips2003interpolation}, \cite{Gasca2000-hx}), given $|D|$ distinct values  $\mu(d_1),\ldots,\mu(d_{|D|})\in \{0,1\}^{\uval}$ (see \cref{th:bin_encoding}) and corresponding values $d_1,\ldots, d_{|D|}$, there exists a polynomial $\p$ of degree at most $|D|-1$ that interpolates the data, i.e. $\p(\mu(d_i))=d_i$ for each $i=1,\ldots,|D|$.
\paragraph{A reduction to the Boolean domain.}
As in the proof of \cref{th:reduction_csp}, we want to map every variable $x_i\in X$ to a tuple of $\uval$ new binary variables $y_{i1},\ldots, y_{i\uval}$. Let $Y_i=y_{i1},\ldots, y_{i\uval}$ and $Y=\{Y_1,\ldots, Y_n\}$. \\
To guarantee that each tuple $Y_i$ assumes only values that correspond to valid encondings of elements in $D$, we consider the following ``low'' degree polynomial:
\begin{align*}
    \T(y_{1},\ldots, y_{\uval})=\prod_{v_1,\ldots,v_\uval\in D^{01}}(1-\prod_{j=1}^{\uval}(1-v_j+y_j))).
\end{align*}
Every $\T(Y_i)$ has degree $|D|\uval$. \\
For any given $\CSP(\Gamma)$-instance $\Cc=(X,D,C)$, the corresponding combinatorial ideal  (see \cref{def:combinatorial_ideal})  
\begin{align}
    \F&=\{f_{R_1}(x_{1_{R_1}}\ldots,x_{k_{R_1}}),\ldots,f_{R_\ell}(x_{1_{R_\ell}}\ldots,x_{k_{R_\ell}}),f_D(x_1),\ldots,f_D(x_n)\} \label{eq:F}\\
    \I_\Cc&=\langle \F\rangle
\end{align}
is mapped to
\begin{align}
 \F^{01}&=\{f_{R_1}(\p(Y_{1_{R_1}}),\ldots,\p(Y_{k_{R_1}})),\ldots,f_{R_\ell}(\p(Y_{1_{R_\ell}})\ldots,\p(Y_{k_{R_\ell}})),\nonumber\\ 
&\qquad \T(Y_1), \ldots,\T(Y_n),y_1^2-y_1,\ldots,y_{n\uval}^2-y_{n\uval}\} \label{eq:F01}\\ 
\I_\Cc^{01}&=\langle \F^{01}\rangle.
\end{align}
Note that 
\begin{itemize}
    \item The polynomial constraints $\{\T(Y_i)=0\mid \forall i\in [n]\}$ (along with $y_i^2-y_i=0$) forces each tuple $Y_i$ to take only the values from $D^{01}$.
    \item The polynomial constraint  $f_{R_j}\left(\p(Y_{1_{R_j}}),\ldots,\p(Y_{k_{R_j}})\right)=0$, for all $j\in [\ell]$, forces the tuples $Y_{i_{R_j}}$ to take only the values whose corresponding value (according to $\mu$, see \cref{th:bin_encoding}) in finite domain satisfy $f_{R_1}(x_{1_{R_1}}\ldots,x_{k_{R_1}})$.
    \item Let $u(x_1, \dots, x_n)$ and consider $u^{01}(\p(Y_1), \dots, \p(Y_n))$. Then
    \begin{align*}
        u(x) = 0 \iff u^{01}(Y) = 0 \iff (\p(Y_1), \dots, \p(Y_n)) \in \Variety{\I_\Cc}.
    \end{align*}
    Thus $\Variety{\I_\Cc}$ and $\Variety{\I_{\Cc^{01}}}$ are in one-to-one correspondence and $u \in \I_\Cc$ if and only if $u^{01} \in \I_{\Cc}^{01}$.
    \item If $\deg(u) = d$, then $\deg(u^{01}) \leq \uval d$.
\end{itemize}
Note that the set of satisfying assignments $Sol(\Cc)$ corresponds to the variety $\Variety{\I_\Cc}$ of $\I_\Cc$, i.e. $Sol(\Cc)=\Variety{\I_\Cc}$. 
\begin{corollary}\label{th:imp_red}
 $\IMP_d(\Gamma)$ is polynomial time reducible to $\IMP_{\uval d}(\Gamma^{01})$, where $\Gamma^{01}$ if a finite constraint language over $\{0, 1\}$ that is closed under  a 2-element semilattice operation polymorphism, and such that $\Variety{\I_\Cc}$ and $\Variety{\I_{\Cc^{01}}}$ are in one-to-one correspondence. 
\end{corollary}

\subsubsection[Mapping the Boolean PC proof back to finite domain]{Mapping the Boolean $\PC$ proof back to finite domain}\label{sect:mappingBack}
Let $|D|=O(1)$. By \cref{th:Min}, we can solve $\IMP_{\uval d}(\Gamma^{01})$ by bounded-degree $\PC$ in $n^{O(d)}$ time. In the following we show the existence of a polynomial time bounded-degree $\PC$ proof for $\IMP_{\uval d}(\Gamma^{01})$ implies a polynomial time bounded-degree $\PC$ proof for $\IMP_{d}(\Gamma)$, and hence the proof of \cref{th:semilattice} follows.

To begin, we introduce the following terminology.
\begin{definition}
    Let $\I\subseteq \Field[x_1, \ldots, x_n]$ be an ideal, and let $f,g\in \Field[x_1, \ldots, x_n]$. We say that $f$ and $g$ are \textbf{\emph{congruent modulo $\I$}}, written $f\cong g \mod{\I}$, if $f-g\in \I$.
\end{definition}

\paragraph{The interpolating polynomial $\q_j$ in finite domain.}
For any given value $d_i\in D$ and $j\in [\uval]$, we need a polynomial function $\q_j(d_i)$ that returns the $j$-th bit of $\mu(d_i)=b_{i1},\ldots,b_{i\uval}$, i.e. $\q_j(d_i)=b_{ij}$.
By Lagrange interpolating polynomials (see e.g. \cite{phillips2003interpolation}, \cite{Gasca2000-hx}), given $|D|$ distinct values  $d_1,\ldots,d_{|D|}\in D$ and corresponding values $b_{1j},\ldots,b_{|D|j}$, there exists a polynomial $\q_j$ of degree $|D|-1$ that interpolates the data, i.e. $\q_j(d_i)=b_{ij}$ for each $j=1,\ldots,\uval$.

\begin{lemma}\label{th:x}
    \begin{enumerate}
        \item $x \cong \p(\q_1(x),\ldots,\q_\uval(x)) \mod{\langle f_D(x)\rangle}$. \footnote{Recall $\langle f_D(x)\rangle$ is the ideal generated by the domain polynomial $f_D(x)$ of $x$.}
        \item $\T(\q_1(x), \q_2(x), \dots, \q_{\uval}(x)) \cong 0 \mod{\langle f_D(x)\rangle}.$
        \item $\q_{j}(x)^2 - \q_{j}(x) \cong 0 \mod{\langle f_D(x)\rangle}$
    \end{enumerate}
\end{lemma}
\begin{proof}
    First note that by definition of $\p$ and $\q_j$, we have that $\p(\q_1(x),\ldots,\q_\uval(x))$ is a univariate polynomial in $x$ that assumes the exact value of $x$ every time $x\in D$. It follows that  $\p(\q_1(x),\ldots,\q_\uval(x))= x + R(x)$, where $R\in \Field[x]$ and $R(x)=0$ for every $x\in D$, i.e $R(x)\in \langle f_D(x)\rangle$. \\
    For point 2. it suffices to observe that $\T(\q_1(x), \q_2(x), \dots, \q_{\uval}(x)) = 0$ for all $x \in D$. For point 3. it suffices to observe that $\q_{j}(x)^2 = \q_{j}(x)$ for all $x \in D$.
\end{proof}

\begin{proof} [Proof of \cref{th:semilattice}]
    Assume that $u(x_1,\ldots,x_n)$ is a polynomial of degree $d=O(1)$ such that  $u(x_1,\ldots,x_n)\in \I_\Cc$. We show that $u(x_1,\ldots,x_n)$ has a $\PC$ proof of degree $O(d)$.
    \begin{enumerate}
        \item Reduce the problem to a Boolean problem by replacing variables $x_i$ with $\p(Y_i)$, where $Y_i = y_{i1}, \dots, y_{i\uval}$ are tuples of $\uval$ binary variables (see \cref{sect:2Bool}). Thus $y_{ij}$ represents the $j$-th bit of the binary representation of $x_i$. Let us use $u^{01}$ to denote the polynomial $u$ after the above variables replacement. By \cref{th:2semilattice} and \cref{th:imp_red}, we know that $u^{01}(Y_1,\ldots,Y_n)$ admits a bounded-degree $\PC$ proofs in the Boolean domain. This means that there exists a {$\PC$ derivation} of $u^{01}$ from $\F^{01}$ (see \cref{eq:F01}), i.e. there is a sequence $(r_1^{01}=0,\ldots,r_L^{01}=0)$ of polynomial equations sequencially derived by using \eqref{eq:PCrules} and starting from from $\F^{01}$, with $u^{01}=r_L^{01}$.
        \item\label{def_corr} Consider any polynomial $r^{01}_\ell\in (r_1^{01},\ldots,r_L^{01})$, and replace each Boolean variable $y_{ij}$ that appears in $r^{01}_\ell$ with $\q_j(x_i)$.  Let $r_j$ denote the polynomial $r_j^{01}$ after the just described replacement.
        We will call \textbf{$r_j$ the corresponding polynomial of $r_j^{01}$ in $D$}.
        Note that $r_j\in \Field[x_1, \ldots, x_n]$ and has degree $O(d)$. Therefore we obtain the following sequence $(r_1=0,\ldots,r_L=0)$ of polynomial equations, that represents the $\PC$ proofs in the Boolean domain, but represented by using finite domain $D$ variables.
        \item Let $\I_D=\langle f_D(x_1),\ldots,f_D(x_n)\rangle$. We observe $r_L \cong u \mod{\I_D}$ by a simple application of \cref{th:x}, i.e. there exists a polynomial $R \in \I_D$ such that $r_L = u + R$. Moreover, it is immediate to see that $\deg(R) = O(d)$. The claim of \cref{th:semilattice} follows by proving the following lemma, which, in words, states the sequence $(r_1, \dots, r_L)$ of corresponding polynomials in $D$ can be derived in bounded-degree $\PC$ and there exists a bounded-degree $\PC$ proof of $u$ from $\mathcal{F}$ that is equivalent to $(r_1, \dots, r_L) \ \mod{\I_D}$. 
    \end{enumerate}

\begin{lemma}\label{th:PC_proof_semilattice}
        In the finite domain $D$ setting, starting from $\F$ (see \cref{eq:F}), there exists a degree-$O(d)$ $\PC$ proof, namely a sequence $\mathcal{S} = (s_1=0,\ldots,s_M=0)$ of polynomial equations sequencially derived by using \eqref{eq:PCrules}, such that for each $r_k\in (r_1,\ldots,r_L)$ there is $s_h\in \mathcal{S}$ such that $r_k \cong s_h \mod{\I_D}$.
    \end{lemma}
    
    \begin{proof}
        By induction on $k=1,\ldots,L$, we prove that statement holds for every $r_k$ with $k\in[L]$. We will construct the sequence $\mathcal{S} = (s_1, \dots, s_M)$ with the desired properties. We start adding $\mathcal{F}$ (see \cref{eq:F}) to $\mathcal{S}$.\\
        \begin{itemize}
            \item Assume that $r^{01}_k\in \F^{01}$ (see \cref{eq:F01}), then by a \cref{th:x} we have that $r_k \cong 0 \mod{I_D}$. Thus we do not to have to add anything to $\mathcal{S}$.
            \item Assume $r_k^{01}$ is derived by the third rule, i.e.
                \begin{equation}
                    \frac{f^{01}=0 \qquad g^{01}=0}{r_k^{01}=a f^{01} + b g^{01} =0}. 
                \end{equation}
                Let $f$ and $g$ be the corresponding polynomials of $f^ {01}$ and $g^{01}$ in $D$. By the induction hypothesis there exist polynomials $s_f, s_g \in \mathcal{S}$ such that $s_f \cong f \mod{I_D}$ and $s_g \cong g \mod{I_D}$. Thus we add $a s_f + b s_g$ to $\mathcal{S}$.
            \item Assume $r_k^{01}$ is derived by the fourth rule, i.e.
                \begin{equation*}
                    \frac{f^{01}=0}{r_k^{01}=y_{ij}f^{01}=0}. 
                \end{equation*}
                Let $f$ be the corresponding polynomial of $f^{01}$. Then we have that 
                \begin{equation*}
                    \frac{f=0}{\q_j(x_i)f=0},
                \end{equation*}
                where this is a derivation that actually requires $O(1)$ derivations for any $|D| = O(1)$.\\
                By the induction hypothesis there exists $s_f \in \mathcal{S}$ such that $s_f \cong f \mod{I_D}$. Thus we add $Q_j(x_i)s_f$ to $\mathcal{S}$ for which we have that $r_k \cong Q_j(x_i)s_f \mod{I_D}$.
            \end{itemize}
    \end{proof}
    
    By \cref{th:PC_proof_semilattice} there exists $\mathcal{S} = (s_1, \dots, s_M)$ such that for each $r_k \in (r_1, \dots, r_L)$ there is some $s_h \in \mathcal{S}$ such that $r_k \cong s_h \mod{I_d}$. Then by simulating $\mathcal{S}$ with $\PC$ from $\mathcal{F}$, we obtain a $\PC$ proof of $u$ from $\mathcal{F}$ with degree bounded by $O(d)$.
    
\end{proof}

\section{Proof of \cref{th:dual_discriminator}}\label{sect:dual-proof}
In this section, we focus on $\IMP_d(\Gamma)$, where $\Gamma$ is a language closed under the \emph{dual discriminator} polymorphism and show the proof of~\cref{th:dual_discriminator}. 
The dual discriminator is a well-known majority operation \cite{Jeavons:1997:CPC,barto_et_al:DFU:2017:6959} and is often used as a starting point in many CSP-related classifications \cite{barto_et_al:DFU:2017:6959}. For a finite domain $D$, a ternary operation $f$ is called a majority operation if  $f(a,a,b)=f(a,b,a)=f(b,a,a)=a$ for all $a,b\in D$. 
\begin{definition}\label{def:dual discriminator}
  The \emph{dual discriminator} on a domain $D$, denoted by $\nabla$, is a majority operation such that $\nabla(a,b,c)=a$ for pairwise distinct $a,b,c\in D$.
\end{definition}

The input for $\IMP_d(\Gamma)$ consists of any given set of polynomials that defines the combinatorial ideal $\I_\Cc$ (see \cref{def:combinatorial_ideal}) corresponding to a dual discriminator closed language: 
\begin{align}\label{eq:dual_generators}
    f_{R_1}(X_{R_1}),\ldots,f_{R_\ell}(X_{R_\ell}),f_D(x_1),\ldots,f_D(x_n).
\end{align}
We want to show that $\PC$ is capable of computing the full \GB basis (in \grlex order) in polynomial time.

We will assume that the solution set is non-empty, as the search version of $\IMP_0(\Gamma)$ can be solved by \PC\ in polynomial time using ``local-consistency" algorithms, valid for the dual-discriminator, as shown in \cite{JeffersonJGD13}.

\paragraph{\cref{th:dual_discriminator} proof structure.}\label{sect:dual_discriminator_structure}
The central idea is to adapt the algorithms presented in \cite{BharathiM21, BharathiM25} and \cite{BulatovRSTOC22} to design a $\PC$ algorithm that runs in polynomial time for any given domain~$D$.
The main arguments are as follows.
\begin{enumerate}[(i)]
    \item For any given instance of $\IMP_d(\Gamma)$, consider the corresponding $\CSP(\Gamma)$ input instance $\Cc=(X,D,C)$ (see \cref{def:IMP}). It is known that any instance $\Cc = (X, D, C)$ of $\CSP(\Gamma)$ can be reduced to an equivalent $\CSP(\Gamma)$ instance with only binary constraints (that is, constraints with at most two variables in their scope). In this transformed instance, the constraints are organized into three categories: \emph{permutation} constraints, \emph{complete} constraints, or \emph{two-fan} constraints. This restructuring aligns with the classification introduced in prior work (see \cite{Cooper1994CharacterisingTC}). We derive binary constraints according to the $\CSP$ classification. \label{i}
    \item Consider the input combinatorial ideal $\I_\Cc$ associated with the given instance of $\IMP_d(\Gamma)$ (see \cref{eq:comb_Id}). The next phase consists in the decomposition of the combinatorial ideal $\I_{\Cc}$ into a collection of simpler ideals. Each of these simpler ideals arises from the structured binary constraints considered above. The advantage here is that these individual ideals have \GB bases that can be derived efficiently through the $\PC$ algorithm.
    \item Finally, we combine the \GB bases corresponding to the simpler ideals into a single \GB basis for the entire combinatorial ideal $\I_{\Cc}$. This approach allows us to efficiently compute a solution to the original problem in polynomial time.
\end{enumerate}

{Schematically, \cref{th:dual_discriminator} is proven by the following arguments:}
\begin{enumerate}
    \item In \cref{sect:intro_dual_discriminator}, by using the known result for $\CSP$s mentioned in \ref{i}, $\IMP_d(\Gamma)$ can be shown to be equivalent to $\IMP_{d}(\Pi)$, where $\Pi$ is a binary constraint language. To achieve this, we start from the given input polynomials \eqref{eq:dual_generators} and derive bivariate polynomials describing the mentioned binary constraints. The derivation can be done in polynomial time and entirely within the framework of $\PC$. 
    \item In \cref{sect:permutation_constraints}, it is shown that a \GB basis for the combinatorial ideal generated by the \textit{permutation} constraints $\I_{perm}$ can be calculated in polynomial time by $\PC$. In particular, we define new constraints $CPC_i$ that arise from ``chaining'' together permutation constraints, with the property that, if $X_i$ and $X_j$ are the variables in $CPC_i$ and $CPC_j$ respectively, then $X_i \cap X_j = \emptyset$. It follows that $\I_{perm} = \sum_i \I_{CPC_i}$, and a set of generators for each $\I_{CPC_i}$ is found.
    \item In \cref{sect:complete_two-fan_constraints}, similar to the permutation constraints case, a set of generators for the combinatorial ideal $\I_{CF}$ generated by the \textit{complete} and \textit{two-fan} constraints is found. In particular, we find a \GB basis for $\I_{CF}$.
    \item In \cref{sect:combin_gen}, we construct generators for simpler ideals by combining those of $\I_{CPC_p}$ and $\I_{CF}$. This combination preserve the structure that $\I_{\Cc} = \sum_{p \in J} CPC_p + \I_{CF}$. We will show that a \GB basis for $\I_{Cc}$ can be computed with bounded degree in polynomial time.
\end{enumerate}

We will show how to convert the ad-hoc algorithm in \cite{BharathiM21, BharathiM25} into a standard $\PC$ algorithm. Moreover, techniques from \cite{BulatovRSTOC22} are implemented for the case of the complete and two-fan constraints. This algorithm can be used to solve $\IMP_d(\Gamma)$ in polynomial time and leads to \cref{th:dual_discriminator}.

We define the sets of constraints
\begin{align*}
    &C_{P} = \{ C_{ij} \in C \, | \, C_{ij} \text{ is a \textit{permutation} constraint} \} \\
    &C_{CF} = \{ C_{ij} \in C \, | \, C_{ij} \text{ is a \textit{complete} or a \textit{two-fan} constraint} \}.
\end{align*}
Note that that $\Cc = C_{P} \cup C_{CF}$. Therefore,
\begin{equation}\label{eqn:dual_discr_ideal_decomposition}
    \I_{\Cc} = \I_{C_P} + \I_{C_{CF}}.
\end{equation}
The idea is to find a set of generators for each addenda in the sum on the RHS and the to combine these polynomials together in a single \GB basis for $\I_{\Cc}$. We recall that for having the identity \cref{eqn:dual_discr_ideal_decomposition}, by radicality, it is sufficient to have that
\begin{equation*}
    \Variety{I_{\Cc}} = \Variety{I_{C_P}} \cap \Variety{I_{C_{CF}}}.
\end{equation*}

\subsection{Binary constraints}\label{sect:intro_dual_discriminator}

Let $\Gamma$ be a language over a finite domain $D$ closed under the dual-discriminator polymorphism, i.e. $\nabla \in Pol(\Gamma)$. Let $\Cc = (X, D, C)$ be an instance of $\CSP(\Gamma)$. In general, if a language $\Gamma$ is closed under a majority polymorphism $\mu$, any instance of $CSP(\Gamma)$ is equivalent to an instance that has only binary constraints. 

Consider a relation $R$ with arity $m$ and let $J \subseteq [m]$ be a set of indices. We denote $X[J] = (x_j)_{j \in J}$ the subset of variables with indices in $J$, and similarly we denote $pr_J(R)$ the projection of $R$ to the components with indices in $J$, i.e. the tuples $(a_j)_{j \in J}$ such that there exists a $n$-tuple $(b_1, \ldots, b_n) \in R$ with $a_j = b_j$ for all $j \in J$.

\begin{proposition}[\cite{Jeavons:1997:CPC}]\label{prop:majority_to_binary_majority}
    Let $R$ be a relation of arity $m$ that is closed under a majority operation, and let $C$ be any constraint $(X,R)$ constraining the variables in $X$ with relation $R$.
    For any problem $\mathcal{P}$ containing the constraint $C$, the problem $\mathcal{P}'$ obtained by replacing $C$ with the set of binary constraints
    \begin{equation*}
        \{((X[i],X[j]), pr_{i,j}(R)) \, | \, 1 \leq i \leq j \leq m \}
    \end{equation*}
    has exactly the same solutions as $\mathcal{P}$.
\end{proposition}

Moreover, if the language $\Gamma$ is closed under the dual-discriminator polymorphism, i.e. $\nabla \in Pol(\Gamma)$, the binary constraints $C_{ij}$ are well-structured into three types.

\begin{proposition}\cite{Cooper1994CharacterisingTC}
    Suppose $\nabla \in Pol(\Gamma)$. Then each constraints $C_{ij} = \langle (x_i,x_j),R_{ij}) \rangle$ is one of the following three types.
    \begin{enumerate}
        \item Permutation constraint: $R_{ij} = \{(a, \pi(a)) \, | \, a \in D_i \}$ for some $D_i \subseteq D$ and some bijection $\pi:D_i \rightarrow D_j$, where $D_j \subseteq D$.
        \item Complete constraint: $R_{ij} = D_i \times D_j$ for some $D_i, D_j \subseteq D$.
        \item Two-fan constraint: $R_{ij} = \{(\{a\} \times D_j) \cup (D_i \times \{ b \})\}$ for some $D_i, D_j \subseteq D$ and $a \in D_i, b \in D_j$.
    \end{enumerate}
\end{proposition}

\begin{remark} \label{rmrk:binary_constraints_bounded_derivations} 
    We emphasize here that given an instance of a CSP $\Cc = (X,D,C)$ whose language $\Gamma$ is dual-discriminator closed, we can derive bivariate polynomials describing the binary constraints $C_{ij}$ in polynomial time entirely within the framework of \PC. 
    
    Indeed, first we note that $\Gamma$ is a \emph{finite} constraint language. Let $arity(R)$ denote the arity of $R$. Then $M := \max_{R \in \Gamma} arity(R) = O(1)$, i.e. the maximum arity of a relation in $\Gamma$ is constant. 
    
    Second, for any constraint we can easily find a set of generators for its combinatorial ideal. Let $T = R^{T}(x_{i_1}, \ldots, x_{i_m})$ be a $m$-ary constraint from $C$. Let $\mathcal{P}^{T}$ be a set of polynomials such that $Sol(T) = \Variety{\mathcal{P}^{T}}$ and such that the domain polynomials are in $\mathcal{P}^{T}$ for each variable appearing in any of the polynomials. Then the combinatorial ideal of $T$ is equal to the ideal generated by $\mathcal{P}^{T}$, i.e. $I_T = \langle \mathcal{P}^{T} \rangle$. 
    
    Third, since $arity(R^{T}) \leq M$ is bounded by a constant, it follows that the reduced \GB basis of $\langle \mathcal{P}^{T} \rangle$ can be calculated in constant time with respect to the number $n$ of variables in $X$. Indeed, finding the reduced \GB basis is in general an EXPSPACE-complete problem (see \cite{MAYR1982305, Mayr1989}) but only with respect to the number of variables in $\mathcal{P}^{T}$, which however is bounded by $M$ and so independent from $n$.

    Lastly, we observe that any set of polynomials describing a binary constraint can be derived from the generators of the initial (non-binary) constraint. Indeed, by \cref{prop:majority_to_binary_majority} we can describe the solution set of $T$ by using binary constraints $T_{ij}$. Let $\mathcal{P}^{T} \subseteq \mathbb{R}[x_1, \dots, x_n]$ be a set of polynomial generators such that $Sol(T) = \Variety{\mathcal{P}^{T}} \subseteq D^m$ and similarly let $\mathcal{P}^{T_{ij}} \subseteq \mathbb{R}[x_i,x_j]$ such that $Sol(T_{ij}) = \Variety{\mathcal{P}^{T_{ij}}} \subseteq D^2$. Note that the polynomial rings over which $\mathcal{P}^{T}$ and $\mathcal{P}^{T_{ij}}$ differ in the variables. Moreover, if $q \in \langle \mathcal{P}^{T_{ij}} \rangle$, then $q = 0$ over $\Variety{\mathcal{P}^{T_{ij}}}$, but if we interpret $q \in \mathbb{R}[x_1, \dots, x_n]$, then also $q = 0$ over $\Variety{\mathcal{P}^{T}}$ and therefore $q \in \langle \mathcal{P}^{T} \rangle$. Thus, $q$ can be derived by $\PC$ with bounded degree and bounded coefficients from $\mathcal{P}^{T}$.

    Note that if there are two constraints $T^1$ and $T^2$ that contain variables $x_i$ and $x_j$ in their scope, then by \cref{prop:majority_to_binary_majority} the initial $\CSP$ can be equivalently formulated with some binary constraint $T_{ij}$. By the same reasoning as in the previous case, a set of generators for $\mathcal{P}^{T_{ij}}$ can be derived by combining together polynomials in $\mathcal{P}^{T_1}_{ij}$ and $\mathcal{P}^{T_2}_{ij}$, the restrictions to variables $x_i$ and $x_j$ of $\mathcal{P}^{T_1}$ and $\mathcal{P}^{T_2}$ respectively.
\end{remark}

In light of the above results and remarks, we can assume without loss of generality that all the constraints in the CSP instance $\Cc$ are binary, and that the generators for the combinatorial ideal $I_{\Cc}$ are sets of polynomials $\mathcal{P}_{ij}$ with the property $Sol(C_{ij}) = \Variety{\mathcal{P}_{ij}}$. Furthermore, the points of any bivariate variety $\Variety{\mathcal{P}_{ij}}$ arise from either a permutation, a complete or a two-fan constraint.

\subsection{Generating sets}

In this section, we walk through the proof sketched at the beginning of \cref{sect:dual_discriminator_structure}. We start by presenting some derivations that can be efficiently made by $\PC$. We will refer these derivations in the subsequent sections. We then proceed to find generating sets for ideals arising from permutation and from complete and two-fan constraints. Finally, we will combine these generators together and find a \GB basis of $\I_{\Cc}$.

\subsubsection{Derivation schemes}

We present and prove five derivation schemes that can be performed by Polynomial Calculus in polynomial time. This subsection serves as a reference for later discussions and can be skipped at a first read.

Throughout this section $x$ will denote a variable and $D_f, D_g, D_h \subseteq D$.

\begin{lemma}[Derivation Scheme 1]\label{th:derivation_scheme_1}
    Let $h = \Pi_{a \in D_h}(x-a)$ and consider polynomials $f = h(x-\alpha)$ and $g = h(x-\beta)$ with $\alpha \neq \beta$. Then $h$ can be $\PC$-derived from $f$ and $g$ in polynomial time.
\end{lemma}

\begin{proof}
    \begin{align*}
        \frac{f=0 \qquad g=0}{f-g = (-\alpha + \beta) h =0} \quad \Rightarrow \quad \frac{(-\alpha + \beta)h = 0}{h = 0}.
    \end{align*}
\end{proof}

\begin{lemma}[Derivation Scheme 2]\label{th:derivation_scheme_2}
    Let $f = \Pi_{a \in D_f} (x-a)$ and $g = \Pi_{b \in D_g}(x-b)$ with $D_f \cap D_g \neq \emptyset$. Then $h = \Pi_{c \in D_f \cap D_g}(x-c)$ can be $\PC$-derived from $f$ and $g$ in polynomial time.
\end{lemma}

\begin{proof}
    Recall the definition of symmetric difference $\Delta$: given two sets $A$ and $B$, then $A \Delta B := (A \cup B) \setminus (A \cap B)$.

    We prove by induction on the cardinality of the symmetric difference $k = |D_f \Delta D_g|$. We can assume without loss of generality that $D_f \nsubseteq D_g$ and $D_g \nsubseteq D_f$ as otherwise it suffices to set $h := f$ and $D_h := D_f$, or $h := g$ and $D_h := D_g$ respectively. This also implies that $k \geq 2$.
    \begin{itemize}
        \item \textit{Base Case (k=2).} Follows from Derivation Scheme 1.
        \item \textit{Induction Step.} Suppose the result holds for $k \in \mathbb{N}$, we will prove it for $k+1$. Since $D_f \nsubseteq D_g$ and $D_g \nsubseteq D_f$, we can pick $\alpha \in D_f \setminus D_g$ and $\beta \in D_g \setminus D_f$. We first derive polynomials
        \begin{align*}
            &\tilde{f} = \Pi_{a \in (D_f \cup D_g) \setminus \{\beta\}} (x-a)\\
            &\tilde{g} = \Pi_{b \in (D_f \cup D_g) \setminus \{ \alpha \} } (x-b).
        \end{align*}
        By Derivation Scheme 1. we can also derive
        \begin{equation*}
            \tilde{h} = \Pi_{c \in D_{\tilde{h}}} (x-c) \quad \text{where} \quad D_{\tilde{h}} = (D_f \cup D_g) \setminus \{\alpha,\beta\}.
        \end{equation*}
    \end{itemize}
    If $D_{\tilde{h}} \subseteq D_f \cap D_g$ then it suffices to set $h := \tilde{h}$ and $D_h := D_{\tilde{h}}$. Otherwise, we have that $D_{\tilde{h}} \supseteq D_f \cap D_g$ since $\alpha,\beta \notin D_f \cap D_g$. Moreover, $|D_{\tilde{h}} \Delta D_f| < k+1$ and thus, by the inductive hypothesis, $h_f = \Pi_{a \in D_f \cap D_{\tilde{h}}}(x-a)$ can be derived. Similarly, $|D_{\tilde{h}} \Delta D_g| < k+1$ and $h_g = \Pi_{a \in D_g \cap D_{\tilde{h}}}(x-a)$ can be derived. Now $|(D_f \cap D_{\tilde{h}}) \cap (D_g \cap D_{\tilde{h}})| < k+1$ and $(D_f \cap D_{\tilde{h}}) \cap (D_g \cap D_{\tilde{h}}) = D_f \cap D_g$. Therefore, again by the inductive hypothesis, from polynomials $h_f$ and $h_g$ we can derive polynomial $h := \Pi_{c \in D_f \cap D_g}(x-c)$.
\end{proof}

\begin{lemma}[Derivation Scheme 3]
    Let $f = \Pi_{a \in D_f} (x-a)$ and $g = \Pi_{b \in D_g}(x-b)(x^2 + \alpha)$ with $D_f, D_g \subseteq D$, $D_f \cap D_g \neq \emptyset$ and $\alpha \in \mathbb{R}$. Then $h = \Pi_{c \in D_f \cap D_g}(x-c)$ can be $\PC$-derived from $f$ and $g$ in polynomial time.
\end{lemma}

\begin{proof}
    Assuming without loss of generality that $D_f \nsubseteq D_g$, let $\beta \in D_f \setminus D_g$. Now derive polynomials
    \begin{align*}
        &h_1 = x \Pi_{a \in D_f \cup D_g} (x-a), \\
        &h_2 = (x^2 + \alpha) \Pi_{b \in (D_f \cup D_g) \setminus \{\beta\}} (x-b).
    \end{align*}
    Thus
    \begin{equation*}
        h_2 - h_1 = \Pi_{b \in (D_f \cup D_g) \setminus \{\beta \}} (x-b)[\beta x + \alpha],
    \end{equation*}
    from which we can derive polynomial
    \begin{equation*}
        \tilde{h} := \left( x + \frac{\alpha}{\beta} \right)\Pi_{b \in (D_f \cup D_g) \setminus \{\beta\}}(x-b).
    \end{equation*}
    Applying Derivation Scheme 2. first with $f$ and $\tilde{h}$ and then with $g$ we obtain the result.
\end{proof}

\begin{corollary}[Derivation Scheme 4]
    Let $f = \Pi_{a \in D_f} (x-a)$ and $g = \Pi_{b \in D_g}(x-b)\Pi_{d \in F_g}(x^2 + \beta_d)$ with $D_f, D_g \subseteq D$ and some set of indices $F_g$, $D_f \cap D_g \neq \emptyset$ and $\beta_d \in \mathbb{R}$. Then $h = \Pi_{c \in D_f \cap D_g}(x-c)$ can be $\PC$-derived from $f$ and $g$ in polynomial time.
\end{corollary}

\begin{lemma}[Derivation Scheme 5]
    Let $\sigma_{ij}:D_i \rightarrow D_j$ be a bijection with $D_i, D_j \subseteq D$. Let $f$ be the Lagrange interpolating polynomial that simulates $\sigma_{ji} = \sigma_{ij}^{-1}$ over $D_j$. Consider polynomials $p_1 := x_i - f(x_j)$, $p_2 := x_i - a$ and $p_3 = \Pi_{b \in D_j}(x_j - b)$ for some $a \in D_i$. Then $h = x_j - \sigma_{ij}(a)$ can be $\PC$-derived from $p_1$ and $p_2$ in polynomial time.
\end{lemma}

\begin{proof}
    Start with the derivation
    \begin{align*}
        \frac{p_2 = 0 \qquad  p_1 = 0}{p_2 - p_1 = f(x_j) - a = 0}
    \end{align*}
    and observe that $f(x_j)$ has degree at most $|D_j| - 1$. By the Fundamental Theorem of Algebra, the polynomial $f(x_j) - a$ has $\deg(f) \leq |D_j| - 1$ roots with multiplicity and $\sigma_{ij}(a)$ is a root since $f(\sigma_{ij}(a)) = \sigma_{ji}(\sigma_{ij}(a)) = \sigma_{ij}^{-1}(\sigma_{ij}(a)) = a$. On the other hand, $\sigma_{ji}$ is a bijection, therefore $f(b) \neq a$ for every $b \in D_j \setminus \{\sigma_{ij}(a)\}$. Thus
    \begin{equation*}
        f(x_j) - a = (x_j - \sigma_{ij}(a))\Pi_{b \in R}(x_j - b)
    \end{equation*}
    for some $R \subseteq \mathbb{C}$ such that $R \cap \left( D_j \setminus \{ \sigma_{ij}(a) \} \right) = \emptyset$. While some roots might be complex, we have that if a root $b$ is complex also its conjugate $\bar{b}$ is a root. Multiplying together $x_j - b$ and $x_j - \bar{b}$ we get the degree 2 polynomial $x_j^2 + b^2$. Therefore we can rewrite
    \begin{equation*}
        f(x_j) - a = (x_j - \sigma_{ij}(a))\Pi_{b \in R\cap \mathbb{R}}(x_j - b) \Pi_{c \in S} (x_j^2 + \beta_c)
    \end{equation*}
    for some $S \subseteq \mathbb{N}$ and $\beta_c \in \mathbb{R}$.

    The results follows by applying Derivation Scheme 2. and Derivation Scheme 4. to $f(x_j) - a$ and the domain polynomial $p_3$.
\end{proof}

\subsubsection{Permutation constraints}\label{sect:permutation_constraints}

Let $C_{ij} = R_{ij}(x_i,x_j) \in C_P$ be a permutation constraint where $R_{ij} = \{(a, \pi_{ij}(a) \, | \, a \in D_i \}$ for some $D_i,D_j \subseteq D$ and a bijection $\pi_{ij}:D_i \rightarrow D_j$. To the constraint $C_{ij}$ it corresponds the set of polynomials $\{ x_j - f(x_i), x_i - g(x_j), \Pi_{a \in D_i} (x_i - a), \Pi_{b \in D_j}(x_j - b)\}$, where $f$ and $g$ are the (Lagrange) polynomials interpolating the points $\{(a,\pi_{ij}(a))\}_{a \in D_i}$ and $\{(\pi_{ij}^{-1}(b),b)\}_{b \in D_j}$ respectively. Recall \cref{rmrk:binary_constraints_bounded_derivations}, thus all these polynomials can be derived by $\PC$ in polynomial time.

Next we use a construction similar to the one in \cite{BharathiM21, BharathiM25} to define larger constraints called \emph{chain permutation constraints} (CPCs) that combine multiple permutation constraints together. We maintain the notation. More precisely, it is possible to define constraints
\begin{equation*}
    CPC_p := R_p(X_p = \{ x_{p_1}, \ldots, x_{p_r} \})
\end{equation*}
such that the solutions to the constraints $C_{P}$ are also the solutions to the constraints $CPC := \{CPC_p \, | \, p \in J\}$ for some $J \subseteq [n]$. Moreover, the following property holds.

\begin{lemma}[\cite{BharathiM21, BharathiM25}]
    Let $CPC_p = R_p(X_p)$ and $CPC_q = R_q(X_q)$ with $p,q \in J$ be two CPCs. If $p \neq q$, then $X_p \cap X_q = \emptyset$.
\end{lemma}

However, we do not really need to calculate any CPC: it suffices for us to derive a set of polynomials $\mathcal{P}^{CPC_p}$ such that $\Variety{\mathcal{P}^{CPC_p}} = Sol(CPC_p)$ for any $p \in J$. In order to do so, we define

\begin{equation*}
    \mathcal{P}^{CPC_p} := \bigcup_{i,j \in J_p} \mathcal{P}_{ij}
\end{equation*}

where $J_p \subseteq [n]$ is a set of indices such that for any pair of indices there exists a chain of pairs of indices "connecting" them. More precisely, let $i,j \in J_p$, then there exist pairs $$\{ l_1^1, l_2^1 \}, \{ l_1^2, l_2^2 \}, \{l_1^3, l_2^3\}, \ldots, \{l_1^k, l_2^k\} \subseteq J_p$$ for some $k \in [n]$ such that $i \in \{ l_1^1, l_2^1 \}, j \in \{l_1^k, l_2^k\}$ and $|\{l_1^w, l_2^w\} \cap \{l_1^{w+1}, l_2^{w+1}\}| = 1$ for any $w = 1, \ldots, k-1$.

Let $\I_{CPC_p}$ be the combinatorial ideal associated with $CPC_p = R_p(X_p = \{x_{p_1}, \ldots, x_{p_r} \})$. Let $S_i = pr_i(R_p) \subseteq D$ be the $i$-th projection of relation $R_p$, i.e. the set of values $x_i$ can assume for each valid solution in $R_p$. As a result of the construction of the CPCs, there exist bijections between any pair of variables in $X_p$. We denote $\sigma_{ij}: S_i \rightarrow S_j$ any such bijection between two variables $x_i, x_j \in X_p$.

\begin{lemma}
    Let $\mathcal{P}^{CPC_p}$ be a generating set of $\I_{CPC_p}$ for some $p \in J$. If $i,j \in J_p$, then there exist interpolating polynomials $f$ and $g$ simulating the bijections $\sigma_{ij}$ and $\sigma_{ji}$ respectively, i.e. $f(a) = \sigma_{ij}(a)$ for all $a \in D_i$ and similarly for $g$. It follows that $x_j - f(x_i), x_i - g(x_j) \in \langle \mathcal{P}^{CPC_p} \rangle$. Furthermore, $\deg(f), \deg(g) \leq |D| - 1$ and polynomials $x_j - f(x_i)$ and $ x_i - g(x_j)$ can be derived by $\PC$ in polynomial time.
\end{lemma}

\begin{proof}
    Since $i,j \in J_p$, there exist pairs $\{ l_1^1, l_2^1 \}, \{ l_1^2, l_2^2 \}, \{l_1^3, l_2^3\}, \ldots, \{l_1^k, l_2^k\} \subseteq J_p$ for some $k \in [n]$ such that $i \in \{ l_1^1, l_2^1 \}, j \in \{l_1^k, l_2^k\}$ and $|\{l_1^w, l_2^w\} \cap \{l_1^{w+1}, l_2^{w+1}\}| = 1$ for any $w = 1, \ldots, k-1$. We proceed by induction on $k$.

    \textit{Base Case (k=1):} in this case there exists a constraint $C_{ij} = R_{ij}(x_i,x_j) \in CPC_p$ with $R_{ij} = \{(a,\pi_{ij}(a))\}_{a \in D_i}$. The result follows by \cref{rmrk:binary_constraints_bounded_derivations}, where $f$ and $g$ are the Lagrange polynomials simulating permutation $\pi_{ij}$, and hence with degree $\deg(f), \deg(g) \leq |D| - 1$.

    \textit{Induction step:} we assume that the statement holds for some $k \in [n-1]$ and we will prove it for $k+1$. By assumption, there exists the chain of pairs of indices $\{ l_1^1, l_2^1 \}, \{ l_1^2, l_2^2 \},$ $ \ldots,$  $ \{l_1^k, l_2^k\}, \{l_1^{k+1}, l_2^{k+1}\}$ such that $x_i$ is "connected" to $x_j$. Suppose $l_2^{k+1} = j$, then without loss of generality either $l_1^k = l_1^{k+1}$ or $l_2^k = l_1^{k+1}$. By the induction hypothesis, polynomials $x_i - f(x_{l_1^{k+1}})$ and $x_{l_1^{k+1}} - g(x_i)$ are in $\langle \mathcal{P}^{CPC_p} \rangle$, can be $\PC$ derived in polynomial time and $\deg(f), \deg(g) \leq |D| - 1$. Moreover, by the inductive hypothesis there also exist polynomials $f'(x_j)$ and $g'(x_{l_1^{k+1}})$ such that $x_{l_1^{k+1}} - f'(x_j)$ and $x_j - g'(x_{l_1^{k+1}})$ are in $\langle \mathcal{P}^{CPC_p} \rangle$ with $\deg(f'),\deg(g') \leq |D| - 1$ and can be derived within $\PC$. Now consider the polynomial $x_j - g'(g(x_i))$ and also consider polynomial $x_j - \tilde{g}(x_i)$, where $\tilde{g}$ is the Lagrange polynomial simulating $\sigma_{ij}$. It follows that these two polynomials evaluate to 0 on $\{(a,\sigma_{ij}(a))\}_{a \in S_i}$. Therefore $x_j - \tilde{g}(x_i) \in \langle x_j - g'(g(x_i)), \Pi_{a \in S_i}(x_i - a), \Pi_{b \in S_j} (x_j - b) \rangle $ since the generated ideal is radical. Moreover, the derivation can be simulated by $\PC$ in time and degree both independent from $n$, as noted in \cref{rmrk:binary_constraints_bounded_derivations}.
\end{proof}

\begin{remark}
    There are at most $|D|!$ different bijections between variables in $X_p$, implying that many variables are actually related by linear polynomials of the type $x_j - x_k$ for some $x_j, x_k \in X_p$. Indeed, consider variable $x_i \in X_p$ and consider all the bijections $\sigma_{il}$ such that $x_l \in X_p$. Suppose there exist two variables $x_j, x_k \in X_p$ such that $\sigma_{ij} = \sigma_{ik}$. Then $x_j = x_k$ for all values in $S_j = S_k$. It follows that $\sigma_{jk} = \sigma_{kj} = id$. Thus the Lagrange interpolating polynomials $f$ and $g$ are $x_j - x_k$ and $x_k - x_j$ respectively.
\end{remark}

From the above remark it follows that the number of variables which are \textit{not} linearly related is at most $|D|! = O(1)$, while for the remaining variables linear polynomials can be derived from $\mathcal{P}^{CPC_p}$. Finding a \GB basis for $\I_{CPC_p}$ thus reduces to finding the \GB basis of an ideal with at most $|D|!$ variables. Indeed, consider the set
\begin{equation}\label{eqn:CPC_linear_polynomials}
    \{ x_i - x_j \, | \, x_i > x_j \text{ and } \sigma_{ij} \text{ is the identity function } \forall x_i, x_j \in X_p \}.
\end{equation}

Let $\mathcal{S}_p$ be the reduced \GB basis of \cref{eqn:CPC_linear_polynomials}. Note that it can be derived by $\PC$ with bounded degree and bounded coefficients. Now define

\begin{equation*}
    M_p := \{LM(s) \, | \, s \in \mathcal{S}_p\}.
\end{equation*}

Therefore, by radicality it follows

\begin{equation}\label{eqn:dual_discriminator_CPC_decomposition}
    \I_{CPC_p} = \langle \mathcal{T}_p \rangle + \langle \mathcal{D}_p \rangle + \langle \mathcal{S}_p \rangle.
\end{equation}

where $\mathcal{T}_p$ is the set of interpolating polynomials for the bijections different from the identity $\sigma_{ij} \neq id$, and $\mathcal{D}_p = \{ \Pi_{a \in S_i}(x_i - a) \, | \, x_i \in X_p \setminus M_p\}$ is the set of domain polynomials.

At this stage, finding a \GB basis for $\I_{CPC_p}$ would not be difficult, but we choose to defer this step to \cref{sect:combin_gen}. As we will see, we will update the sets $\mathcal{D}_p$ so that it is easy to find a \GB basis for $\mathcal{T}_p \cup \mathcal{D}_p \cup \mathcal{S}_p \cup G$, where $G$ is a set of generators arising from the complete and two-fan constraints.

\subsubsection{Complete and two-fan constraints}\label{sect:complete_two-fan_constraints}

We consider the set of constraints $C_{CF}$ comprising of the complete and two-fan constraints. Let $G = \emptyset$. We will add polynomials to $G$ until it represents the constraints in $C_{CF}$.

A constraint $C_{ij} = R(x_i, x_j)$ is \textit{complete} whenever $R = D_i \times D_j$ with $D_i, D_j \subseteq D$. It is described by a pair of \emph{partial domain polynomials} defined as 
\begin{equation*}
    \Pi_{a \in D_i} (x_i - a), \qquad \Pi_{a \in D_j} (x_j - a).
\end{equation*}
For every complete constraint, we can derive such polynomial as seen in \cref{rmrk:binary_constraints_bounded_derivations} and add them to $G$.

A constraint $C_{ij} = R(x_i, x_j)$ is \textit{two-fan} if $R = \{(\{a\} \times D_j) \cup (D_i \times \{b\})\}$ with $D_i, D_j \subseteq D$, $a \in D_i$ and $b \in D_j$. A two-fan constraint is described by polynomials
\begin{equation*}
    (x_i - a)(x_j - b), \quad \Pi_{c \in D_i} (x_i - c), \quad \Pi_{d \in D_j} (x_j - d).
\end{equation*}
We also add those to $G$. 

It might happen that there exists a variable $x_i$ for which two partial domain polynomials have been added, say $\Pi_{c \in D_{i_1}} (x_i - c)$ and $\Pi_{d \in D_{i_2}} (x_i - d)$. In this case, we derive by Derivation Scheme 2. the polynomial $\Pi_{c \in D_i} (x_i - c)$ where $D_i = D_{i_1} \cap D_{i_2}$ and replace the two initial partial domain polynomials in $G$ with this new one. If for some variable $x_i$ no partial domain polynomial has been added to $G$, we add to $G$ the full domain polynomial $\Pi_{a \in D} (x_i - a)$.

Lastly, we observe that we can consider the equivalent $(2,3)$-consistent version $\Cc' = (X, D, C')$ of the initial $\CSP$ $\Cc = (X, D, C)$. We follow along the algorithm presented in \cite{BulatovRSTOC22}. However, we expand on that result by presenting a $\PC$ simulation of the algorithm.
\begin{itemize}
    \item Repeat until possible: consider three variables $x_i, x_j, x_k \in X$ and consider the set
    \begin{align*}
        T_{ij,k} = \{(a,b) \in R_{ij} \ | \ \nexists c \in D \text{ s.t } (a,c) \in R_{ik} \wedge (c,b) \in R_{kj} \} 
    \end{align*}
    If $T_{ij,k} \neq \emptyset$ do the following. Let $f$ and $g$ be interpolating polynomials vanishing at $R_{ik}$ and $R_{kj}$ respectively, i.e. $f(\alpha, \beta) = 0$ if and only if $(\alpha, \beta) \in R_{ik}$, and similarly we define $g$. Note that $\deg(f), \deg(g) = O(|D|^2)$. Define $h(x_i, x_j, x_k) := f(x_i, x_k)g(x_k, x_j)$. Then, by definition, $h(a,b,c) \neq 0$ if $(\alpha,\beta) \in T_{ij,k}$ and for all $c \in D$. It follows that, as done in \cref{rmrk:binary_constraints_bounded_derivations}, we can derive a bivariate polynomial $\tilde{h}(x_i,x_j)$ such that $\tilde{h}(a,b) = 0$ if and only if $(a,b) \in R_{ij} \setminus T_{ij,k}$. Add $\tilde{h}$ to $G$.
\end{itemize}
When the algorithm stops, we consider the $\CSP$ generated by the polynomials in $G$, i.e. for every $x_i, x_j \in X$ we consider the constraint $C' = \mathbf{V}_{ij}(x_i, x_j)$, where $\mathbf{V}_{ij}$ is the variety generated by the polynomials in $x_i$ and $x_j$. It turns out that $C'$ is the $(2,3)$-consistent version of $C$. Therefore, $C'$ and $C$ have the same solutions and $\nabla$ is a polymorphism of $C'$. Moreover, the following holds.

\begin{lemma}\cite[Lemma~4.1.5]{rafiey_constraint_2022}\label{th:complete_two-fan_grobner_basis}
    Let $G$ be defined as above. Then $G$ is a \GB basis of~$\I_{CF}.$
\end{lemma}



\subsubsection[Combining I(CPCp) and I(CF)]{Combining $\I_{CPC_p}$ and $\I_{CF}$}\label{sect:combin_gen}

Next, we want to combine the generators of $\I_{CPC_p}$ and $\I_{CF}$. For the moment we have
\begin{equation}\label{eqn:combinatorial_ideal_decomposition_1}
    \I_{\Cc} = \sum_{p \in J} \I_{CPC_p} + \I_{CF} = \sum_{p \in J} (\langle \mathcal{T}_p \rangle + \langle \mathcal{D}_p \rangle + \langle \mathcal{S}_p \rangle) + \langle G \rangle.
\end{equation}

The remainder of this section completes the proof of \cref{th:dual_discriminator}.

\begin{proof} [Proof of \cref{th:dual_discriminator}]

Consider a variable in $M_p = \{LM(s) \ | \ s \in \mathcal{S}_p\}$. First we reduce to the case where all the variables in $G$ are in $\bigcup_{p \in J} X_p \setminus M_p$. Let $f \in \mathcal{S}_p$ and $g \in G$. Assume that $LM(f)$ and $LM(g)$ contain the same variable $x_i$. Therefore $f = x_i - x_j$ for some $x_i \in M_p$ and $x_j \in X_p$ and $g = \Pi_{a \in D_i}(x_i - a)$. Let $b \in D_i$. It suffices to consider $(x_i - x_j)\Pi_{a \in D_i \setminus \{ b \}}(x_i - a)$ and $g$ to obtain $g - f = (x_j - b)\Pi_{a \in D_i \setminus \{ b \}}(x_i - a)$. By iterating over $D_i$, we derive $\tilde{g} = \Pi_{a \in D_i}(x_j - a)$. We then remove $g$ from $G$ and add $\tilde{g}$. Similarly, if $g = (x_i - a)(x_k - b)$, we can derive $\tilde{g} = (x_j - a)(x_k -b)$ and substitute it to $g$ in $G$.

Next, let $f \in \mathcal{T}_p \cup \mathcal{D}_p$ and let $g \in G$. Assume that $LM(f)$ and $LM(g)$ contain $x_i \in X_p$.

\textbf{Case 1.} Suppose $g = \Pi_{a \in D_i} (x_i - a)$. If $g \notin \mathcal{D}_p$, then derive $\tilde{g} = \Pi_{a \in S_i \cap D_i} (x_i - a)$ using Derivation Scheme 2. and replace the (partial) domain polynomials of $x_i$ in $\mathcal{D}_p$ and in $G$ with $\tilde{g}$. Moreover, all the variables in $X_p$ are linked by bijections. So we must also derive updates for any variable $x_j \in X_p \setminus \{x_i\}$, that is, we have to add polynomials $\Pi_{b \in \sigma_{ij}(S_i \cap D_i)}(x_j - b)$ to $\mathcal{D}_p$ and $G$. To do so, it suffices to iteratively consider the factors of polynomial $\tilde{g}$, i.e. polynomials $x_i - a$ for some $a \in S_i \cap D_i$, then consider polynomial $x_i - f(x_j) \in \mathcal{T}_p$, and polynomial $\Pi_{b \in S_j} (x_j - b) \in \mathcal{D}_p$. Using Derivation Scheme 5., iterating over $a \in S_i \cap D_i$, we update each factor of $\tilde{g}$ from $x_i - a$ to $x_j - \sigma_{ij}(a)$, thus ending up with $\Pi_{b \in \sigma_{ij}(S_i \cap D_i)}(x_j - b)$ instead of $\tilde{g}$. We add these new partial domain polynomials to $\mathcal{D}_p$ and $G$.

\textbf{Case 2.} Suppose $ g = (x_i - a)(x_j - b)$. If $a \notin S_i$, then from $g$ and $\Pi_{c \in S_i}(x_i - c)$ we can derive $(\sum_{c \in S_i} c - |S_i| a)(x_j - b)$ and add it to $\mathcal{D}_p$ and $G$ in place of the partial domain polynomials corresponding to $x_j$. The derivation follows from the observation that $[(x_i - a)(x_j - b)] - [(x_i - c)(x_j - b)] = (c - a)(x_j - b)$.

\textbf{Case 3.} Suppose $ g = (x_i - a)(x_j - b)$, $a \in S_i$ and $x_j \in X_q$ with $X_p \neq X_q$. Then using Derivation Scheme 5. we add to $G$ all the polynomials $(x_k - \sigma_{ik}(a))(x_l - \sigma_{jl}(b))$, where $x_k \in X_p$, $x_l \in X_q$.

\textbf{Case 4.} Suppose $g = (x_i - a)(x_j - b)$, $a \in S_i$ and $x_j \notin \cup_q X_q$. Then using Derivation Scheme 5. we add to $G$ all the polynomials $(x_k - \sigma_{ik}(a))(x_j - b)$, where $x_k \in X_p$.

We obtain again that
\begin{equation}\label{eqn:combinatorial_ideal_decomposition_2}
    \I_{\Cc} = \sum_{p \in J} (\langle \mathcal{T}_p \rangle + \langle \mathcal{D}_p \rangle + \langle \mathcal{S}_p \rangle) + \langle G \rangle,
\end{equation}
where $\mathcal{D}_p$ and $G$ might have been updated in multiple instances. However, $\mathcal{T}_p \cup \mathcal{D}_p \cup \mathcal{S}_p$ generates the ideal generated by some $CPC_p$ and, similarly, $G$ generates the ideal generated by a set of complete and two-fan constraints $C_{CF}$. We show next how to compute a \GB for $\I_\Cc$.

First, we observe now that $\PC$ can simulate efficiently Buchberger's algorithm to calculate the \GB basis of the ideal $\langle \mathcal{T}_p \rangle + \langle \mathcal{D}_p \rangle + \langle \mathcal{S}_p \rangle$ for any $p \in J$. Indeed, we recall the definitions of $\mathcal{T}_p$, $\mathcal{D}_p$ and $\mathcal{S}_p$ of \cref{sect:permutation_constraints} and observe that the number of variables in $\mathcal{T}_p$ is at most $|D|!$. On the other hand, for any $s \in \mathcal{S}_p$ and $t \in \mathcal{T}_p \cup \mathcal{D}_p$ we have that $LM(s)$ and $LM(t)$ are coprime. Therefore, the reduced \GB  of $\I_{CPC_p}$ can be calculated in polynomial time, thus independent from $n$. We denote the reduced \GB basis of $\I_{CPC_p}$ with $\mathcal{G}_p$.

Second, by \cref{th:complete_two-fan_grobner_basis} we have again that $G$ is a \GB basis for $\I_{CF}$.

Lastly, we have the following lemma.

\begin{lemma}\cite{BharathiM21, BharathiM25}
    The set $\left( \bigcup_p \mathcal{G}_p \right) \cup G$ is a \GB basis for $\I_\Cc$.
\end{lemma}

The result follows from the lemma above.

\end{proof}

\section{Conclusions and research directions}\label{sect:open problem}

In this paper it is shown that for two classes of problems that generalize \textsc{HORN-SAT} and \textsc{2-SAT} a $\PC$ proof of degree $d$ can be found in time $n^{O(d)}$, if it exists (see also \cite{BharathiM21} for related results). This is obtained by first showing that a (truncated) \GB basis for the graded lexicographic order can be computed by $\PC$ in polynomial time for any fixed $d$ (and therefore with polynomial bit complexity). By a simple polynomial division argument (see \cref{sect:PC_bit}), the latter implies that for these two classes there are no bit-complexity issues. Furthermore, both \textsc{HORN-SAT} and \text{2-SAT}, along with their generalizations to finite domains---semilattice and dual-discriminator closed languages, respectively---fit within the framework of bounded width languages \cite{JEAVONS_TRACTABLE_CONSTRAINTS}. As a step towards understanding the boundary of tractability of the $\PC$ criterion, it would be interesting to explore how $\PC$ can be applied to solve the $\IMP_d(\Gamma)$ for bounded width languages. Moreover, results regarding the tractability of the $\IMP_d$, even when using restricted form of algorithms such as those encapsulated in the Polynomial Calculus proof system, would be valuable on their own right.

Similar to SoS, it has often been stated that a $\PC$ refutation of degree $d$ can be found in time $n^{O(d)}$, if it exists. For $\PC$ over finite fields, this is already clear from the algorithm provided in \cite{CleggEI96}. However, in the case of $\PC$ over reals or rationals, the search for proofs can potentially result in bit complexity issues as recently shown by Hakoniemi in \cite{Hakoniemi21}. Indeed, in \cite{Hakoniemi21} it is shown that there is a set of polynomial constraints $Q_n$ on Boolean variables that has both $\sos$ and $\PC$ over rationals refutations of degree 2, but for which any $\sos$ or $\PC$ refutation over rationals must have exponential bit-complexity. The author remarks that the constraints in $Q_n$ do not arise from any CNF, and raise the open question to understand whether the
two measures of bit-complexity and monomial-size are polynomially equivalent for CNFs. Our $\PC$ criterion does not apply to other CNF problems like \textsc{3Lin(2)}, where $\PC$ and $\sos$ are known to be not complete for any fixed $d$. Moreover, we remark that \textsc{3Lin(2)} problems do not arise from bounded width languages \cite{BartoK14}.
As an intermediate step for the open question raised in \cite{Hakoniemi21}, it would be interesting to understand the bit complexity of problems with these CNF constraints.


In this paper, we have made partial advancements in the understanding of the bit complexity of \(\sos\), an issue that has only recently garnered attention and remains in its early stages of research. Since it was first raised 2017, progress has been relatively limited. In this section, we have offered some insights that we hope will stimulate further exploration and enhance our understanding of this fundamental problem.

{
\bibliography{references}

\begin{thebibliography}{10}

\bibitem{BaldiKM24}
Lorenzo Baldi, Teresa Krick, and Bernard Mourrain.
\newblock An effective positivstellensatz over the rational numbers for finite semialgebraic sets, 2024.
\newblock URL: \url{https://arxiv.org/abs/2410.04845}, \href {https://arxiv.org/abs/2410.04845} {\path{arXiv:2410.04845}}.

\bibitem{BartoK14}
Libor Barto and Marcin Kozik.
\newblock Constraint satisfaction problems solvable by local consistency methods.
\newblock {\em J. ACM}, 61(1), January 2014.
\newblock \href {https://doi.org/10.1145/2556646} {\path{doi:10.1145/2556646}}.

\bibitem{barto_et_al:DFU:2017:6959}
Libor Barto, Andrei Krokhin, and Ross Willard.
\newblock {Polymorphisms, and How to Use Them}.
\newblock In Andrei Krokhin and Stanislav Zivny, editors, {\em The Constraint Satisfaction Problem: Complexity and Approximability}, volume~7 of {\em Dagstuhl Follow-Ups}, pages 1--44. Schloss Dagstuhl--Leibniz-Zentrum fuer Informatik, Dagstuhl, Germany, 2017.
\newblock \href {https://doi.org/10.4230/DFU.Vol7.15301.1} {\path{doi:10.4230/DFU.Vol7.15301.1}}.

\bibitem{BeameIKPP94}
Paul Beame, Russell Impagliazzo, Jan Kraj{\'{\i}}cek, Toniann Pitassi, and Pavel Pudl{\'{a}}k.
\newblock Lower bound on {H}ilbert's {N}ullstellensatz and propositional proofs.
\newblock In {\em 35th Annual Symposium on Foundations of Computer Science, Santa Fe, New Mexico, USA, 20-22 November 1994}, pages 794--806, 1994.

\bibitem{berkholz18}
Christoph Berkholz.
\newblock {The Relation between Polynomial Calculus, Sherali-Adams, and Sum-of-Squares Proofs}.
\newblock In Rolf Niedermeier and Brigitte Vall{\'e}e, editors, {\em 35th Symposium on Theoretical Aspects of Computer Science (STACS 2018)}, volume~96 of {\em Leibniz International Proceedings in Informatics (LIPIcs)}, pages 11:1--11:14, Dagstuhl, Germany, 2018. Schloss Dagstuhl--Leibniz-Zentrum fuer Informatik.
\newblock \href {https://doi.org/10.4230/LIPIcs.STACS.2018.11} {\path{doi:10.4230/LIPIcs.STACS.2018.11}}.

\bibitem{BharathiM21}
Arpitha~P. Bharathi and Monaldo Mastrolilli.
\newblock Ideal membership problem for boolean minority and dual discriminator.
\newblock In Filippo Bonchi and Simon~J. Puglisi, editors, {\em 46th International Symposium on Mathematical Foundations of Computer Science, {MFCS} 2021, August 23-27, 2021, Tallinn, Estonia}, volume 202 of {\em LIPIcs}, pages 16:1--16:20. Schloss Dagstuhl - Leibniz-Zentrum f{\"{u}}r Informatik, 2021.
\newblock \href {https://doi.org/10.4230/LIPIcs.MFCS.2021.16} {\path{doi:10.4230/LIPIcs.MFCS.2021.16}}.

\bibitem{BharathiM22}
Arpitha~P. Bharathi and Monaldo Mastrolilli.
\newblock Ideal membership problem over 3-element csps with dual discriminator polymorphism.
\newblock {\em {SIAM} J. Discret. Math.}, 36(3):1800--1822, 2022.
\newblock \href {https://doi.org/10.1137/21M1397131} {\path{doi:10.1137/21M1397131}}.

\bibitem{BharathiM25}
Arpitha~P. Bharathi and Monaldo Mastrolilli.
\newblock Ideal membership problem for boolean minority and dual discriminator.
\newblock {\em SIAM Journal on Discrete Mathematics}, 39(1):206--230, 2025.
\newblock \href {https://doi.org/10.1137/23M1556010} {\path{doi:10.1137/23M1556010}}.

\bibitem{BuchbergerThesis}
B.~Buchberger.
\newblock {\em {Ein Algorithmus zum Auffinden der Basiselemente des Restklassenringes nach einem nulldimensionalen Polynomideal (An Algorithm for Finding the Basis Elements in the Residue Class Ring Modulo a Zero Dimensional Polynomial Ideal)}}.
\newblock PhD thesis, Mathematical Institute, University of Innsbruck, Austria, 1965.
\newblock English translation in J. of Symbolic Computation, Special Issue on Logic, Mathematics, and Computer Science: Interactions. Vol. 41, Number 3-4, Pages 475--511, 2006.

\bibitem{Bulatov23}
Andrei Bulatov.
\newblock Personal communication, 2023.

\bibitem{Bulatov17}
Andrei~A. Bulatov.
\newblock A dichotomy theorem for nonuniform {CSP}s (best paper award).
\newblock In {\em 58th {IEEE} Annual Symposium on Foundations of Computer Science, {FOCS} 2017, Berkeley, CA, USA, October 15-17, 2017}, pages 319--330, 2017.

\bibitem{2017dfu7}
Andrei~A. Bulatov.
\newblock Constraint satisfaction problems: Complexity and algorithms.
\newblock {\em ACM SIGLOG News}, 5(4):4--24, November 2018.
\newblock \href {https://doi.org/10.1145/3292048.3292050} {\path{doi:10.1145/3292048.3292050}}.

\bibitem{BulatovARXIV21}
Andrei~A. Bulatov and Akbar Rafiey.
\newblock On the complexity of csp-based ideal membership problems.
\newblock {\em CoRR}, abs/2011.03700, 2020.
\newblock URL: \url{https://arxiv.org/abs/2011.03700}, \href {https://arxiv.org/abs/2011.03700} {\path{arXiv:2011.03700}}.

\bibitem{BulatovRSTACS22}
Andrei~A. Bulatov and Akbar Rafiey.
\newblock The ideal membership problem and abelian groups.
\newblock In Petra Berenbrink and Benjamin Monmege, editors, {\em 39th International Symposium on Theoretical Aspects of Computer Science, {STACS} 2022, March 15-18, 2022, Marseille, France (Virtual Conference)}, volume 219 of {\em LIPIcs}, pages 18:1--18:16. Schloss Dagstuhl - Leibniz-Zentrum f{\"{u}}r Informatik, 2022.
\newblock \href {https://doi.org/10.4230/LIPIcs.STACS.2022.18} {\path{doi:10.4230/LIPIcs.STACS.2022.18}}.

\bibitem{BulatovRSTOC22}
Andrei~A. Bulatov and Akbar Rafiey.
\newblock On the complexity of csp-based ideal membership problems.
\newblock In Stefano Leonardi and Anupam Gupta, editors, {\em {STOC} '22: 54th Annual {ACM} {SIGACT} Symposium on Theory of Computing, Rome, Italy, June 20 - 24, 2022}, pages 436--449. {ACM}, 2022.
\newblock \href {https://doi.org/10.1145/3519935.3520063} {\path{doi:10.1145/3519935.3520063}}.

\bibitem{Buss96}
Samuel~R Buss.
\newblock Lower bounds on nullstellensatz proofs via designs.
\newblock {\em Proof complexity and feasible arithmetics}, 39:59--71, 1996.

\bibitem{BussP98}
Samuel~R. Buss and Toniann Pitassi.
\newblock Good degree bounds on {N}ullstellensatz refutations of the induction principle.
\newblock {\em J. Comput. Syst. Sci.}, 57(2):162--171, 1998.

\bibitem{Chen09}
Hubie Chen.
\newblock A rendezvous of logic, complexity, and algebra.
\newblock {\em ACM Comput. Surv.}, 42(1):2:1--2:32, December 2009.
\newblock \href {https://doi.org/10.1145/1592451.1592453} {\path{doi:10.1145/1592451.1592453}}.

\bibitem{CleggEI96}
Matthew Clegg, Jeff Edmonds, and Russell Impagliazzo.
\newblock Using the {G}roebner basis algorithm to find proofs of unsatisfiability.
\newblock In {\em Proceedings of the Twenty-Eighth Annual {ACM} Symposium on the Theory of Computing, 1996}, pages 174--183, 1996.

\bibitem{Cooper1994CharacterisingTC}
Martin~C. Cooper, David~A. Cohen, and Peter~G. Jeavons.
\newblock Characterising tractable constraints.
\newblock {\em Artificial Intelligence}, 65(2):347--361, 1994.
\newblock \href {https://doi.org/10.1016/0004-3702(94)90021-3} {\path{doi:10.1016/0004-3702(94)90021-3}}.

\bibitem{Cox}
David~A. Cox, John Little, and Donal O'Shea.
\newblock {\em Ideals, Varieties, and Algorithms: An Introduction to Computational Algebraic Geometry and Commutative Algebra}.
\newblock Springer Publishing Company, Incorporated, 4th edition, 2015.

\bibitem{Davey_Priestley_2002}
B.~A. Davey and H.~A. Priestley.
\newblock {\em Introduction to Lattices and Order}.
\newblock Cambridge University Press, 2 edition, 2002.

\bibitem{FAUGERE1993329}
Jean-Charles Faug{\`e}re, Patrizia~M. Gianni, Daniel Lazard, and Teo Mora.
\newblock {Efficient Computation of Zero-dimensional {G}r{\"o}bner Bases by Change of Ordering}.
\newblock {\em Journal of Symbolic Computation}, 16(4):329 -- 344, 1993.
\newblock \href {https://doi.org/10.1006/jsco.1993.1051} {\path{doi:10.1006/jsco.1993.1051}}.

\bibitem{FlemingKothariPitassi19}
Noah Fleming, Pravesh Kothari, and Toniann Pitassi.
\newblock Semialgebraic proofs and efficient algorithm design.
\newblock {\em Foundations and Trends in Theoretical Computer Science}, 14(1-2):1--221, 2019.
\newblock \href {https://doi.org/10.1561/0400000086} {\path{doi:10.1561/0400000086}}.

\bibitem{Gasca2000-hx}
Mariano Gasca and Thomas Sauer.
\newblock Polynomial interpolation in several variables.
\newblock {\em Advances in Computational Mathematics}, 12(4):377--410, March 2000.

\bibitem{GoemansWilliamson1995}
Michel~X. Goemans and David~P. Williamson.
\newblock Improved approximation algorithms for maximum cut and satisfiability problems using semidefinite programming.
\newblock {\em J. ACM}, 42(6):1115–1145, nov 1995.
\newblock \href {https://doi.org/10.1145/227683.227684} {\path{doi:10.1145/227683.227684}}.

\bibitem{Gribling23}
Sander Gribling, Sven Polak, and Lucas Slot.
\newblock A note on the computational complexity of the moment-sos hierarchy for polynomial optimization.
\newblock In {\em Proceedings of the 2023 International Symposium on Symbolic and Algebraic Computation}, ISSAC '23, page 280–288, New York, NY, USA, 2023. Association for Computing Machinery.
\newblock \href {https://doi.org/10.1145/3597066.3597075} {\path{doi:10.1145/3597066.3597075}}.

\bibitem{GRIGORIEV2001613}
Dima Grigoriev.
\newblock Linear lower bound on degrees of positivstellensatz calculus proofs for the parity.
\newblock {\em Theoretical Computer Science}, 259(1):613--622, 2001.
\newblock \href {https://doi.org/10.1016/S0304-3975(00)00157-2} {\path{doi:10.1016/S0304-3975(00)00157-2}}.

\bibitem{GrigorievHP_MMJ02}
Dima Grigoriev, Edward~A. Hirsch, and Dmitrii~V. Pasechnik.
\newblock Complexity of semialgebraic proofs.
\newblock {\em Moscow Mathematical Journal}, 2(4):647--679, 2002.

\bibitem{GrigorievV01}
Dima Grigoriev and Nicolai N.~Vorobjov Jr.
\newblock Complexity of null-and positivstellensatz proofs.
\newblock {\em Ann. Pure Appl. Log.}, 113(1-3):153--160, 2001.
\newblock \href {https://doi.org/10.1016/S0168-0072(01)00055-0} {\path{doi:10.1016/S0168-0072(01)00055-0}}.

\bibitem{Hakoniemi21}
Tuomas Hakoniemi.
\newblock Monomial size vs. bit-complexity in sums-of-squares and polynomial calculus.
\newblock In {\em 36th Annual {ACM/IEEE} Symposium on Logic in Computer Science, {LICS} 2021, Rome, Italy, June 29 - July 2, 2021}, pages 1--7. {IEEE}, 2021.
\newblock \href {https://doi.org/10.1109/LICS52264.2021.9470545} {\path{doi:10.1109/LICS52264.2021.9470545}}.

\bibitem{Hilbert1893}
David Hilbert.
\newblock Ueber die vollen invariantensysteme.
\newblock {\em Mathematische Annalen}, 42:313--373, 1893.
\newblock URL: \url{http://eudml.org/doc/157652}.

\bibitem{JEAVONS1998185}
Peter Jeavons.
\newblock On the algebraic structure of combinatorial problems.
\newblock {\em Theoretical Computer Science}, 200(1):185 -- 204, 1998.
\newblock \href {https://doi.org/10.1016/S0304-3975(97)00230-2} {\path{doi:10.1016/S0304-3975(97)00230-2}}.

\bibitem{Jeavons:1997:CPC}
Peter Jeavons, David Cohen, and Marc Gyssens.
\newblock Closure properties of constraints.
\newblock {\em J. ACM}, 44(4):527--548, July 1997.
\newblock \href {https://doi.org/10.1145/263867.263489} {\path{doi:10.1145/263867.263489}}.

\bibitem{JEAVONS_TRACTABLE_CONSTRAINTS}
Peter~G. Jeavons and Martin~C. Cooper.
\newblock Tractable constraints on ordered domains.
\newblock {\em Artificial Intelligence}, 79(2):327--339, 1995.
\newblock \href {https://doi.org/10.1016/0004-3702(95)00107-7} {\path{doi:10.1016/0004-3702(95)00107-7}}.

\bibitem{JeffersonJGD13}
Christopher Jefferson, Peter Jeavons, Martin~J. Green, and M.~R.~C. van Dongen.
\newblock Representing and solving finite-domain constraint problems using systems of polynomials.
\newblock {\em Annals of Mathematics and Artificial Intelligence}, 67(3):359--382, Mar 2013.
\newblock \href {https://doi.org/10.1007/s10472-013-9365-7} {\path{doi:10.1007/s10472-013-9365-7}}.

\bibitem{JoszH16}
C{\'e}dric Josz and Didier Henrion.
\newblock Strong duality in lasserre's hierarchy for polynomial optimization.
\newblock {\em Optimization Letters}, 10(1):3--10, January 2016.

\bibitem{KLM-hard-prob-formulation}
Adam Kurpisz, Samuli Lepp{\"{a}}nen, and Monaldo Mastrolilli.
\newblock On the hardest problem formulations for the 0/1 lasserre hierarchy.
\newblock {\em Math. Oper. Res.}, 42(1):135--143, 2017.
\newblock \href {https://doi.org/10.1287/MOOR.2016.0797} {\path{doi:10.1287/MOOR.2016.0797}}.

\bibitem{KLM-unbounded-SOS-hierarchy-integrality-gap}
Adam Kurpisz, Samuli Lepp\"{a}nen, and Monaldo Mastrolilli.
\newblock An unbounded sum-of-squares hierarchy integrality gap for a polynomially solvable problem.
\newblock {\em Math. Program.}, 166(1–2):1–17, November 2017.
\newblock \href {https://doi.org/10.1007/s10107-016-1102-7} {\path{doi:10.1007/s10107-016-1102-7}}.

\bibitem{KLM-tight-SOS-LB-binary-POP}
Adam Kurpisz, Samuli Lepp\"{a}nen, and Monaldo Mastrolilli.
\newblock Tight sum-of-squares lower bounds for binary polynomial optimization problems.
\newblock {\em ACM Trans. Comput. Theory}, 16(1), March 2024.
\newblock \href {https://doi.org/10.1145/3626106} {\path{doi:10.1145/3626106}}.

\bibitem{KLM-SOS-hierarchy-LB-symmetric-formulations}
Adam Kurpisz, Samuli Leppänen, and Monaldo Mastrolilli.
\newblock Sum-of-squares hierarchy lower bounds for symmetric formulations.
\newblock {\em Mathematical Programming}, 182(1-2):369 – 397, 2020.
\newblock Cited by: 4.
\newblock \href {https://doi.org/10.1007/s10107-019-01398-9} {\path{doi:10.1007/s10107-019-01398-9}}.

\bibitem{Lasserre2001}
Jean~B. Lasserre.
\newblock An explicit exact sdp relaxation for nonlinear 0-1 programs.
\newblock In Karen Aardal and Bert Gerards, editors, {\em Integer Programming and Combinatorial Optimization}, pages 293--303, Berlin, Heidelberg, 2001. Springer Berlin Heidelberg.

\bibitem{Laurent2009}
Monique Laurent.
\newblock {\em Sums of Squares, Moment Matrices and Optimization Over Polynomials}, pages 157--270.
\newblock Springer New York, New York, NY, 2009.
\newblock \href {https://doi.org/10.1007/978-0-387-09686-5_7} {\path{doi:10.1007/978-0-387-09686-5_7}}.

\bibitem{magron-schwei}
Victor Magron, Mohab {Safey El Din}, and Markus Schweighofer.
\newblock Algorithms for weighted sum of squares decomposition of non-negative univariate polynomials.
\newblock {\em Journal of Symbolic Computation}, 93:200--220, 2019.
\newblock \href {https://doi.org/10.1016/j.jsc.2018.06.005} {\path{doi:10.1016/j.jsc.2018.06.005}}.

\bibitem{Mastrolilli21TALG}
Monaldo Mastrolilli.
\newblock The complexity of the ideal membership problem for constrained problems over the boolean domain.
\newblock {\em {ACM} Trans. Algorithms}, 17(4):32:1--32:29, 2021.
\newblock \href {https://doi.org/10.1145/3449350} {\path{doi:10.1145/3449350}}.

\bibitem{Mayr1989}
Ernst~W. Mayr.
\newblock Membership in polynomial ideals over q is exponential space complete.
\newblock In B.~Monien and R.~Cori, editors, {\em STACS 89}, pages 400--406, Berlin, Heidelberg, 1989. Springer Berlin Heidelberg.

\bibitem{MAYR1982305}
Ernst~W. Mayr and Albert~R. Meyer.
\newblock The complexity of the word problems for commutative semigroups and polynomial ideals.
\newblock {\em Advances in Mathematics}, 46(3):305--329, 1982.
\newblock \href {https://doi.org/10.1016/0001-8708(82)90048-2} {\path{doi:10.1016/0001-8708(82)90048-2}}.

\bibitem{odonnell2017}
Ryan O'Donnell.
\newblock {SOS Is Not Obviously Automatizable, Even Approximately}.
\newblock In Christos~H. Papadimitriou, editor, {\em 8th Innovations in Theoretical Computer Science Conference (ITCS 2017)}, volume~67 of {\em Leibniz International Proceedings in Informatics (LIPIcs)}, pages 59:1--59:10, Dagstuhl, Germany, 2017. Schloss Dagstuhl--Leibniz-Zentrum fuer Informatik.
\newblock \href {https://doi.org/10.4230/LIPIcs.ITCS.2017.59} {\path{doi:10.4230/LIPIcs.ITCS.2017.59}}.

\bibitem{palomba}
Marilena Palomba, Lucas Slot, Luis~Felipe Vargas, and Monaldo Mastrolilli.
\newblock Computational complexity of sum-of-squares bounds for copositive programs, 2025.
\newblock URL: \url{https://arxiv.org/abs/2501.03698}, \href {https://arxiv.org/abs/2501.03698} {\path{arXiv:2501.03698}}.

\bibitem{DonaPapert1964}
Dona Papert.
\newblock {Congruence Relations in Semi-Lattices}.
\newblock {\em Journal of the London Mathematical Society}, s1-39(1):723--729, 01 1964.
\newblock \href {https://arxiv.org/abs/https://academic.oup.com/jlms/article-pdf/s1-39/1/723/2721805/s1-39-1-723.pdf} {\path{arXiv:https://academic.oup.com/jlms/article-pdf/s1-39/1/723/2721805/s1-39-1-723.pdf}}, \href {https://doi.org/10.1112/jlms/s1-39.1.723} {\path{doi:10.1112/jlms/s1-39.1.723}}.

\bibitem{Parrilo03}
Pablo~A. Parrilo.
\newblock Semidefinite programming relaxations for semialgebraic problems.
\newblock {\em Mathematical Programming}, 96(2):293--320, 2003.

\bibitem{PartFTT21}
Fedor Part, Neil Thapen, and Iddo Tzameret.
\newblock First-order reasoning and efficient semi-algebraic proofs.
\newblock In {\em 2021 36th Annual ACM/IEEE Symposium on Logic in Computer Science (LICS)}, pages 1--13, 2021.
\newblock \href {https://doi.org/10.1109/LICS52264.2021.9470546} {\path{doi:10.1109/LICS52264.2021.9470546}}.

\bibitem{PatakiT24}
G\'{a}bor Pataki and Aleksandr Touzov.
\newblock How do exponential size solutions arise in semidefinite programming?
\newblock {\em SIAM Journal on Optimization}, 34(1):977--1005, 2024.

\bibitem{phillips2003interpolation}
G.M. Phillips.
\newblock {\em Interpolation and Approximation by Polynomials}.
\newblock CMS Books in Mathematics. Springer, 2003.

\bibitem{potechin:SOS-LB-from-symmetry}
Aaron Potechin.
\newblock {Sum of Squares Lower Bounds from Symmetry and a Good Story}.
\newblock In Avrim Blum, editor, {\em 10th Innovations in Theoretical Computer Science Conference (ITCS 2019)}, volume 124 of {\em Leibniz International Proceedings in Informatics (LIPIcs)}, pages 61:1--61:20, Dagstuhl, Germany, 2019. Schloss Dagstuhl -- Leibniz-Zentrum f{\"u}r Informatik.
\newblock \href {https://doi.org/10.4230/LIPIcs.ITCS.2019.61} {\path{doi:10.4230/LIPIcs.ITCS.2019.61}}.

\bibitem{rafiey_constraint_2022}
Akbar Rafiey.
\newblock {\em Constraint {Satisfaction} {Problems} and friends: {Symmetries} and algorithm design}.
\newblock PhD thesis, August 2022.
\newblock Publisher: Simon Fraser University.
\newblock URL: \url{https://summit.sfu.ca/item/35622}.

\bibitem{raghavendra_weitz2017}
Prasad Raghavendra and Benjamin Weitz.
\newblock {On the Bit Complexity of Sum-of-Squares Proofs}.
\newblock In Ioannis Chatzigiannakis, Piotr Indyk, Fabian Kuhn, and Anca Muscholl, editors, {\em 44th International Colloquium on Automata, Languages, and Programming (ICALP 2017)}, volume~80 of {\em Leibniz International Proceedings in Informatics (LIPIcs)}, pages 80:1--80:13, Dagstuhl, Germany, 2017. Schloss Dagstuhl--Leibniz-Zentrum fuer Informatik.
\newblock \href {https://doi.org/10.4230/LIPIcs.ICALP.2017.80} {\path{doi:10.4230/LIPIcs.ICALP.2017.80}}.

\bibitem{Schaefer78}
Thomas~J. Schaefer.
\newblock The complexity of satisfiability problems.
\newblock In {\em Proceedings of the Tenth Annual ACM Symposium on Theory of Computing}, STOC '78, pages 216--226, New York, NY, USA, 1978. ACM.
\newblock \href {https://doi.org/10.1145/800133.804350} {\path{doi:10.1145/800133.804350}}.

\bibitem{Sokolov20}
Dmitry Sokolov.
\newblock ({S}emi)algebraic proofs over {\(\pm\)}1 variables.
\newblock In Konstantin Makarychev, Yury Makarychev, Madhur Tulsiani, Gautam Kamath, and Julia Chuzhoy, editors, {\em Proccedings of the 52nd Annual {ACM} {SIGACT} Symposium on Theory of Computing, {STOC} 2020, Chicago, IL, USA, June 22-26, 2020}, pages 78--90. {ACM}, 2020.
\newblock \href {https://doi.org/10.1145/3357713.3384288} {\path{doi:10.1145/3357713.3384288}}.

\bibitem{ThapperZ18}
Johan Thapper and Stanislav {\v{Z}}ivn{\`y}.
\newblock The limits of sdp relaxations for general-valued csps.
\newblock {\em ACM Transactions on Computation Theory (TOCT)}, 10(3):1--22, 2018.

\bibitem{vandongenPhd}
Marc~R.C. van Dongen.
\newblock {\em Constraints, Varieties, and Algorithms}.
\newblock PhD thesis, Department of Computer Science, University College, Cork, Ireland, 2002.
\newblock URL: \url{http://csweb.ucc.ie/~dongen/papers/UCC/02/thesis.pdf}.

\bibitem{Weitz:Phd}
Benjamin Weitz.
\newblock {\em Polynomial Proof Systems, Effective Derivations, and their Applications in the Sum-of-Squares Hierarchy}.
\newblock PhD thesis, EECS Department, University of California, Berkeley, May 2017.

\bibitem{Zhuk17}
Dmitriy Zhuk.
\newblock A proof of {CSP} dichotomy conjecture (best paper award).
\newblock In {\em 58th {IEEE} Annual Symposium on Foundations of Computer Science, {FOCS} 2017, Berkeley, CA, USA, October 15-17, 2017}, pages 331--342, 2017.
\newblock \href {https://doi.org/10.1109/FOCS.2017.38} {\path{doi:10.1109/FOCS.2017.38}}.

\bibitem{Zhuk20}
Dmitriy Zhuk.
\newblock A proof of the {CSP} dichotomy conjecture.
\newblock {\em J. {ACM}}, 67(5):30:1--30:78, 2020.
\newblock \href {https://doi.org/10.1145/3402029} {\path{doi:10.1145/3402029}}.

\end{thebibliography}
}

\newpage

\appendix
\section{Refutation degree for Horn clauses}\label{sect:ref_deg}
Consider the case all clauses are duals of Horn clauses (for simplicity, Horn clauses work out identically), namely at most one variable is negated per clause. 
We encode these clauses as a set of polynomial identities in a way that preserves their semantics over $\{0,1\}^n$ assignments.
Namely, let $\Cc=C_1\wedge C_2 \wedge \ldots \wedge C_m$ be a dual Horn clause formula. We encode each clause $C_i=\neg x_{i_1}\lor x_{i_2}\lor \ldots \lor x_{i_k}$ by introducing a polynomial identity $P_i:x_{i_1}(x_{i_2}-1)\cdots (x_{i_k}-1)=0$. The set of $\{0,1\}$ assignments that satisfy the newly introduced set of polynomial identities is exactly the set of satisfying assignments to $\Cc$.

The refutation by \PC\ works as follows. Take all the variables that are already known to be false, say set $F$. These variables belong to $\I_\Cc$. Consider clause $C_i$ and the corresponding polynomial identity $P_i:x_{i_1}(x_{i_2}-1)\cdots (x_{i_k}-1)=0$. If $\{x_{i_2},\ldots, x_{i_k}\}\subseteq F$ then it is easy show that $x_{i_1}\in \I_\Cc$ since by using the aforementioned polynomial identity $P_i$ we can express $x_{i_1}$ as polynomial combination of the variables in $\{x_{i_2},\ldots, x_{i_k}\}$.\footnote{For example if $P:x_{1}(x_{2}-1)(x_{3}-1)=0$ and $x_2,x_3\in\I_\Cc$ then by the polynomial identity $P$ we have that $x_1=x_{1}x_{2}(1-x_{3})+x_{1}x_{3}$, implying that $x_1\in\I_\Cc$.}
So we have added a new variable to the set $F$ of known false variables. If the set of $F$ covers an entire clause with no negated variables, then we can derive that $1\in \I_\Cc$ and we are done. If at some point, neither is true, by setting all remaining variables to 1 we satisfy all the clauses. So if the Horn clauses were unsatisfiable we find a proof whose degree is at most the degree of the polynomial identities encoding the clauses. 

\section{Complexity of Polynomial Division}\label{sect:PC_bit}
Consider the polynomial ring $\mathbb{R}[x_1, \dots, x_n]$ ordered according to the $\grlexns$ order, with $x_1 > x_2 > \dots > x_n$. We will study the complexity of the standard division algorithm for multivariate polynomials (see \cite[Section 2]{Cox}). In particular, we observe next that it is a polynomial time algorithm.

\begin{lemma}\label{th:complexity_polynomial_division}
    Let $\mathcal{P} = \{p_1, \dots, p_m\}$ be a set of polynomials in $\mathbb{R}[x_1, \dots, x_n]$ and consider a polynomial $f \in \mathbb{R}[x_1, \dots, x_n]$. Assume that $f, p_1, \dots, p_m$ have degree at most $d$ and bit complexity polynomial in $n$. Then $f$ can be written as
    \begin{equation*}
        f = h_1 p_1 + \dots h_m p_m + r,
    \end{equation*}
    with $r, h_1, \dots, h_m$ having bit complexity polynomial in $n$.
\end{lemma}

\begin{proof}
    We will refer to $h_1, \dots, h_m$ as the \emph{quotients} and to $r$ as the \emph{remainder}.
    
    We start by observing that the algorithm runs at most in $n^{O(d)}$ iterations. Indeed, at every iteration the polynomial $f$ is divided by a polynomial from $\mathcal{P}$. Thus, its remainder either has smaller degree or the corresponding leading term is smaller with respect to the $lex$ order. Then $f$ is updated to be the division's remainder. It follows that the number of polynomial divisions is bounded by $n^{O(d)}$.

    Next we argue that the bit complexity of the remainder and of the quotients is polynomial in $n$. Let $b$ be the largest number of bits to encode (the coefficients of) a polynomial $f, p_1, \dots, p_m$. Recall that, by assumption, $b = n^{O(d)}$. At every iteration of the algorithm, the bit complexity of the quotients and of the remainder is increased by $O(b)$ bits since a polynomial (quotient or remainder) is updated by summing a term of bit complexity $O(b)$.
\end{proof}

From this lemma it follows immediately that the existence of a "small" \GB basis implies that the \IMP\ can be solved efficiently.

\begin{corollary}
    Let $\mathcal{G}_{2d} = \{g_1, \dots, g_s\}$ be a $2d$-truncated \GB basis of the polynomial ideal $\I \subseteq \mathbb{R}[x_1, \dots, x_n]$. Consider a polynomial $r$ of degree at most $2d$. If the polynomials $r,g_1, \dots, g_s$ have bit complexity polynomial in $n$, then the (search version) of the $\IMP_{2d}$ for $r$ can be solved in time polynomial in $n$.
\end{corollary}

\begin{proof}
    We have that 
    \begin{align*}
        r \in I \iff r|_{\mathcal{G}_{2d}} = 0.
    \end{align*}
    The result follows from \cref{th:complexity_polynomial_division}.
\end{proof}

\end{document}